\documentclass[12pt]{article}

\usepackage{amsmath,amssymb,amsthm,cite,enumitem,mathtools,xcolor}
\usepackage[margin=3cm]{geometry}
\usepackage{setspace}

 \usepackage[utf8]{inputenc} 
\usepackage[T1]{fontenc}
\usepackage{newtxtext} 
\usepackage{bm} 
\usepackage[normalem]{ulem}  

\usepackage{multicol}
\usepackage{authblk}
\newcommand\Ker{\mathcal{N}}
\newcommand\Ran{\mathcal{R}}

\title{Propagators in curved spacetimes \\ from operator theory}

\author[1]{Jan Dereziński}
\author[1,2]{Christian Ga\ss}

\affil[1]{\small Department of Mathematical Methods in Physics, Faculty of Physics, \protect\\
University of Warsaw, Pasteura 5, 02-093 Warszawa, Poland, \protect\\
email: jan.derezinski@fuw.edu.pl, christian.gass@fuw.edu.pl}

\affil[2]{\small Department of Physics, University of Vienna, \protect\\
Boltzmanngasse 5, A-1090 Vienna, Austria}

\date{\today}

\newcommand\RR{\mathbb{R}}
\newcommand\CC{\mathbb{C}}
\DeclarePairedDelimiter{\abs}{\lvert}{\rvert}
\DeclarePairedDelimiter{\norm}{\lVert}{\rVert}

\DeclarePairedDelimiterX{\cinner}[2]{(}{)}{#1\mkern2mu\delimsize\vert\mkern2mu\mathopen{}#2}
\DeclarePairedDelimiterX{\rinner}[2]{\langle}{\rangle}{#1\mkern2mu\delimsize\vert\mkern2mu\mathopen{}#2}

\DeclareMathOperator{\sgn}{sgn}     
\DeclareMathOperator{\supp}{supp}   
\DeclareMathOperator{\T}{T}         

\DeclareMathOperator{\realpart}{Re}
\DeclareMathOperator{\imaginarypart}{Im}
\renewcommand{\Re}{\realpart}
\renewcommand{\Im}{\imaginarypart}

\newcommand{\lrnb}{\overset{\leftrightarrow}{\nabla}}
\newcommand{\KG}{\mathrm{KG}}
\renewcommand\sc{\mathrm{sc}}
\newcommand\sym{\mathrm{sym}}
\newcommand\as{\mathrm{as}}
\newcommand\op{\mathrm{op}}
\newcommand{\dd}{\textup{d}}
\newcommand{\ii}{\textup{i}}
\newcommand{\ee}{\textup{e}}

\newcommand{\dS}{\textup{dS}}

\newcommand{\PJ}{\textup{PJ}}
\newcommand{\AdS}{\textup{AdS}}
\newcommand{\tAdS}{\widetilde{\textup{AdS}}}

\DeclareMathOperator{\re}{Re}

\newcommand{\R}{\textup{R}}

\newcommand{\TT}{\textup{T}}
\newcommand{\F}{\textup{F}}

\newcommand\Span{\mathrm{Span}}
\newcommand{\eps}{\varepsilon}      

\newcommand{\bC}{\mathbb{C}}        
\newcommand{\bN}{\mathbb{N}}        
\newcommand{\bR}{{\mathbb{R}}}      
\newcommand{\bS}{\mathbb{S}}        
\newcommand{\bH}{\mathbb{H}}        
\newcommand{\bZ}{\mathbb{Z}}        
\newcommand{\s}{\textup{s}}
\newcommand{\h}{\textup{h}}

\newcommand{\sW}{\mathcal{W}}       
\newcommand{\sWsc}{\mathcal{W}_{\rm sc}} 

\newcommand{\del}{\partial}         

\newcommand{\aside}[1]{} 

\newcommand{\word}[1]{\quad\text{#1}\quad} 

\def\wick:#1:{\,\mathopen:#1\mathclose:\,} 

\def\epsi^#1_#2{\eps^{#1}{}_{\!#2}} 
\def\epsii_#1^#2{\eps_{#1}{}^{\!#2}} 
\def\lLa^#1_#2{\Lambda^{#1}{}_{#2}}  

\def\duo<#1,#2>{\langle#1,#2\rangle} 
\def\scal<#1|#2>{\langle#1\mathbin|#2\rangle} 

\newcommand{\bes}{\begin{subequations}}
\newcommand{\ees}{\end{subequations}}

\theoremstyle{plain}
\newtheorem{thm}{Theorem}           
\numberwithin{thm}{section}         
\newtheorem{prop}[thm]{Proposition} 
\newtheorem{lem}[thm]{Lemma}       
\newtheorem{con}[thm]{Conjecture}   

\theoremstyle{definition}
\newtheorem{defn}[thm]{Definition}  
\newtheorem{remark}[thm]{Remark}      
\newtheorem{definition}[thm]{Definition}      

\def\bbbone{{\mathchoice {\rm 1\mskip-4mu l} {\rm 1\mskip-4mu l}
{\rm 1\mskip-4.5mu l} {\rm 1\mskip-5mu l}}}
\def\one{\bbbone}

\renewcommand\bar{\overline}
\def\cZ{{\mathcal Z}}

\def\cD{{\mathcal D}}

\def\cW{{\mathcal W}}
\def\cR{{\mathcal R}}
\def\cN{{\mathcal N}}
\newcommand\e{\mathrm{e}}
\newcommand{\beq}{\begin{equation}}
\newcommand{\eeq}{\end{equation}}

\usepackage{graphicx}
\usepackage{floatrow}
\usepackage{tikz,tikz-3dplot}
\usetikzlibrary{calc,arrows,shapes,positioning,decorations,3d,patterns}
\tdplotsetmaincoords{80}{45}
\tdplotsetrotatedcoords{-90}{180}{-90}

\tikzset{surface/.style={draw=blue!70!black, fill=blue!40!white, fill opacity=.25}}

\makeatletter
\renewcommand{\section}{\@startsection{section}{1}{\z@}%
                       {-3.5ex \@plus -1ex \@minus -.2ex}%
                       {2.3ex \@plus.2ex}%
                       {\normalfont\large\bfseries}}
\renewcommand{\subsection}{\@startsection{subsection}{2}{\z@}%
                       {-3.25ex \@plus -1ex \@minus -.2ex}%
                       {1.5ex \@plus .2ex}%
                       {\normalfont\normalsize\bfseries}}
\makeatother

\usepackage{parskip}

\numberwithin{equation}{section}

\usepackage{hyperref}
\usepackage{xcolor}
\hypersetup{
    colorlinks,
    linkcolor={blue!30!black},
    citecolor={blue!30!black},
    urlcolor={blue!30!black}
}

\begin{document}

\maketitle
\begin{abstract} \noindent
We discuss two distinct operator-theoretic settings useful for
describing (or defining) propagators associated with a scalar
Klein-Gordon field on a Lorentzian manifold $M$. Typically, we
assume that $M$ is globally hyperbolic. The term {\em propagator}
here refers to any Green function or bisolution of the Klein-Gordon
equation pertinent to  Quantum Field Theory.

\vspace{4pt} \noindent
The {\em off-shell} setting is based on the Hilbert space $L^2(M)$.
It leads to the definition of the operator-theoretic Feynman and
anti-Feynman propagators, which often coincide with the so-called
in-out Feynman and out-in anti-Feynman propagator.
On some special spacetimes, the sum of the operator-theoretic Feynman
and anti-Feynman propagator equals the sum of the forward and backward
propagator. This is always true on static stable spacetimes and,
curiously, in some other cases as well.

\vspace{4pt} \noindent
The {\em on-shell} setting is  based on the Krein space $\sW_\KG$ of
solutions of the Klein-Gordon equation. It allows us to define 2-point
functions associated to two, possibly distinct, Fock states as the
Klein-Gordon kernels of projectors onto maximal uniformly positive
subspaces of  $\sW_\KG$.

\vspace{4pt} \noindent
After a general discussion, we review a number of examples. We start
with static and asymptotically static spacetimes, which are especially
well-suited for Quantum Field Theory. Then we discuss FLRW spacetimes,
reducible by a mode decomposition to 1-dimensional Schr\"odinger
operators. We compare various approaches to de Sitter space where,
curiously, the off-shell approach  gives non-physical propagators.
Finally, we discuss the universal cover of anti-de Sitter
spaces, where the on-shell approach may require boundary conditions,
unlike the off-shell approach.
\end{abstract}

{\footnotesize
\tableofcontents
}
\section{Introduction}
\label{sec:intro}

\subsection{Propagators and states}

Let $M$ be  a Lorentzian manifold  of dimension $d$ with a
{\em pseudometric tensor} $g_{\mu\nu}$. Let $Y(x)$ be a real-valued
{\em   scalar potential}, e.g. $Y(x)=m^2$. Consider a field on $M$
satisfying the {\em Klein-Gordon equation }
\begin{align}
\big(-\Box  +Y(x)\big)\phi(x)=0,
\label{kg1}
\end{align}
where $\Box:=  |g|^{-\frac12}
\partial_\mu|g|^{\frac12} g^{\mu\nu}\partial_\nu$ is the {\em d'Alembertian}.
If one wants to compute various pertinent quantities related to $\phi$,
and especially to its quantization $\hat\phi$,  one needs to know several
distributions on $M\times M$, often called ``propagators'' or
``two-point functions''.

These  distributions  fall into two categories:
{\em Green functions}  (also called {\em fundamental solutions}),
and {\em bisolutions} of the Klein-Gordon equation.
A {\em Green function} of the Klein-Gordon
equation is a distribution $G^\bullet$ on $M\times M$ satisfying
\begin{align}
  (-\square_{x}+Y(x)) G^\bullet(x,y)
  = \delta(x,y)= (-\square_{y}+Y(y)) G^\bullet(x,y),
\end{align}
where $\delta(x,y)$ denotes the distributional kernel of the identity.
A {\em bisolution} of the Klein-Gordon equation is a distribution
$G^\bullet$ on $M\times M$ satisfying
\begin{align}
  (-\square_{x}+Y(x)) G^\bullet(x,y)
  = 0= (-\square_{y}+Y(y)) G^\bullet(x,y).
\end{align}
In our paper we will colloquially use the term  ``propagator''\footnote{
This nomenclature is in accordance with the previous papers \cite{DS18,DS19,DS22}.
Note, however, that the term ``propagator'' is often reserved  only for
some of these distributions. Following the usage common in physics
we will often also use the term ``two-point function'' for
 (anti-)Wightman bisolutions.}
for various  distinguished Green functions and bisolutions of
\eqref{kg1} motivated by  QFT: the advanced and retarded propagators,
the Pauli-Jordan propagator,  Feynman and anti-Feynman propagators,
 and  Wightman and anti-Wightman two-point functions (in some
  situations also called
positive/negative frequency bisolutions). Wightman functions serve as
two-point functions of a quantum state. Hence,  abusing somewhat terminology, they
are often simply called {\em states}.

In most of our paper we assume that $M$ is globally hyperbolic.
Then one can show the existence of the \emph{retarded} (or \emph{forward})  and \emph{advanced}
 (or \emph{backward}) propagator
$G^\vee(x,x')$ and $G^\wedge(x,x')$, which are the unique Green functions
supported for $x$ in the causal future resp. causal past
of $x'$.  The bisolution  defined by
\begin{align}
G^\PJ(x,x')=G^\vee(x,x')-G^\wedge(x,x')\label{paulijordan}
\end{align}
is usually called  the \emph{Pauli-Jordan
propagator} or the \emph{commutator function}.
It also possesses a causal support. All three propagators $G^\vee$,
$G^\wedge$ and  $G^\PJ$ are useful in the Cauchy problem of the
Klein-Gordon equation.
The classical field $\phi(x)$ satisfying \eqref{kg1}
is equipped with the Poisson bracket
\[\{\phi(x),\phi(y)\}=-G^\PJ(x,y).\]
Therefore, following \cite{DS18,DS19,DS22}, $G^\vee$,
$G^\wedge$ and  $G^\PJ$ will be called {\em classical propagators}.
 In Quantum Field Theory one uses a few other propagators, whose operator-theoretic meaning -- especially on curved
spacetimes -- is the main subject of this article.

 Quantization of the classical field $\phi(x)$ is performed in two
 steps. In the first step we replace it by
 an operator valued distribution $\hat\phi(x)$, which beside the
 Klein-Gordon equation
\begin{align}
\big(-\Box  +Y(x)\big)\hat\phi(x)=0
\label{kg1.}
\end{align}
satisfies
 the so
 called {\em Peierls relation}
 \[[\hat\phi(x),\hat\phi(y)]=-\ii G^\PJ(x,y)\one.\]
 The fields $\hat\phi(x)$ generate a $*$-algebra.

In the second step one selects a representation of the fields in a
Hilbert space.  In practice, this is done by choosing  a state $\omega_\alpha$
on this $*$-algebra, that is, a positive and normalized linear functional.
Then $\omega_\alpha$ defines the GNS Hilbert space
with a distinguished  vector $\Omega_\alpha$.
One usually considers a Fock state (a pure quasifree state), where
the GNS representation has the form of a bosonic Fock space and
$\Omega_\alpha$ is its vacuum.
The expectation values in this state define four important two-point
functions:
\begin{align}
 G_{\alpha}^{(+)}(x,y) &:=
                         \big(\Omega_\alpha|\hat\phi(x)\hat\phi(y)\Omega_\alpha\big)
,\label{G1}\\
G_{\alpha}^{(-)}(x,y) &:= \big(\Omega_\alpha|\hat\phi(y)\hat\phi(x)\Omega_\alpha\big),
\label{G2}\\
\label{eq:GFa}
 G_\alpha^\F(x,y) &:= \ii
 \big(\Omega_\alpha|\T\big(\hat\phi(x)\hat\phi(y)\big)\Omega_\alpha\big),
                                \\ \label{eq:GFb}
  G_\alpha^{\overline{\F}}(x,y) &:= -\ii
\big(\Omega_\alpha| \overline\T\big(\hat\phi(x)\hat\phi(y)\big)\Omega_\alpha\big).
 \end{align}Here, $\T$ and $ \overline\T$ denote the chronological,
 resp. anti-chronological time ordering.
Note that $G_\alpha^{(+)}$ and  $G_\alpha^{(-)}$ are automatically
bisolutions; $G_\alpha^\F$ and  $G_\alpha^{ \overline\F}$ are Green functions.

It is perhaps less known that it is useful  to define mixed
propagators corresponding to two {\em different} states. Suppose
that they are given by vectors
$\Omega_\alpha$ and $\Omega_\beta$, belonging to the same
representation space, with nonzero $(\Omega_\alpha|\Omega_\beta)$.
Then we set
\begin{align}
\label{eq:zero_divided_by_zero}
 G_{\alpha,\beta}^{(+)}(x,y) &:= \frac{\big(\Omega_\alpha|\hat\phi(x)\hat\phi(y)\Omega_\beta\big)}{\big(\Omega_\alpha|\Omega_\beta\big)}
,\\\label{eq:GF0}
G_{\alpha,\beta}^{(-)}(x,y) &:=\frac{ \big(\Omega_\alpha|\hat\phi(y)\hat\phi(x)\Omega_\beta\big)}{\big(\Omega_\alpha|\Omega_\beta\big)},
\\
\label{eq:GFa.2}
  G_{\alpha,\beta}^\F(x,y) &:=\ii
\frac{ \big(\Omega_\alpha|\T\big(\hat\phi(y)\hat\phi(x)\big)\Omega_\beta\big)}{\big(\Omega_\alpha|\Omega_\beta\big)},
         \\ \label{eq:GFb.2}
  G_{\alpha,\beta}^{\overline{\F}}(x,y) &:= -
\ii
\frac{ \big(\Omega_\alpha| \overline \T\big(\hat\phi(y)\hat\phi(x)\big)\Omega_\beta\big)}{\big(\Omega_\alpha|\Omega_\beta\big)}.
\end{align}
Again, $G_{\alpha,\beta}^{(+)}$ and  $G_{\alpha,\beta}^{(-)}$ are
bisolutions; the Feynman propagator $G_{\alpha,\beta}^\F$ and the
anti-Feynman propagator $G_{\alpha,\beta}^{ \overline\F}$ are Green
functions.

We have a minor terminological problem: should
$G_{\alpha,\beta}^\F$ be called the ``$\alpha-\beta$ Feynman
propagator'' or the
``$\beta-\alpha$ Feynman propagator''? The latter choice is consistent
with the ``time arrow'': in typical applications the vacuum
$\Omega_\beta$ is first, and $\Omega_\alpha$ is later. This order is
used e.g. in \cite{FSS13} (see e.g. equation (74)).  In
  symbols we will use the former
order, in  names we will use
the latter order. So $G_{\alpha,\beta}^\F$ will be called the
{\em $\beta-\alpha$ Feynman propagator}.

The functions
$ G_{\alpha}^{(+)}(x,y) $ are used to define the GNS representation
for the state $\omega_\alpha$ and Wick-ordered
product of fields.  Wick ordering is a first step to
renormalization, which is needed to define higher order monomials of
fields.  The renormalization
procedure will not work for an arbitrary state. In practice one
assumes that it has the so-called {\em Hadamard property}, and then
renormalization  works well.
Note that this analysis can be performed on a local level,
without considering the whole spacetime.

Let us now describe the application of Feynman propagators.
Suppose we perturb the dynamics and we want to compute the
{\em scattering operator} $S_\alpha$ in the representation given by
$\Omega_\alpha$. By a standard argument going back to Dyson, often
called the {\em Wick
Theorem},  $S_\alpha$   can be expressed as a perturbation series
with terms labelled by Feynman diagrams. In order to evaluate Feynman
diagrams one needs to
replace the lines by
$G_\alpha^\F(x,y)$.

Often it is natural to compute the {\em renormalized scattering operator}
$S_{\alpha,\beta}$, acting from the representation
generated by $\Omega_\beta$ to the representation generated by
$\Omega_\alpha$.
  Actually, it is then useful to divide the scattering
operator by the overlap between the vacua, and compute
\beq
\tilde S_{\alpha,\beta}:=\frac{S_{\alpha,\beta}}{\big(\Omega_\alpha|\Omega_\beta\big)}.\eeq The
  algorithm is similar as above, except that we put
  $G_{\alpha,\beta}^\F$  at each line of a Feynman diagram.

We will see that $G_{\alpha,\beta}^\F$ can usually be defined even
if $(\Omega_\beta|\Omega_\alpha)=0$.  Therefore, we can then also
compute $\tilde S_{\alpha,\beta}$. In fact, if the theory is linear,
$\tilde S_{\alpha,\beta}$ will be usually a well-defined unbounded
quadratic form, whose integral kernel
$\tilde S_{\alpha \beta }(k_\alpha ,k_\beta )$ can be called the
``renormalized scattering amplitude''. Obviously, the unitarity of
$\tilde S_{\alpha \beta }$ is lost, hence renormalized scattering
amplitudes will not have a direct probabilistic interpretation.
However  their ratios
\begin{align}
\frac{\tilde S_{\alpha \beta }(k_\alpha ,k_\beta )}{
    \tilde S_{\alpha \beta }(k_\alpha ',k_\beta ')}
\end{align}
have a meaning: they can be used to compute {\em branching ratios} of
various processes.

If we want to compute 
$\frac{S_{\alpha,\beta}^*}{(\Omega_\beta|\Omega_\alpha)}$ we
  proceed similarly, except that Feynman propagators need to be
  replaced by anti-Feynman propagators $  G_{\beta,\alpha}^{ \overline\F}$.

One of the important problems of QFT on curved spacetimes is the choice of
a state. In Minkowski space and with $Y(x)=m^2\geq0$ there is a natural
state, described in all textbooks on QFT.
More generally, every   stationary and stable Klein-Gordon equation
possesses a natural state. Stationarity means that one can identify $M$ with
$\mathbb{R}\times\Sigma$ so that $g^{\mu\nu}$ and
$Y$ are independent of $t\in\mathbb{R}$, $\Sigma$ is spacelike and
$\partial_t$ is timelike.
Stability means that the corresponding classical
 Hamiltonian is bounded from below.
Again,  requiring that the state is invariant  under the
time evolution,
and in the GNS representation the dynamics is
implemented by a positive quantum Hamiltonian fixes the state uniquely.
The one-particle
Hilbert space is then  taken to be the {\em positive frequency space}, that is, the spectral
subspace of the generator of the evolution corresponding to the positive
part of the spectrum.

On generic spacetimes there are no distinguished states.
There is however one class of spacetimes, particularly well
adapted to QFT, where there are {\em two} distinguished states.
These are spacetimes with asymptotically stationary and
stable future and past. Such spacetimes possess two distinguished
states: the {\em in-state}  and the {\em out-state}, given by  vectors
$\Omega_-$ and $\Omega_+$. Obviously, they define two pairs of two-point
functions
\begin{align}
 G_{\pm}^{(+)}(x,x')
 &= \big(\Omega_{\pm}|\hat\phi(x)\hat\phi(x')\Omega_{\pm}\big),\\
   G_{\pm}^{(-)}(x,x')
 &= \big(\Omega_{\pm}|\hat\phi(x')\hat\phi(x)\Omega_{\pm}\big).
\end{align}
One can use them to define two GNS representations acting on two Fock
spaces.

More interesting are however the following mixed Feynman propagators: the
\emph{in-out Feynman propagator}
  $G_{+-}^\F$ and the
\emph{out-in  anti-Feynman propagator} $G_{-+}^{\overline{\F}}$:
\begin{align}\label{fyn1}
 G_{+-}^\F(x,x')
  &= \ii \frac{\big(\Omega_{+}|\T\big(\hat\phi(x)\hat\phi(x')\big)\Omega_{-}\big)
 }{(\Omega_{+}|\Omega_{-})}, \\ \label{fyn2}
  G_{-+}^{\overline{\F}}(x,x')
 &= -\ii \frac{\big(\Omega_{-}|\overline{\T}\big(\hat\phi(x)\hat\phi(x')\big)\Omega_{+}\big)
 }{(\Omega_{-}|\Omega_{+})}.
\end{align}
We will see below that $G_{+-}^\F$ and $G_{-+}^{\overline{\F}}$ play an
important role in applications, and possess an alternative definition that
works well even if the overlap $(\Omega_{-}|\Omega_{+})$ is formally zero.

In a generic situation, \eqref{fyn1}, \eqref{fyn2} and \eqref{Scat2} may be
ill-defined because the overlap $(\Omega_{+}|\Omega_{-})$ is
zero.
Fortunately, as we will see, one can define $G_{+-}^\F$ and
$G_{-+}^{\overline{\F}}$ independently via operator theory,
without a division by zero.

On an asymptotically stationary and stable spacetime it is natural to
use for the initial, resp. final  representation the Hilbert space
generated by $\Omega_-$, resp.
$\Omega_+$.
Thus the main objects of interest are \begin{align}
\frac{S_{+-}}
{(\Omega_{+}|\Omega_{-})},&\quad
\frac{S_{+-}^*}
{(\Omega_{-}|\Omega_{+})}. \label{Scat2}\end{align}
They can be evaluated using
$G_{+-}^\F$ and $  G_{-+}^{\overline{\F}}$, even  if
 ${(\Omega_{+}|\Omega_{-})}=0$.

The main topic of the present article is  how to define
various propagators using tools of operator
theory. We will see in particular that one does not need to worry about
dividing by the overlap $(\Omega_\alpha|\Omega_\beta)$. It is
possible to give a purely operator theoretic definition of
\eqref{eq:zero_divided_by_zero},  \eqref{eq:GF0}, \eqref{eq:GFa.2} and
\eqref{eq:GFb.2}, which works  also if
$(\Omega_{\alpha}|\Omega_{\beta})=0$.

\subsection{Operator-theoretic interpretations of propagators}
There are two distinct operator-theoretic settings related to the
Klein-Gordon equation, which are useful in defining and computing
propagators: the space of solutions to
\eqref{kg1}, which we denote $\cW_\KG$, and the Hilbert
space $L^2(M,|g|^{\frac12})$.  The space $\cW_\KG$ may be called
the \emph{on-shell space} and $L^2(M,|g|^{\frac12})$ the \emph{off-shell space}.

To define the on-shell space one usually  starts from the space of
complex space-compact solutions to \eqref{kg1},
denoted $\cW_\sc$. This space is endowed with the so-called
\emph{Klein-Gordon charge form}---an indefinite sesquilinear form obtained
by integrating the natural current over an arbitrary Cauchy surface.
 In the generic case, this space does not have a
distinguished positive scalar product. Nevertheless,  one can fix
  a family of equivalent positive scalar products compatible
  with the Klein-Gordon form. Then, for technical
  reasons, we extend $\cW_\sc$ to a complete space $\cW_\KG$,
which has the structure of Krein space:  a space
 with a Hilbertian topology equipped with a distinguished
indefinite  form given by a bounded self-adjoint involution. Using elements of the theory of
 Krein spaces one is able to give meaning to
the quantities \eqref{eq:zero_divided_by_zero},  \eqref{eq:GF0},
\eqref{eq:GFa.2} and \eqref{eq:GFb.2}, avoiding expressions of the type
$\tfrac{0}{0}$. This is a big advantage of the
operator-theoretic viewpoint.

In practice, it is convenient to represent the space $\cW_\KG$ in terms
of Cauchy data. More precisely, we first identify
$M=\mathbb{R}\times\Sigma$, where $\Sigma$ has a spatial
signature  and $\partial_t$ a temporal signature. Each element of $\cW_\KG$ is uniquely determined by its
value at $\{t\}\times\Sigma$ and its temporal derivative. This allows us to
describe elements of $\cW_\KG$ as pairs of  functions on $\Sigma$.

The space $\cW_\KG$
is not the only operator-theoretic setting for propagators.
There is another one, provided by the Hilbert space is
$L^2(M,|g|^{\frac12})$.
At first many readers may protest -- this space does not describe
physically relevant states. However, as we will see it is very useful
 for the computation of propagators.

It can be easily   shown that on Minkowski space the
usual Feynman and anti-Feynman propagator are the boundary values of
the resolvent kernel of the Klein-Gordon operator on $L^2(\mathbb{R}^{1,d-1})$:
\begin{align}\label{opero1}
  G^\F(x,y)&:=\lim_{\epsilon\searrow0}\frac{1}{(-\Box+m^2+\ii\epsilon)}(x,y),\\
  G^{ \overline\F}(x,y)&:=\lim_{\epsilon\searrow0}\frac{1}{(-\Box+m^2-\ii\epsilon)}(x,y).\label{opero2}
\end{align}
It is not difficult to see that an analogous statement is true on  stationary
stable spacetimes.

More generally, suppose  we use the path integral
formalism to define perturbative  QFT. The usual prescription says that
one should split the action in a quadratic part and the interaction,
and then derive Feynman diagrams from the path integral. It is easy to
see that this prescription formally yields \eqref{opero1} and \eqref{opero2} as the
expressions corresponding to the lines in Feynman diagrams.
This suggests an  alternative  definition of Feynman and
anti-Feynman propagator, which we describe below.

It is clear that for real-valued $Y(x)$, $-\Box+Y(x)$ is a Hermitian
operator on $C_c^\infty(M)$ in the sense of $L^2(M,|g|^{\frac12})$. Suppose that it
is essentially self-adjoint. Then its spectrum is contained in $\mathbb{R}$ and we
may define the  {\em operator-theoretic  Feynman and anti-Feynman propagator}
$ G_\op^\F(x,y)$ and $G_\op^{ \overline\F}(x,y)$ via
\begin{align}\label{opero1.}
  G_\op^\F(x,y)&:=\lim_{\epsilon\searrow0}\frac{1}{(-\Box+Y(x)+\ii\epsilon)}(x,y),\\
  G_\op^{ \overline\F}(x,y)&:=\lim_{\epsilon\searrow0}\frac{1}{(-\Box+Y(x)-\ii\epsilon)}(x,y),\label{opero2.}
\end{align}
provided that the distributional limits on the right-hand side exist.

Note  that there is no guarantee that the limits \eqref{opero1.}
and \eqref{opero2.} exist. For instance, if the Klein-Gordon operator
has a zero eigenvalue, they fail to exist.

For static stable Klein-Gordon operators the existence of
\eqref{opero1.} and \eqref{opero2.} is proven in \cite{DS18}.
There are heuristic arguments \cite{DS19,DS22} showing that the above
definitions work   on asymptotically stationary stable
spacetimes and
the in-out Feynman and the out-in anti-Feynman propagator
coincide with the
operator-theoretic Feynman propagators:
\begin{align}\label{pqp1}
 G^\F_{\rm op} (x,y)= G_{+-}^\F(x,y),\\
 G^{\overline{\F}}_{\rm op}(x,y) =  G_{-+}^{\overline{\F}}(x,y). \label{pqp2}
\end{align}
These identities can be viewed as a justification of the
path-integral approach to QFT.

From the rigorous point of view, the definitions \eqref{opero1.} and \eqref{opero2.} raise difficult
mathematical questions.
First, the essential self-adjointness for generic spacetimes is a
nontrivial problem. For asymptotically Minkowskian spacetimes
satisfying some non-trapping conditions it has been proven in
\cite{V20,NT23,NT23_2}. Under similar conditions one can show that
\eqref{pqp1} and \eqref{pqp2} are true.

Propagators satisfy various identities. We already mentioned
\eqref{paulijordan}, which defines the Pauli-Jordan
propagator. Another identity universally true is
\begin{align}
G^\PJ(x,x')=\ii G_{\alpha,\beta}^{(+)}(x,x')-\ii G_{\alpha,\beta}^{(-)}(x,x'),\label{paulijordan00}
\end{align}
 valid for any pair of Fock states $\omega_\alpha,\omega_\beta$.

On Minkowski space with $Y(x)=m^2\geq0$,
and more generally for a stationary stable  Klein-Gordon equation, we
have the identity
\begin{align}\label{special}
G_\op^\F + G_\op^{\overline{\F}}&= G^{\lor} + G^{\land}.
\end{align}
 In particular, the support of $G_\op^\F + G_\op^{\overline{\F}}$ is
causal.

\begin{definition}
 We will say that the Klein-Gordon equation is {\em special} if one
 can define $G_\op^\F$  and $G_\op^{\overline{\F}}$ (which we expect
 to be true in typical situations) and
 the support of $G_\op^\F + G_\op^{\overline{\F}}$ is
causal.
We will then also say that the {\em specialty condition} is satisfied.
\end{definition}

 Special Klein-Gordon equations have the following
advantage. One may expect that it is in many situations quite
simple to compute the distributions $G_\op^\F$  and $G_\op^{\overline{\F}}$
using operator-theoretic tools. Then, splitting
$G_\op^\F + G_\op^{\overline{\F}} $ into two distributions,
one supported in the causal future and the other supported in
the causal past, we may determine $G^{\lor}$ and $G^{\land}$.

 The specialty condition is generically violated. It is however very
  useful if it holds. We will discuss some
  interesting cases when it is true.

\subsection{Outline of the paper}
In Section \ref{ssc:props_curved}, we describe in detail
both basic operator-theoretic settings to QFT on curved spacetimes that we
outlined in the introduction: the on-shell space $\cW_\KG$ and the
off-shell space $L^2(M)$.

The remaining sections are dedicated to the discussion of various examples
of spacetimes with largely different properties:
\begin{enumerate}
\item
First we discuss stationary spacetimes. Here one can give fairly
explicit formulas for   the Pauli-Jordan bisolution, and the four
  basic Green functions: the avanced/retarded propagators, and
  operator-theoretic (anti-)Feynman propagators.
  If in addition
  the Klein-Gordon equation is stable, then there is a distinguished
Fock  state. The corresponding 2-point functions and
(anti-)Feynman propagators are easy to describe. The
specialty condition is fulfilled and the (anti-)Feynman propagators
defined in the off-shell and on-shell formalism coincide.

In the tachyonic case, that is, if the Hamiltonian is not positive,
the speciality condition is violated, and we cannot define
positive/negative  frequency bisolutions. This includes the Minkowski
space with imaginary mass, that is, $m^2<0$. Of course, this
case is not very physical, but it is occasionally discussed in the
literature.

\item Spacetimes asymptotically stationary and stable in the past and
  future form a class well suited for the formalism of QFT. After
  identifying $M$ with $\mathbb{R}\times \Sigma$, where $\mathbb{R}$
  describes time and $\Sigma$ is a Cauchy surface with a time dependent
  Riemannian metric, one can
  give a fairly explicit description of all propagators using the time
  evolution of solutions, as described in \cite{DS22}. Remarkably, the
in-out  Feynman and out-in anti-Feynman propagator are well
defined---this is a non-trivial statement proven in \cite{DS22}.
The existence of the operator-theoretic (anti-)Feynman propagator
is a difficult mathematical problem, solved only under some strong
  assumptions. There are heuristic arguments showing that, if they
  exist, they coincide with the in-out  Feynman and  out-in anti-Feynman
  propagator.
As we mentioned above, the
specialty condition is rarely fulfilled.
 \item The Klein-Gordon operator on $1+0$-dimensional spacetimes
   essentially reduces to a one-dimensional
 Schrödinger operator. The corresponding propagators are
 well-known objects from the theory of such operators.
The  speciality condition is fulfilled if and only if
the scattering operator is reflectionless. Obviously, it is satisfied if the potential is a
constant. But curiously, as is well known, there exist potentials
which are reflectionless at all energies. The best known such
potential is
\beq-\frac{\mu^2-\frac14}{\cosh^2 x}\label{scarfi}
\eeq
for half-integer $\mu$.
\item Spacetimes, whose pseudometric depends on time only through a
  conformal factor,  are usually called
  Friedmann-Lemaître-Robertson-Walker (FLRW) spacetimes. In such
  spacetimes, after diagonalization of the spatial Laplacian, or in
  other words, after decomposing it into ``modes'', the Klein-Gordon
  equation can be reduced to the $1+0$-dimensional setting.
Thus in principle one can write all propagators as the direct sum
 or integral of
propagators for each mode. In particular, the Klein-Gordon equation is
special if each mode is reflectionless.

\item The theory of propagators on the $d$-dimensional de Sitter space $\dS_d$
  is especially rich and surprising.

The de Sitter space can be interpreted as the ``Wick rotated''
  $d$-dimensional sphere. Analytically continuing the Green
  function of the sphere  in the usual spherical
  coordinates we obtain a certain Feynman and anti-Feynman propagator.
  For  $m^2\geq (\frac{d-1}{2})^2$,  they can be used to write down
  the Wightman two-point functions of a state, as well as the classical propagators. This
  state is usually called the Euclidean (or Bunch-Davies) state and is believed to be
  the physical choice on the de Sitter space, because it is Hadamard.

  The d'Alembertian on de Sitter space is essentially self-adjoint
  on smooth compactly supported functions. This is a special case of a
  general mathematical theorem saying that invariant differential
  operators on maximally symmetric pseudo-Riemannian manifolds are
  essentially self-adjoint. In our paper we compute the integral
  kernel of resolvent of the d'Alembertian  on $\dS_d$. Taking its
  boundary values yields the operator-theoretic Feynman and anti-Feynman
  propagator. Curiously, they are different from the Euclidean Feynman
  and anti-Feynman propagator. The specialty condition is satisfied in odd
  dimensions; it is not true in even dimensions.

 Our derivation of the formula for the resolvent is based on an
argument which, while in our opinion convincing and elegant, is not fully
rigorous. Alternative, more complicated proofs of our formula are possible,
e.g. following the approach of \cite{FSS13} based on mode decompositions.

  It is well-known that all de Sitter invariant states can be
  described and expressed in terms of Gegenbauer functions.
  They are
  usually called  {\em $\alpha$-vacua}, where $\alpha$ is a complex
  parameter that can be used to parametrize them.
  $\alpha=0$
  corresponds  to the Euclidean vacuum. All other
  $\alpha$-vacua are not Hadamard.

   The de Sitter space  is not asymptotically stationary.
   However, it possesses two distinguished states, which can be called
   the in-state and the-out
state. The former has an incoming behavior in the past,
the latter is outgoing in the future. The operator theoretic Feynman
and anti-Feynman
propagators satisfy the identitities \eqref{pqp1} and
\eqref{pqp2}. In odd dimensions the in-state coincides with the
out-state. In even dimensions this is not the case. In all dimensions,
the in-state and out-state are distinct from Euclidean state.

De Sitter space is a FLRW spacetime (with a conformal factor that
blows up exponentially). Therefore, it is possible to decompose the
Klein-Gordon equation into modes. In each mode one obtains the 1-dimensional
Schr\"odinger operator with the potential \eqref{scarfi},
 where $\mu $
depends on the dimension and the degree of spherical harmonics. $\mu $
is a half-integer for odd
dimensions and an integer for even dimensions. This is another way to
see that the Klein-Gordon equation in odd dimensions is special and in
even dimension is not.

 One can define retarded and advanced
propagators for all values of $m^2\in\mathbb{R}$. However,
 the case $m^2<   (\frac{d-1}{2})^2$ seems not physical. In fact,
 below $ (\frac{d-1}{2})^2$ the spectrum of the
d'Alembertian is discrete. Operator-theoretic Feynman and anti-Feynman propagators
are well defined (and identical) outside of this spectrum.
As can be expected, the specialty condition is then violated.

\item  The universal cover of anti-de Sitter space $\tAdS_d$
  is another maximally symmetric spacetime, where one can compute all
  propagators.
It is a stationary spacetime, which is not globally hyperbolic:
it possesses geodesics that escape to the spatial boundary in a
finite proper time.
  One can apply two approaches to define the propagators on the
  universal cover of anti-de Sitter space.

  The first approach uses $L^2(\tAdS_d)$.
 The d'Alembertian  is essentially self-adjoint---there is no need to
 fix boundary conditions. We compute the resolvent of the
   d'Alembertian and define the operator-theoretic Feynman and
   anti-Feynman propagators as its limits. (Again, our computation is based
   on a conjecture and thus not fully rigorous.)
   If $m^2>-(\frac{d-1}{2})^2$, then their sum has a causal support, so one can
   define the retarded and advanced propagator by  splitting
this sum. In
particular, the specialty condition is satisfied.

Alternatively, one can use the evolution of the Cauchy data. For $m^2\geq
-(\frac{d-1}{2})^2+1$ this evolution is uniquely defined--one does not
need to specify boundary conditions. For $m^2<-(\frac{d-1}{2})^2+1$
boundary conditions are needed. For $-(\frac{d-1}{2})^2\leq
m^2<-(\frac{d-1}{2})^2+1$ there exists a distinguished boundary
condition (corresponding to the Friedrichs extension), which agrees
with the propagators obtained from the operator-theoretic Feynman
propagator. In particular, we have  distinguished retarded and
advanced propagators. For $m^2 <-(\frac{d-1}{2})^2$ there are no distinguished
boundary conditions at spatial infinity.
Thus retarded and advanced propagators are non-unique and none is distinguished.
\end{enumerate}

 Pertinent elements of the theory of Krein spaces are discussed in Appendix
\ref{app:Involutions_and_projections}.
Propagators on de Sitter and anti-de Sitter space can be
described explicitly in terms of special functions (Gegenbauer
functions). We introduce their relevant properties in Appendix
\ref{app:gegenbauer}.

\begin{remark}
 We restrict our considerations to a real scalar field $\hat\phi(x)$,
 but they can be generalized to a complex scalar field in a fairly
 straightforward manner. One needs then two pairs of creation
 and annihilation operators. Both the
 real and the complex formalism are treated in \cite{DS22}.

\end{remark}

\subsection{Literature about the subject}

Quantum Field Theory on curved spacetimes is one of the most discussed
and developed areas of theoretical physics. It has enormous literature,
including numerous standard textbooks \cite{BirrellDavies,PT,BF,Gerard}.
At first glance, our paper may appear to be a review article, as it presents
various facts and concepts from existing literature. Surprisingly, however,
many of ideas of our paper seem to be clearly articulated here for the first time.
Let us in particular mention:
\begin{itemize}
 \item the description of the two operator-theoretic
setups in Section \ref{ssc:props_curved} and Appendix
\ref{app:Involutions_and_projections}, which is a continuation
of the works \cite{DS18,DS19,DS22} of D. Siemssen and one of the authors
(JD);
\item the comparison of four different approaches to the Klein-Gordon
equation on de Sitter space from Section \ref{ssc:dS},
where we also present a list of new formulas e.g. for correlation
functions between different states;
\item the
discussions of the ``speciality condition'' throughout all sections,
which provides a useful tool for computations.
\end{itemize}

We start our review of the literature with the
``classical propagators'', that is, the retarded and
advanced propagator, and the Pauli-Jordan function. They belong to
standard knowledge and are
well-studied in standard references. In the  massless case on the flat
 $\mathbb{R}^{1,3}$ the retarded and advanced propagators are well
 known from classical electrodynamics, and are sometimes called
the {\em Lienard-Wiechert potentials}. In the massive flat case their expressions
in terms of Hankel functions are contained in many textbooks.
The Cauchy problem of the wave equation on curved
spacetimes was studied already  by Hadamard \cite{hadamard}, at least locally.
A recent reference to this subject  on arbitrary globally
hyperbolic manifolds is
the book by B\"ar, Ginoux and Pf\"affle \cite{bar}. In the
introduction to this book one reads: ``Tracing back the references
[on the uniqueness and existence  of linear wave equation on
lorentzian manifolds] one typically ends at unpublished
lecture notes of Leray \cite{leray} or their exposition by
Choquet-Bruhat \cite{choquet-bruhat}.''

In the literature the Pauli-Jordan
function is often called the commutator function or (recently, in the
mathematics oriented literature) the
causal propagator, \cite{BF,Gerard}. Note, however, that the latter name can lead to
confusion: in \cite{bogoliubov} the Feynman propagator is called the
causal Green function.

Propagators on the Minkowski space, including ``non-classical''
ones,
are well-known
from various textbooks on Quantum Field Theory (especially the
old-fashioned ones). For instance,  Appendix 2 of Bogoliubov--Shirkov
\cite{bogoliubov} and Appendix C of Bjorken--Drell \cite{bjorken}
contain expressions for these functions in the position space in the
physical case of $\mathbb{R}^{1,3}$, and discuss  conventions used by various
authors.

``Non-classical'' propagators are expectation values of
products of two fields. Those without time-ordering, sometimes called
Wightman functions, are ubiquitous in the mathematical
literature, since they are needed to define the GNS
representation and multiplication in appropriately defined local algebras. One of major questions, which is asked in various
papers  is whether they satisfy the Hadamard condition.

Expectation values of time-ordered fields, that is, Feynman
propagators, are needed when we want to
find scattering amplitudes. They often appear  in the physics
literature as
mixed two-point
functions, typically with the  out-vacuum on the left and in-vacuum on
the right. For instance, in Birrell-Davies \cite{BirrellDavies} in (9.13) one finds
the following definition of Green functions:
\beq
\tau(x_1,x_2...x_m)=
\frac{\langle \mathrm{out},0|
  T(\phi(x_1)\phi(x_2)...\phi(x_m)|0,\mathrm{in}\rangle}
{\langle \mathrm{out},0|0,\mathrm{in}\rangle}.\eeq
Then the authors write: ``...unlike the case of Minkowski space where
$|0,\mathrm{out}\rangle=|0,\mathrm{in}\rangle$    (up to a phase
factor), the vacuum $|0,\mathrm{in}\rangle$    in curved spacetime
will not in general be stable: $\langle \mathrm{out},0|0,\mathrm{in}\rangle\neq1$.''
In particular the relationship \eqref{pqp1}, which says that the
``in-out Feynman propagator''$G_{+-}^\F$ coincides with the Feynman propagator
formally computed  in the path-integral approach (which can be
interpreted as $G_\op^\F$) is implicitly contained in
\cite{BirrellDavies} (and in general in the physics literature).
 Elements of this philosophy are also found
in \cite{rumpf1,RumpfdS}.

In the more recent rigorous literature, mixed (two-state)
propagators are almost absent. The majority of recent works,
 for example the seminal papers
\cite{brunetti-fredenhagen,hollands-wald1}, emphasize
the local point of view. Their usual goal is to construct
a net of local algebras, for which it is enough to fix a single
state, preferably Hadamard, which can be done locally.

A systematic rigorous study of various natural propagators on curved
spacetimes was undertaken in the series of papers by one of the
authors (JD)  with a coauthor \cite{DS18,DS19,DS22}. In particular,
the construction of the distinguished Feynman propagator by methods of
Krein spaces on an asymptotically stationary stable spacetimes is
contained in \cite{DS22}. A construction of the same Feynman
propagator on a (more narrow) class of asymptotically Minkowskian
spaces by methods of pseudodifferential calculus was given by G\'erard
and Wrochna \cite{GW19,Gerard}.

There exist many works, especially  in the PDE
literature, about {\em parametrices} of the
Klein-Gordon equations, that is, inverses modulo a {\em smoothing
  operator}. A celebrated paper with this philosophy is the work by
Duistermaat and H\"ormander \cite{DuHor}, which describes four natural
parametrices: retarded, advanced, Feynman and anti-Feynman. Such
parametrices are enough in the study of propagation of singularities,
and they do not require a global knowledge of the spacetime. Anyway,
most of the literature using parametrices seems restricted to retarded
and advanced propagators. See e.g. \cite{CDD20,CV22} where parametrices
involving Fourier integral operators are used as approximations of exact
retarded and advanced propagators.

Similarly, it is often argued in mathematical physics papers that it is
enough to know a two-point function only up to a {\em smooth term}. This
is sufficient if we want to prove the existence of renormalized powers
of fields \cite{brunetti-fredenhagen}.
In our paper we are interested in {\em exact} Green functions and
bisolutions, which are needed to compute scattering amplitudes exactly.

The usefulness of the setting of Krein spaces for the Klein-Gordon
equation has been known for a long time
\cite{Ves70,Naj79,Naj80,Ves90,GGH13,GGH15}.
Note that in some of these papers  the ``Klein-Gordon operator'' means the
``generator of the Klein-Gordon evolution'', denoted in our paper by
$B(t)$. For us the ``Klein-Gordon'' operator is the operator on
$L^2(M)$ whose resolvent appears in \eqref{opero1.} and \eqref{opero2.}.

To our knowledge the above papers miss the relevance of
the Krein setting for two-state Wightman  functions and in-out
Feynman propagators. This seems to have been noted only in \cite{DS22}.

The rigorous literature about the off-shell approach to the Klein-Gordon
equation seems very scarce. To the papers about self-adjointness of
Klein-Gordon operators mentioned above \cite{V20,NT23,NT23_2,DS22}, one
could add \cite{kaminski} about pathological examples and \cite{BN}
about the Wick rotation on a background of an ADM metric.

We will discuss the literature about the examples that we present
in Sections \ref{sec:asymp_stat}, \ref{sec:FLRW}, \ref{ssc:dS}
and \ref{sec:AdS} in the respective sections.

\section{Propagators in curved spacetimes}
\label{ssc:props_curved}

\subsection{Klein-Gordon equation}

In this section we will describe how to generalize the well-known  propagators
from ${\mathbb R}^{1,d-1}$ to generic spacetimes.

Consider  a  Lorentzian manifold   $M$  of dimension $d$ with
{\em pseudometric tensor}
$g_{\mu\nu}$. Define the {\em d'Alembertian}
\begin{align}
\label{klein}
\Box:=   |g|^{-\frac12} \partial_\mu|g|^{\frac12} g^{\mu\nu}
\partial_\nu.
\end{align}
Note that we have
\beq
  \Box=g^{\mu\nu}\nabla_\mu\nabla_\nu,\label{klein1}\eeq
  where the left $\nabla$
is the covariant derivative on covectors, and the right $\nabla$ on
scalars   (which coincides with the usual derivative $\partial$).
We also define the {\em Klein-Gordon operator}
$-\Box+Y(x),$ where $Y(x)$ is
an $x$-dependent, real-valued {\em   scalar potential}.
Most of the time we will assume that $Y(x)=m^2$, so that
the Klein-Gordon operator is $-\Box+m^2$.

Note that the d'Alembertian
\eqref{klein} acts on scalar functions. It is sometimes
more convenient to
replace it by the d'Alembertian  in the half-density formalism, that
is
\begin{align} \Box_{\frac12}&:=|g|^{\frac14}\Box |g|^{-\frac14}
=   |g|^{-\frac14} \partial_\mu|g|^{\frac12} g^{\mu\nu}
    \partial_\nu |g|^{-\frac14}
                              .\label{half}
\end{align}
In the half-density formalism the space $L^2(M,|g|^{\frac12})$ is
replaced by $L^2(M)$, where we just take the Lebesgue measure with respect to
given coordinates. We will write $\Box$ for $\Box_{\frac12}$ when it is clear
from the context that we use the half-density formalism. See e.g. \cite{DS22}.

\subsection{Green functions and bisolutions}
Suppose that we have a continuous  sesquilinear form
\begin{align}
  C_\mathrm{c}^\infty(M)\times
  C_\mathrm{c}^\infty(M)\ni(f_1,f_2)\mapsto
 ( f_1|A f_2)\in\CC.\label{afa}
 \end{align}
  By the Schwartz Kernel Theorem there exists a distribution
  $A(\cdot,\cdot)$ on $M\times M$, so that \eqref{afa} in local
  coordinates can be written as
  \beq
   (
  f_1|Af_2)=\int\int \overline{f_1(x)}|g|^{\frac12}(x)A(x,y)
  |g|^\frac12(y)f_2(y)\dd x\dd y.\eeq
 The distribution   $A(\cdot,\cdot)$ will be called the {\em integral
 kernel of $A$}. Note that we use the integral notation for distributions
 and that we say ``integral kernel''  for $A(\cdot,\cdot)$ even if it
 is a distribution.

  Actually, the  above definition applies only to the scalar formalism.
If we use the half-density formalism, then the integral kernel is
different:
\beq\label{pass}
A_\frac12(x,y):=|g|^{\frac14}(x)A(x,y) |g|^{\frac14}(y).\eeq

For instance, the integral kernel of the identity is in local coordinates
\begin{align}
  \text{ in the scalar formalism }& \quad|g|^{-\frac12}(x)\delta(x-y),\\
  \text{ in the half-density formalism }&\quad  \delta(x-y).
\end{align}

The definition of a Green function (of the Klein-Gordon operator)
will have also two versions. It is a distribution on $M\times M$
satisfying
\begin{align}\label{green1}
  \text{ in the scalar formalism }&
                                \quad\big(-\Box+Y(x)\big)G_\bullet(x,y)=
                                |g|^{-\frac12}(x)\delta(x-y),\\
  \text{ in the half-density formalism }&\quad \quad\big(-\Box_{\frac12}+Y(x)\big)G_{\bullet,\frac12}(x,y)=
                               \delta(x-y),
                                     \label{green2} \end{align}
and analogous conditions with $x$ replaced by $y$. We also have
similar definitions of bisolutions  (of the Klein-Gordon operator),
where the right hand sides of \eqref{green1}
and \eqref{green2} are zero. One can pass from the scalar formalism
to the half-density formalism as in \eqref{pass}:
\beq\label{pass1}
G_{\bullet,\frac12}(x,y):=|g|^{\frac14}(x)G_\bullet(x,y) |g|^{\frac14}(y).\eeq

\subsection{Classical propagators}
Suppose that $M$ is
globally hyperbolic.
It is well-known \cite{bar} that there exist unique fundamental solutions
$G^{\lor}(x,y)$ and $G^{\land}(x,y)$ of the Klein-Gordon equation
which have future- respectively past-directed causal support:
\begin{align}
(x,y)\in\supp  G^{\lor} &\quad\Rightarrow\quad
\exists \textup{ future oriented causal curve from } y \textup{ to } x,
\\ \notag
(x,y)\in\supp  G^{\land} &\quad\Rightarrow\quad
\exists \textup{ future oriented causal curve from } x \textup{ to } y.
\end{align}

$G^{\lor}(x,y)$ is called the \emph{forward} (or \emph{retarded})
\emph{propagator}, $G^{\land}(x,y)$ is called the \emph{backward}
(or \emph{advanced}) \emph{propagator}. Their difference, which
obviously is a bisolution of the Klein-Gordon equation, is called
the \emph{Pauli-Jordan propagator}  (or {\em commutator function)}
\begin{align}\label{paulijordan1}
 G^\PJ (x,y):= G^\lor(x,y)-G^\land(x,y).
\end{align}
These three propagators are sometimes called jointly
\emph{classical propagators} \cite{DS18,DS22}.

Identify $M$ with $\mathbb{R}\times\Sigma$, where for $t\in\RR$
  the metric on
${t}\times\Sigma$ is Riemannian and $\partial_{t}$ is timelike.  Such
an identification is always possible for globally hyperbolic
manifolds. We will then say that we fixed a {\em time variable on $M$}.
(We will discuss this in more detail
in Subsection \ref{Classical propagators from evolution equations}).
Suppose that we can multiply the distribution
$ G^\PJ (x,y)$ by the discontinuous function $\theta(x^0-y^0)$. Again
in Subsection \ref{Classical propagators from evolution equations}, we
will see a rather general setting where this is rigorously allowed. Then
we can retrieve the advanced and retarded Green functions from the
Pauli-Jordan propagator:
\begin{align} G^\lor(x,y)&=
                           \theta(x^0-y^0)G^\PJ(x,y),
\\
  G^\land(x,y)&:=-\theta(y^0-x^0)G^\PJ(x,y). \label{paulijordan3}
                 \end{align}

\subsection{Quantum fields and non-classical propagators}
We still assume that $M$ is globally hyperbolic.
Consider a real scalar quantum field $\hat\phi(x)=\hat \phi(x)^*$ on $M$
satisfying
\begin{align}
 \big(-\Box+Y(x)\big)\hat\phi(x) &=0, \\ \notag
[\hat\phi(x),\hat\phi(y)]&=-\ii G^\PJ(x,y)\one.
\end{align}

More precisely, we assume some elements of the Wightman axioms:
We suppose that $\cD$ is a complex vector space equipped with
a scalar product $(\cdot|\cdot)$ such that
\beq C_\mathrm{c}^\infty(M)\ni f\mapsto \hat\phi[f]\eeq
is a  distribution with values in linear operators
from $\cD$ to $\cD$ satisfying
\begin{align}
  \hat \phi\big[\big(-\Box+Y(x)\big)f\big]&=0,\\
  \big[\hat \phi[ f_1],\hat \phi[
  f_2]\big]&=-\ii\int\int G^\PJ(x,y)f_1(x)f_2(y)\dd x\dd y,\\
(   \hat \phi[\bar f]\Phi|\Psi)&=(\Phi|\hat \phi[ f]\Psi),\\
  C_\mathrm{c}^\infty(M)\ni f\mapsto
                               (\Phi|\hat\phi[f]\Psi)&\text{
                               is continuous for }\Phi,\Psi\in\cD.
  \end{align}
(We do not require that $\cD$ is complete).

 Consider $\Omega_\alpha,\Omega_\beta\in\cD$ with
  $\|\Omega_\alpha\|=\|\Omega_\beta\|=1$ and
$(\Omega_\alpha|\Omega_\beta)\neq0$.
By the Schwartz Kernel Theorem
there exist distributions $G_{\alpha,\beta}^{(+)}(\cdot,\cdot)$ and
$G_{\alpha,\beta}^{(-)}(\cdot,\cdot)$   on
  $M\times M$ such that
  \begin{align}\label{2point1}
  \int\int f_1(x)G_{\alpha,\beta}^{(+)}(x,y)f_2(y)\dd x\dd
    y=\frac{\big(\Omega_\alpha|\hat\phi[f_1]\hat\phi[f_2]\Omega_\beta)}{(\Omega_\alpha|\Omega_\beta)},\\\label{2point2}
  \int\int f_1(x)G_{\alpha,\beta}^{(-)}(x,y)f_2(y)\dd x\dd y=\frac{\big(\Omega_\alpha|\hat\phi[f_2]\hat\phi[f_1]\Omega_\beta)}{(\Omega_\alpha|\Omega_\beta)}.\end{align}
 Note that both $ G_{\alpha,\beta}^{(+)}$ and $G_{\alpha,\beta}^{(-)}$
 are bisolutions and they satisfy
 \beq G_{\alpha,\beta}^{(+)}-G_{\alpha,\beta}^{(-)}=-\ii G^\PJ.
\label{paulijordan2}\eeq

 In the special case $\Omega_\alpha=\Omega_\beta$ we will write
\begin{align} G_\alpha^{(+)}&:=G_{\alpha,\alpha}^{(+)},\label{G1a.}\\\label{G2a.}
   G_\alpha^{(-)}&:=G_{\alpha,\alpha}^{(-)}.\end{align}

$\hat\phi[f]$ are called {\em smeared fields}. In the physics
literature one introduce $\hat\phi(x)$, the {\em field at point $x\in M$} and one writes
\beq\hat\phi[f]=\int\hat\phi(x)f(x)\dd x,\eeq
so that \eqref{G1a.}, \eqref{G2a.},
\eqref{2point1} and \eqref{2point2} are rigorous versions of
the heuristic definitions \eqref{G1}, \eqref{G2},
\eqref{eq:zero_divided_by_zero},
\eqref{eq:GF0} from the introduction.

$G_{\alpha}^{(+)}$ and $G_{\alpha}^{(-)}$ are often called
{\em Wightman}, resp. {\em anti-Wightman functions} (of the state
given by $\Omega_\alpha$).
$G_{\alpha,\beta}^{(+)}$ and $G_{\alpha,\beta}^{(-)}$ can be called
{\em 2-state  Wightman}, resp. {\em anti-Wightman functions}.

 Sometimes, it is also useful to consider the symmetric two-point function
\begin{align}
  G^\sym_\alpha(x,x')
  :=&\; G^{(+)}_\alpha(x,x') + G^{(-)}_\alpha(x,x') .
 \end{align}

 For $\Omega_\alpha,\Omega_\beta$ as above, we define
 the Feynman and anti-Feynman propagator assocated with
 $\Omega_\alpha$ and $\Omega_\beta$
as follows:
\begin{align}
  G_{\alpha,\beta}^\F &=  \ii G_{\alpha,\beta}^{(+)} + G^{\land}                                                                   \label{eq:relDcurved1}\\
&= \ii G_{\alpha,\beta}^{(-)} + G^{\lor} ,
                                                                     \label{eq:relDcurved2}
  \\
G_{\alpha,\beta}^{\overline{\F}} &= -\ii G_{\alpha,\beta}^{(+)} +
                                   G^{\lor}
  \label{eq:relEcurved1}\\
&= - \ii G_{\alpha,\beta}^{(-)} + G^{\land} .
\label{eq:relEcurved2}
\end{align}
Note that the equalities \eqref{eq:relDcurved1}$=$\eqref{eq:relDcurved2}
and
\eqref{eq:relEcurved1}$=$\eqref{eq:relEcurved2}
follow from the properties of the Pauli-Jordan propagator:
\eqref{paulijordan1} and \eqref{paulijordan2}.
Obviously,  $G_{\alpha,\beta}^{\F}$ and
$G_{\alpha,\beta}^{\overline{\F}}$ are Green functions, being sums of
bisolutions and Green functions.

 In the special case $\Omega_\alpha=\Omega_\beta$ we will write
\begin{align} G_\alpha^{\F}&:=G_{\alpha,\alpha}^{\F},\label{G1a}\\\label{G2a}
   G_\alpha^{\bar\F}&:=G_{\alpha,\alpha}^{\bar\F}.\end{align}

Suppose we fix a time variable $x^0$ on
$M$. Then, using
\eqref{paulijordan1}, \eqref{paulijordan2}, and \eqref{paulijordan3}
we  can rewrite the definitions of
the Feynman and anti-Feynman propagator in a way that is not manifestly
rigorous, is however more symmetric and closer to the usual
treatment in textbooks:
\begin{align}\label{feyny1.}
G_{\alpha,\beta}^\F(x,y) &:= \ii \Big( \theta(x^0-y^0) G_{{\alpha,\beta}}^{(+)}(x,y)
             +\theta(y^0-x^0)G_{{\alpha,\beta}}^{(-)}(x,y)\Big),
        \\ \label{feyny2.}
G_{\alpha,\beta}^{\overline{\F}}(x,y) &:= -\ii \Big(
          \theta(x^0-y^0) G_{{\alpha,\beta}}^{(-)}(x,y)
        +\theta(y^0-x^0) G_{{\alpha,\beta}}^{(+)}(x,y) \Big).
\end{align}

In the usuall textbook treatment one introduces the chronological and
antichronological time ordering by
\begin{align}\label{eq:Tdef}
  \T\big(\hat\phi(x)\hat\phi(y)\big)
        :=\begin{cases}
             \hat\phi(x) \hat\phi(y) ,\quad x^0>y^0, \\
               \hat\phi(y) \hat\phi(x),\quad y^0>x^0,
             \end{cases}\\ \label{eq:Tbardef}
 \overline    \T\big(\hat\phi(x)\hat\phi(y)\big)
        :=\begin{cases}
             \hat\phi(y) \hat\phi(x) ,\quad x^0>y^0, \\
               \hat\phi(x) \hat\phi(y),\quad y^0>x^0.
            \end{cases}
\end{align}
This definition extends uniquely to $M\times M\setminus\{x=y\}$
because of the Einstein causality of the field $\hat\phi(x)$.
Then one defines $G_\alpha^\F$, $G_\alpha^{ \overline\F}$,
 $G_{\alpha,\beta}^\F$ and
$G_{\alpha,\beta}^{ \overline\F}$ as   in
\eqref{eq:GFa}, \eqref{eq:GFb}, \eqref{eq:GFa.2} and \eqref{eq:GFb.2}
from the introduction. Note that this definition does not extend to
the diagonal, unlike the definition that we gave in \eqref{eq:relDcurved1}---\eqref{eq:relEcurved2}.

In order to express the dependence on   $\Omega_\alpha$
and $\Omega_\beta$ one may call $G_{\alpha,\beta}^{(+)}$ and
$G_{\alpha,\beta}^\F$, the $\beta$-$\alpha$-Wightman function
and the $\beta$-$\alpha$-Feynman propagator, respectively---
see the explanation of the inversion of symbols in the paragraph
following \eqref{eq:GFb.2}.

Let us summarize the identities involving the  propagators:
\begin{subequations}
 \label{eq:relations_props}
\begin{align}
G_{\alpha,\beta}^\F - G_{\alpha,\beta}^{\overline{\F}} &=
\ii\Big( G_{\alpha,\beta}^{(+)} + G_{\alpha,\beta}^{(-)}\Big) ,
\label{eq:relBcurved}\\
G^\PJ &= G^{\lor} - G^{\land}
= \ii\Big( G_{\alpha,\beta}^{(+)} - G_{\alpha,\beta}^{(-)}\Big) ,
\label{eq:relCcurved}\\
G_{\alpha,\beta}^\F &=  \ii G_{\alpha,\beta}^{(+)} + G^{\land}
= \ii G_{\alpha,\beta}^{(-)} + G^{\lor} ,
\label{eq:relDcurved}\\
G_{\alpha,\beta}^{\overline{\F}} &= -\ii G_{\alpha,\beta}^{(+)} + G^{\lor}
= - \ii G_{\alpha,\beta}^{(-)} + G^{\land} .
\label{eq:relEcurved}
\end{align}
\end{subequations}

Note that in this subsection  Wightman and anti-Wightman 2-point
functions,  as well as  Feynman and anti-Feynman
Green functions involved arbitrary vectors $\Omega_\alpha$,
$\Omega_\beta$ in a single represention of quantum fields. In most applications
in  these definitions one assumes that $\Omega_\alpha$,
$\Omega_\beta$ are Fock vacua, that is, vectors that yield
Fock representations.
In the following subsection we will give a different definition of
Feynman and anti-Feynman propagators. These definitions
will be restricted to Fock vacua. They will be purely operator-theoretic
and not make use of quantum fields.
These new definitions will have one important advantage: they will
work also if $\Omega_\alpha$ and $\Omega_\beta$ do not belong to the
same representation, and hence $(\Omega_\alpha|\Omega_\beta)=0$ (which
is actually quite common).

 \subsection{Klein-Gordon  charge form}
\label{ssc:2pf_KGkernel}

Let us now develop a mathematical formalism which will yield  an alternative,
more satisfactory definition of non-classical propagators. It will be based
entirely on operator theory, without going through quantum fields.

For $\zeta,\xi\in C^\infty(M)$, set
\begin{align}\label{charge_curved}
\overline{\zeta(x)} \lrnb_\mu \xi(x): =
 \big(\nabla_\mu \overline{\zeta(x)} \big) \xi(x) -\overline{\zeta(x)} \nabla_\mu \xi(x).
\end{align}

  Let $\sWsc$ denote the space of smooth, space-compact solutions to the
Klein-Gordon equation
\beq\big(-\Box+Y(x)\big)\zeta(x)=0.\eeq
Using \eqref{klein1} we see that  if $\zeta,\xi\in\sWsc$, then
\begin{align}\label{current0}
J_\mu[\zeta,\xi](x) := \overline{\zeta(x)} \lrnb_\mu \xi(x) =
 \big(\nabla_\mu \overline{\zeta(x)} \big) \xi(x) -\overline{\zeta(x)} \nabla_\mu \xi(x)
\end{align}
is a covariantly conserved current, which means
\begin{align} \label{current}
|g|^{-\frac12}\partial_\mu|g|^{\frac12} g^{\mu\nu} J_\nu
=\nabla_\mu J^\mu=0.
\end{align}
Therefore,
\begin{align}
\label{eq:KG_charge_curved}
 (\zeta|\xi)_\KG := \ii \int_\Sigma \overline{\zeta(x)} \lrnb_\mu \xi(x)
 \dd\Sigma^\mu(x),
\end{align}
does not depend on the choice of the Cauchy surface $\Sigma$, where
$\dd\Sigma^\mu(x)$ is the natural measure on $\Sigma$
times the future-directed normal vector. \eqref{eq:KG_charge_curved}
is called the Klein-Gordon
charge form.

The Klein-Gordon charge form is not positive definite. We can,
however, usually extend the space $\sWsc$ to a larger space, denoted $\cW_\KG$,
which admits a direct sum decomposition
\begin{align}
 \cW_\KG =\cZ_\alpha^{(+)}\oplus \cZ_\alpha^{(-)},
 \quad \cZ_\alpha^{(+)}= \overline{\cZ_\alpha^{(-)}},\label{decompo}
\end{align}
where the components are orthogonal with respect to $(\cdot|\cdot)_\KG$,
$\cZ_\alpha^{(+)}$ is positive and complete wrt $\sqrt{(\cdot|\cdot)_\KG}$, and
$\cZ_\alpha^{(-)} $ is negative and  complete wrt
$\sqrt{-(\cdot|\cdot)_\KG}$.

Every $\zeta\in\sWsc$, decomposed according to \eqref{decompo} as
\begin{align}
\zeta=\zeta_\alpha^{(+)}+\zeta_\alpha^{(-)},\quad
\zeta_\alpha^{(+)}= \overline{\zeta_\alpha^{(-)}},
\end{align}
satisfies
\begin{align}
\pm(\zeta_\alpha^{(\pm)}|\zeta_\alpha^{(\pm)})_\KG\geq0,\quad
(\zeta_\alpha^{(\pm)}|\zeta_\alpha^{(\mp)})_\KG=0.
\end{align}
Thus
\begin{align}
(\zeta|\xi)_\KG=(\zeta_\alpha^{(+)}|\xi_\alpha^{(+)})_\KG+
(\zeta_\alpha^{(-)}|\xi_\alpha^{(-)})_\KG.
\end{align}

The index $\alpha$ indicates the decomposition \eqref{decompo}.
We also have a positive definite scalar product
\begin{align} \label{alpha1}
(\zeta|\xi)_\alpha=(\zeta_\alpha^{(+)}|\xi_\alpha^{(+)})_\KG-
(\zeta_\alpha^{(-)}|\xi_\alpha^{(-)})_\KG,
\end{align}
which is however less canonical than the Klein-Gordon charge form
 because it depends on the decomposition \eqref{decompo}. Note that
 $\cW_\KG$ is complete in the topology given by \eqref{alpha1}, and
 that $\sWsc$ is dense in $\cW_\KG$. Clearly, not all elements of
 $\cW_\KG$ are space-compact, but they decay at an appropriate rate
 in spatial directions.

Mathematically, $\cW_\KG$ has the structure of a {\em Krein space}.
A decomposition \eqref{decompo} is an example of a {\em fundamental
  decomposition}, see Prop. \ref{kaku3}. We refer to Appendix
  \ref{app:Involutions_and_projections} for a discussion of Krein spaces.

Let $\Pi_\alpha^{(\pm)}$ be the orthogonal projections onto
$\cZ_\alpha^{(\pm)}$. Denoting by $\cN$ the nullspace and by $\cR$ the
range, we thus have
\begin{align}
  \cN(\Pi_{\alpha}^{(\pm)})=\cZ_\alpha^{(\mp)},&\quad
                                                  \cR(\Pi_{\alpha}^{(\pm)})=\cZ_\alpha^{(\pm)},
                                                  \\ \notag
\big(\Pi^{(\pm)}_\alpha\big)^2= \Pi^{(\pm)}_\alpha,
 &\quad\Pi^{(\pm)}_\alpha
   \Pi^{(\mp)}_\alpha=0,\\ \notag
   \big(\Pi^{(\pm)}_\alpha\zeta|
  \Pi^{(\pm)}_\alpha\zeta)_\KG& \gtrless0, \quad \zeta\in\cW_\KG \\ \notag
  \big( \Pi^{(\pm)}_\alpha \zeta |\xi \big)_\KG
 & = \big(  \zeta |\Pi^{(\pm)}_\alpha \xi \big)_\KG,\quad \zeta,\xi\in\cW_\KG .
 \end{align}

It is important to note that there are many decompositions of the form
\eqref{alpha1} with properties as above leading to the same space
$\cW_\KG$. Note that $\cW_\sc$ is uniquely defined and the Klein-Gordon
charge $(\cdot|\cdot)_\KG$ on $\cW_\sc$ is also unique.
  However,  there can be more then one Krein structure extending
  $\big(\cW_\sc,(\cdot|\cdot)_\KG\big)$. Clearly, specifying a
  fundamental decomposition fixes such a structure. We expect that
  in typical spacetimes all physically
reasonable decompositions lead to the same $\cW_\KG$.

\subsection{Klein-Gordon kernels}
 If $\zeta\in\cW_\sc$, and $\xi$ is a distributional solution of
 the Klein-Gordon equation which is not necessarily space compact,
 then the current $J_\mu[\zeta,\xi]$ defined by \eqref{current0} is a
 well-defined space-compact distribution that satisfies the relation
\eqref{current}. However, for such $\zeta,\xi$ the quantity
\eqref{eq:KG_charge_curved} is problematic, because it is not clear
what it means to integrate a distribution on a surface $\Sigma$.
Let us describe an appropriate replacement of \eqref{eq:KG_charge_curved}
in that case.

We first choose a time variable $x^0$. Let then $j\in C^\infty(\RR)$
be a real function such that $j(t)=0$ for $t<-\frac12$ and $j(t)=1$,
$t>\frac12$. We have
\begin{align} \label{current1}
\nabla_\mu \big(J^\mu(x) j(x^0)\big)= J^0(x)j'(x^0).
\end{align}
Therefore,
\begin{align}
\label{eq:KG_charge_curved-}
 (\zeta|\xi)_\KG := \ii \int \Big(\overline{\zeta(x)} \lrnb_{\!x^0} \xi(x)\Big)
 j'(x^0) |g(x)|^{\frac12}\dd x
\end{align}
is well-defined for distributional $\xi$ and does not depend
on the choice of the time variable $x^0$ and the function $j$.
If $\xi$ is regular enough, it coincides with
\eqref{eq:KG_charge_curved}.

Suppose now that $B(x,y)$ is a bisolution of the Klein-Gordon equation,
which is sufficiently regular, say $C^\infty$. Then it is easy to see
that for $\zeta_1,\zeta_2\in \cW_\sc$, the integral
\begin{align}\label{KGkernel}
\ii \int_{\Sigma_1}\int_{\Sigma_2}\zeta_1(x) \lrnb_{x^\mu } B(x,y) \lrnb_{y^\mu }\zeta_2(y)
  \dd\Sigma_1^\mu(x)   \dd\Sigma_2^\mu(y).
\end{align}
does not depend on the choice of Cauchy surfaces $\Sigma_1,\Sigma_2$
and defines a sesquilinear form on $\cW_\sc$. If this form is
bounded, then it defines a unique operator $B$ on $\cW_\KG$. We then
say that $B(\cdot,\cdot)$ is the {\em Klein-Gordon kernel} of the
operator $B$.

For distributional
Klein-Gordon kernels one needs a different definition:
\begin{definition}
  Let $B(\cdot,\cdot)$ be a (distributional) bisolution.
  Choose a time variable $x^0$ on $M$ (see
 Subsection \ref{Classical propagators from evolution equations}).
 Let $j_i$, $i=1,2$ be
two functions on $\RR$ such that $j_i(t)=0$ for $t<-\frac12$ and
$j_i(t)=1$, $t>\frac12$. Then for $\zeta_1,\zeta_2\in\cW_\sc$
\begin{align}\label{KGkernel.}
\ii \int\int j_1'(x^0)\Big(\zeta_1(x) \lrnb_{x^0 } B(x,y) \lrnb_{y^0 }\zeta_2(y)\Big)
 j_2'(y^0) |g(x)|^{\frac12}\dd x |g(y)|^{\frac12}  \dd y
\end{align}
does not depend on the choice of the time variable and the functions
$j_1,j_2$. It
defines a  bilinear form  on $\cW_\sc$. If this form corresponds
to a bounded operator $B$ on $\cW_\sc$,  then $B(\cdot,\cdot)$ will be
called the
{\em Klein-Gordon kernel of
the operator $B$. }\end{definition}

Note that if $B(\cdot,\cdot)$ is sufficiently regular, then
\eqref{KGkernel.}=\eqref{KGkernel}. An example of a
Klein-Gordon kernel that usually is of distributional nature
is the Pauli-Jordan bisolution $G^\PJ(x,y)$, which is the
Klein-Gordon kernel of the identity.

Let us describe the physical meaning of a fundamental decomposition
 \eqref{decompo}. Let $\pm G^{(\pm)}_\alpha(x,y)$ be the Klein-Gordon
 kernel of the projection $\Pi^{(\pm)}_\alpha$, so that the sum
$G^\sym_\alpha(x,y) := G^{(+)}_\alpha(x,y) + G^{(-)}_\alpha(x,y)$
is the Klein-Gordon kernel of the involution
$S_\alpha := \Pi^{(+)}_\alpha-\Pi^{(-)}_\alpha$.
Then there exists a Fock representation with the Fock vacuum
$\Omega_\alpha$ such that $G_\alpha^{(\pm)}(x,y)$ are the
corresponding  two-point functions
\begin{align}
 G_{\alpha}^{(+)}(x,y) &:= \big(\Omega_\alpha|\hat\phi(x)\hat\phi(y)\Omega_\alpha\big)
,\label{G1.}\\
G_{\alpha}^{(-)}(x,y) &:= \big(\Omega_\alpha|\hat\phi(y)\hat\phi(x)\Omega_\alpha\big).
\label{G2.}
 \end{align}

Let $\cW_\KG =\cZ_\beta^{(+)}\oplus \cZ_\beta^{(-)}$
be another orthogonal decomposition of the Krein space $\cW_\KG$ into
a maximal uniformly positive and maximal uniformly negative subspace,
defining the vacuum $\Omega_\beta$. One can show \cite{DS22}  (see also
Appendix \ref{app:Involutions_and_projections})  that the spaces
$\cZ_\beta^{(+)}$ and $\cZ_\alpha^{(-)}$ are complementary, so that
we have a (non-orthogonal) direct sum decomposition
\begin{align}
  \cW_\KG &=\cZ_\beta^{(+)}\oplus
           \cZ_\alpha^{(-)}.
\end{align}
Therefore, we can define projections $\Pi_{\alpha,\beta}^{(+)}$,
$\Pi_{\alpha,\beta}^{(-)}$ corresponding to this decomposition. They
satisfy
\begin{align}
  \cN(\Pi_{\alpha,\beta}^{(+)})
  =\cR(\Pi_{\alpha,\beta}^{(-)})
  &=\cZ_\beta^{(-)}, \\ \notag
 \cR(\Pi_{\alpha,\beta}^{(+)})
    =\cN(\Pi_{\alpha,\beta}^{(-)})
    &=\cZ_\alpha^{(+)},\\
    \Pi_{\alpha,\beta}^{(+)}+\Pi_{\alpha,\beta}^{(-)}&=\one.  \notag
\end{align}

We may also decompose $\cW_\KG$ the other way around,
\begin{align}
  \cW_\KG &=    \cZ_\alpha^{(+)}\oplus\cZ_\beta^{(-)}.
\end{align}
The corresponding projections are denoted $\Pi_{\beta,\alpha}^{(+)}$,
$\Pi_{\beta,\alpha}^{(-)}$. Then one finds
\begin{align}
 \big( \Pi^{(\pm)}_{\alpha,\beta} \zeta|\xi\big)_\KG &=
 \big(\zeta | \Pi^{(\pm)}_{\beta,\alpha} \xi\big)_\KG .
\end{align}
Thus, these projections are  orthogonal if and only if
$\cZ_\alpha^{(+)}=\cZ_\beta^{(+)}$.

Let $R$ be a bounded linear transformation on $\cW_\KG$. The
Klein-Gordon Hermitian conjugate $R^{*\KG}$ of $R$ is defined by
\begin{align}
(R^{*\KG}\zeta|\xi)_\KG := (\zeta|R\xi)_\KG,
\end{align}
and the complex conjugate $ \overline R$ by
\begin{align}
  \overline{R}\, \overline \zeta :=\overline{R\zeta}.
\end{align}
Linear  transformations that preserve the structure of $\cW_\KG$
are called {\em symplectic}, or (especially in the physics
literature) {\em Bogoliubov transformations}. Here, $R$
preserving the structure of $\cW_\KG$ means that $R$ is
{\em pseudounitary} and {\em real}, i.e.,
\begin{align}
(R\zeta|R\xi)_\KG=(\zeta|\xi)_\KG \word{and}
R\overline\zeta=\overline{R\zeta},
\end{align}
or in other words $R^{*\KG}=R^{-1}$ and $ \overline{R}=R$.

\subsection{Mode expansions}
\label{ssc:mode_expansions}
Many papers about QFT on curved spacetimes do not mention the word
``Krein space''.  Instead, they introduce a decomposition of solutions
to the Klein-Gordon equation into a ``positive frequency part''  and
a ``negative frequency part''. This is usually done through modes, by
assuming that the classical field can be written as
 \begin{align}\label{modes}
 \phi(x) &= \int \big(\overline{\varphi_{\alpha,k}(x)} a_{\alpha,k}
 + \varphi_{\alpha,k}(x)
 a^\ast_{\alpha,k}\big) \dd k,
 \end{align}
where the mode functions $\varphi_{\alpha,k}$ and $\overline{\varphi_{\alpha,k}}$
should satisfy
\begin{align}
\label{eq:modes_orthogonal}
\big(\varphi_{\alpha,k}|\varphi_{\alpha,k'}\big)_\KG
 = -\big(\overline{\varphi_{\alpha,k}}|\overline{\varphi_{\alpha,k'}}\big)_\KG
& = - \delta(k,k'),\quad \\ \notag
   \big(\varphi_{\alpha,k}|\overline{\varphi_{\alpha,k'}}\big)_\KG&=0, \\
   \notag
(-\square_g+Y(x))\varphi_{\alpha,k}(x)&=0.
 \end{align}
The variable $k$ (and the measure $\dd k$) may be continuous or discrete
(e.g. if the Cauchy surface is compact). In Minkowski space it usually
coincides with the $d-1$-momentum, which is not available on a generic
spacetime.

Let us try to interpret this in a more rigorous sense.
Let us assume that $K$ is  a measure set. For the sake of definiteness
we will assume that $K=\RR^n$ with the Lebesgue measure $\dd k$, but
this is not relevant. Suppose that for any ``wave packets'' $f\in
L^2(K)$, we can interpret
\beq \int \bar{\varphi_{\alpha,k} }f(k)\dd k
\word{and}
\int \varphi_{\alpha,k} f(k)\dd k
\eeq
in a rigorous sense as a solution to the Klein-Gordon equation.
The functions
\beq \int\bar{ \varphi_{\alpha,k}} f_1(k)\dd k+ \int
\varphi_{\alpha,k}\bar{ f_2(k)}\dd k\eeq with $f_1,f_2\in L^2(K)$
with the  sesquilinear form defined by \eqref{eq:modes_orthogonal}
form a Krein space.

Thus by introducing the modes satisfying
\eqref{eq:modes_orthogonal} we automatically fix a Krein space.
The positive/negative frequency modes span
maximal uniformly positive/negative subspaces of this Krein space in
the sense of "wave packets"--so we have a distinguished fundamental
decomposition (and the corresponding Fock representation).
Hence the idea of a Krein spaces is introduced in many papers
``through the back door''.

After quantization, we obtain
\begin{align} \label{eq:Qphi_mode}
\hat\phi(x) &= \int \big( \overline{\varphi_{\alpha,k}(x)}\hat  a_{\alpha,k}
 + \varphi_{\alpha,k}(x)
 \hat a^\ast_{\alpha,k} \big)\dd k ,\\ \notag
   [\hat a_{\alpha,k},\hat a^\ast_{\alpha,k'}]& = \delta(k,k'),\quad
  [\hat a_{\alpha,k},\hat a_{\alpha,k'}] = 0.  \end{align}
  Then
  \begin{align}
\label{eq:G_alpha_def}
 G^{(+)}_{\alpha}(x,y)
 =&\, \int
 \overline{\varphi_{\alpha,k}(x)} \varphi_{\alpha,k}(y)  \dd k ,
 \\ \notag
  G^{(-)}_{\alpha}(x,y)
 =&\,
 \int
\varphi_{\alpha,k}(x)  \overline{\varphi_{\alpha,k}(y)} \dd k .
\end{align}

Mode decompositions are convenient, because they allow us to represent
operators on the Krein space in terms of integral kernels on $K\times K$.
As an illustration, assume that we have  another state $\Omega_\beta$
with a decomposition analogous to \eqref{eq:Qphi_mode}, whose modes
generate the same Krein space $\cW_\KG$ and use the same measure space $K$:
 \begin{align}\label{modes.}
 \phi(x) &= \int \big(\overline{\varphi_{\beta,k}(x)} a_{\beta,k}
 + \varphi_{\beta,k}(x)
 a^\ast_{\beta,k}\big) \dd k.
 \end{align}
Assume further, that the two decompositions are related by a Bogoliubov
transformation
\begin{align}
\label{eq:Bogo}
 \varphi_{\beta,k}(x)
&= N(k)
\varphi_{\alpha,k}(x) + \int \Lambda(k,k')
                        \overline{\varphi_{\alpha,k'}(x)}\dd k'
                        .\end{align}
                                           The pseudounitarity of
                                           \eqref{eq:Bogo} is equivalent to
                      \begin{align}
\label{eq:assumptions_Lambda}
N(k') \Lambda(k,k')
&= \Lambda(k',k) N(k),\\ \notag
\int \overline{ \Lambda(k,p)}
\Lambda(k',p) \dd p
&=\big(|N(k)|^2-1\big) \delta(k-k').
\end{align}
Therefore, the transformation inverse to \eqref{eq:Bogo} is
\begin{align}
\label{eq:Bogo.}
 \varphi_{\alpha,k}(x)
&= \bar{N(k)}
\varphi_{\beta,k}(x) - \int \Lambda(k',k)
                        \overline{\varphi_{\beta,k'}(x)}\dd k'
                        .\end{align}
On the level of $a_{\alpha,k}$ and $a_{\beta,k}$, or their quantized
versions,  we have
\begin{align}
\hat a_{\beta,k}
&=
N (k)
 \hat  a_{\alpha,k} -
                   \int \Lambda(k',k) \hat  a_{\alpha,k'}^\ast \dd k',\\
  \hat a_{\alpha,k}
&=
\bar{N (k)}
 \hat  a_{\beta,k} +
\int \Lambda(k,k') \hat  a_{\beta,k'}^\ast \dd k'.
\end{align}

Thus, the fields can be written as
\begin{align}\hat\phi(x)
&=\int \Bigg( \varphi_{\alpha,k}(x)\hat a_{\alpha,k}^\ast
+ \frac{\bar{\varphi_{\alpha,k}(x)}}{N(k)} \int \Lambda(k',k) \hat a_{\alpha,k'}^\ast \dd k'
+\frac{\bar{\varphi_{\alpha,k}(x)}}{N(k)}\hat a_{\beta,k}
\Bigg) \dd k
 \label{kkl}\\
     &=\int \Bigg( \bar{\varphi_{\beta,k}(x)}\hat a_{\beta,k}
- \frac{\varphi_{\beta,k}(x)}{N(k)} \int \overline{\Lambda(k,k')}  \hat a_{\beta,k'} \dd k'
+\frac{\varphi_{\beta,k}(x)}{N(k)}\hat a_{\alpha,k} ^\ast
 \Bigg)\dd k.
 \label{kkr}
   \end{align}
We insert into \eqref{eq:zero_divided_by_zero} and \eqref{eq:GF0} the
expression  \eqref{kkl}  as the left field and  \eqref{kkr} as the right
field. Then, moving $\hat a_{\alpha,k}^*$ to the left and $\hat a_{\beta,k}$
to the right and using $[\hat a_{\beta,k},\hat a_{\alpha,k'}^\ast]=N(k)\delta(k-k')$,
we obtain
 \begin{align}
\label{eq:G_albe_def}
 G^{(+)}_{\alpha\beta}(x,y)
 =&\,  \int\frac{1}{N(k)}
 \overline{\varphi_{\alpha,k}(x)} \varphi_{\beta,k}(y)  \dd k ,
 \\ \notag
  G^{(-)}_{\alpha\beta}(x,y)
 =&\,
  \int \frac{1}{N(k)}
 \overline{\varphi_{\alpha,k}(y)} \varphi_{\beta,k}(x)  \dd k .
\end{align}

\subsection{Operator-theoretic (anti-)Feynman propagator}
\label{ssc:op_inv_special}
The d'Alembertian  $-\Box$ on a Lorentzian manifold
$M$ with pseudometric $g$, in the half-density formalism given by \eqref{half},
 is Hermitian (symmetric) on $C_\mathrm{c}^\infty(M)$ in the sense of
$L^2\big(M\big)$.
The same is true for the Klein-Gordon operator $-\Box+Y(x)$ with a
real  potential $Y$.
Assume that $-\Box+Y(x)$ is
 \emph{essentially
  self-adjoint} (if not, we may choose a self-adjoint extension).

Then its resolvent $G(z):=(-\square+Y(x) -z)^{-1}$ is well-defined for
$z\in\bC\setminus\bR$. It possesses an integral kernel $G(z;x,y)$.
Suppose that there exists
\begin{align}\label{point1}
  G^\F_{\rm op} (x,y):= \lim_{\epsilon\searrow0}G(+\ii\epsilon;x,y)
  ,\\\label{point2}
   G^{ \overline\F}_{\rm op} (x,y):=
  \lim_{\epsilon\searrow0}G(-\ii\epsilon;x,y) ,
\end{align}
where we use the distributional limit. The distributions
$ G^\F_{\rm op}(x,x')$ and $ G^{\overline{\F}}_{\rm op}(x,x')$ will be
called
the \emph{operator-theoretic Feynman and anti-Feynman propagator}.

We expect that the limits \eqref{point1} and \eqref{point2}
exist on most physically interesting globally hyperbolic manifolds
without boundaries. They will not exist at the point spectrum of
$-\Box+Y(x)$ (which  is probably quite rare).

We believe that the following argument justifies this
         definition. Here is an elementary fact about {\em Fresnel
           integrals}. Let $c$ be a real symmetric $n\times n$ matrix,  $u$ a
         variable in $\mathbb{R}^n$ and $J\in\mathbb{R}^n$. Then
\beq\frac{\int\e^{\pm\ii(u^{\TT }\frac{c}{2}u+J^\TT u)}\dd u}{\int\e^{\pm\ii u^\TT\frac{c}{2}u}\dd
    u}=\exp\Big(\mp \frac{\ii}{2} J^\TT (c\pm\ii0)^{-1}J\Big).\label{fresnel}\eeq

                  If we use   {  \em  path integrals} to construct
a                  quantum field theory,  we usually start from
                  defining formally the
                   generating function   as
                   \[Z(J):=\frac{\int\e^{\ii S(\phi)+\ii\int\phi(x)
                         J(x)\dd x}\cD\phi}
                     {\int\e^{\ii S(\phi)}\cD\phi}.
                   \]
                   If the action is {\em quadratic}
                   \begin{align*}S(\phi)=&-\frac12\int\big(g^{\mu\nu}(x)\partial_\mu\phi(x)\partial_\nu\phi(x)+m^2\phi(x)^2\big)\sqrt{|g|}(x)
                     \dd
                                                  x\\
                     =&-\frac12\big(\phi|(-\Box+m^2)\phi\big),\end{align*}
                   then  the path
                   integral by analogy to \eqref{fresnel} can be {\em rigorously defined} as
                   \begin{align*}Z(J)&=
                                       \exp{\frac{\ii}{2}\big(J|(-\Box+m^2-\ii0)^{-1}J\big)}
                                       \\&=\exp\Big(\frac{\ii}{2}\int\int
                     J(x)G_\op^\F(x,y)J(y)\sqrt{|g|}(x)\sqrt{|g|}(y)\dd
                     x\dd y\Big).
                     \end{align*}

Spacetimes where the d'Alembertian is essentially self-adjoint
include stationary spacetimes, Friedmann-Lemaître-Robertson-Walker
(FLRW) spacetimes, $1+0$-dimensional spacetimes, de Sitter and the
universal cover of the
anti-de Sitter space. Essential self-adjointness was also recently
proven on a class of asymptotically Minkowski spacetimes
\cite{V20,NT23}.  However, even on well-behaved
spacetimes, essential self-adjointness is not always true
\cite{kaminski}.

\begin{remark}
  Essential self-adjointness is typically destroyed if
  there are  boundaries. The problem with  boundaries with
  spacelike normal vectors can sometimes be cured by
  imposing boundary conditions---we will see this in Section
  \ref{sec:AdS} about the universal cover of the anti-de Sitter space.
  Boundaries with timelike normal vectors are
different. In particular, if the time is confined to an interval
$]a,b[$ instead of $\mathbb{R}$, then self-adjoint realizations
  of the Klein-Gordon operators  do not lead to physically justified
  Feynman propagators. Instead, one should consider other types of
non-self-adjoint boundary conditions, as explained in \cite{DS22}.
\end{remark}

\begin{remark} There exists a well-developed theory of
limits of the resolvents of the Schr\"odinger operators
$H :=-\Delta+V(x)$ on $\bR^d$. More precisely, if $V$
is a decaying potential on $\RR^d$ satisfying
    appropriate conditions, $E\neq0$ is away from the spectrum
of $H$
    and $s>\frac12$, then the
    following strong operator limit exists:
    \beq\label{LAP}
\lim_{\epsilon\searrow0}
(1+|x|)^{-s}\big(H-E\pm\ii\epsilon\big)^{-1}(1+|x|)^{-s}.\eeq
The existence of \eqref{LAP} goes under the name of Limiting
Absorption Principle. The most powerful method to prove this
is the
so-called Mourre theory, and is treated e.g. in
\cite{CFKS,ABG}. Obviously, the
Limiting
Absorption Principle implies the existence of  the integral kernel of
$\big(H-E\pm\ii0\big)^{-1}$.

One can try to apply similar methods to  Klein-Gordon operators, as
shown in \cite{V20,NT23}. \end{remark}

\subsection{Special Klein-Gordon equations}
\label{Special spacetimes}

\begin{defn}\label{def-special}
  Suppose that the Klein-Gordon operator $-\square+Y(x)$
on a Lorentzian manifold $M$ is essentially self-adjoint.
 We say that $-\square+Y(x)$ is \emph{special} if the sum
 \begin{align}
 \label{eq:sum_special}
  G^\F_{\rm op}(x,x')+G^{\overline{\F}}_{\rm op}(x,x')
 \end{align}
 has causal support.  \end{defn}

\begin{defn} Suppose that the Klein-Gordon operator $-\square+Y(x)$
is essentially self-adjoint and
 $M$ is globally hyperbolic. We say that it is {\em strongly special} if
 \begin{align}
  G^\F_{\rm op}(x,x')+G^{\overline{\F}}_{\rm op}(x,x')
  = G^\lor(x,x') + G^\land(x,x').
 \end{align}
\end{defn}

Clearly, strong specialty implies specialty. We expect that under
broad conditions the converse is also true.

\emph{Special} Klein-Gordon operators are interesting because the
associated propagators can be determined in an easy way.
Indeed, it is often not very difficult to  compute
$ G^\F_{\rm op}(x,x')$ and $G^{\overline{\F}}_{\rm op}(x,x')$. After
all, there are various techniques to compute the kernel of the
resolvent of a differential operator.
From these, one can determine the retarded and advanced
propagators  by
\begin{align}\label{bsik}
 G^{\lor/\land}(x,x') = \theta\big(\pm(x^0-{x'}^0)\big)
 \Big(  G^\F_{\rm op}(x,x')+G^{\overline{\F}}_{\rm op}(x,x') \Big)
\end{align}
as well as the Pauli-Jordan function $G^\PJ = G^{\lor}-
G^{\land}$.

Strictly speaking,  \eqref{bsik} is not fully legal,
because it involves multiplying a distribution
by a discontinuous function $\theta\big(\pm(x^0-{x'}^0)\big) $. In
practice, we expect that this obstacle can be overcome, see \cite{DeSi}.
 In particular, there is no problem with the
  multiplication with the theta function when we can
  apply the method of evolution equations, see Subsect.
\ref{Classical propagators from evolution equations}.

More interestingly, there is a natural candidate for
the Wightman and anti-Wightman two-point function
  of a distinguished state:
\begin{align}
\label{eq:Gpm_special}
 G^{(\pm)}:= -\ii \big( G^\F_{\rm op} - G^{\land/\lor}\big)
  = \ii \big( G^{\overline{\F}}_{\rm op} - G^{\lor/\land}\big) .
\end{align}

\begin{remark}
Actually, we do not know if $ G^{(\pm)}$
defined by \eqref{eq:Gpm_special} in the case when
$-\square+Y(x)$ is special always satisfy the positivity requirement --- in all cases that we worked out they do.
\end{remark}


\section{Stationary and asymptotically stationary spaces}
\label{sec:asymp_stat}

\subsection{Propagators on stationary spacetimes}

Assume that $M=\mathbb{R}\times\Sigma$, with the variables typically
denoted by $(t,{\bf x})$, and sometimes $(s,{\bf y})$. Suppose that
neither $g_{\alpha\beta}$ nor $Y$ depends on time $t$,  the time
slices $\{t\}\times\Sigma$   are spacelike and $\partial_t$ is
timelike. Such spacetimes are called {\em stationary}.

In addition,  we will assume that the spacetime is {\em static}, i.e.
there are no time-position cross-terms.
This is not a necessary condition for the present analysis, however
for static spacetimes many formulas are more explicit. In
other words, we assume that the metric is
\begin{align}
-\alpha^2(\mathbf{x})\dd t^2+h_{ij}(\mathbf{x}) \dd\mathbf{x}^i
\dd\mathbf{x}^j.
\end{align}

Here and in the following, Latin indices run
over the spatial directions. We write $|h|$ for $\det h$.
The Klein-Gordon operator  in the half-density formalism is
\begin{align}
-\Box+Y(\mathbf{x})= \frac{1}{\alpha^2}
\del_t^2 -\alpha^{-\frac12} \,| h|^{-\frac14}
\del_i \alpha \,| h|^{\frac12} h^{ij} \del_j \,\alpha^{-\frac12} \,| h|^{-\frac14}
+Y.
\label{kgor}
\end{align}
It is convenient to replace \eqref{kgor} by
\begin{align}
  -\tilde\Box+\tilde Y:=\alpha\big(-\Box
  + Y\big)
  \alpha =\partial_t^2+ L,
\end{align}
where
\begin{align}
  L&:=-\Delta_{\tilde h}+\tilde Y,
 \quad  \Delta_{\tilde h}
 :=\gamma^{-\frac12}\partial_i \gamma\tilde h^{ij}\partial_j \gamma^{-\frac12},\\ \notag
  \tilde h_{ij}&:= \frac{1}{\alpha^2} h_{ij},\quad [\tilde h^{ij}]=[\tilde h_{ij}]^{-1},\quad \gamma := \frac{|h|^{\frac12}}{\alpha},
  \quad
  \tilde Y = \alpha^2 Y.
  \end{align}

  Note that
  \begin{align}
\text{ if  $\tilde{u}$ solves }
    \label{tildkg}
(-\tilde \Box+\tilde Y) \tilde{u}=0,
\text{ then $ u := \alpha \tilde{u}$ solves }
(- \Box+ Y)u=0.
\end{align}

Let us first describe the approach based on the evolution of Cauchy
data, which is particularly simple for static spacetimes. The equation \eqref{tildkg}
for $
\tilde u(t)= \tilde u(t,\mathbf{x})$
can be rewritten as a 1st order equation for the Cauchy data
\begin{align}
  \bigl(\partial_t+\ii  B\bigr) w&=0,\\
  B
    :=
           \begin{bmatrix}
 0 & \one \\ L & 0
           \end{bmatrix}
  ,                         &\qquad  w=
  \begin{bmatrix}
   w_1 \\  w_2
  \end{bmatrix}
    :=
           \begin{bmatrix}
            \tilde u \\ \ii\partial_t\tilde  u
           \end{bmatrix}
  .
\end{align}
 Assume that $L$ is positive and self-adjoint in the sense of $L^2(\Sigma)$.
We assume that $0$ is not an eigenvalue of $L$ and
endow the space of Cauchy data with the (positive) scalar product
\begin{align}\label{scalar}
(w|v)_0 :=(w_1|\sqrt{ L}v_1)+\big(w_2\big|\tfrac{1}{\sqrt{L}}v_2\big).
\end{align}
The completion of $\cW_\sc$ with respect to this scalar product will be denoted $\cW_0$.
Note that $B$ can be interpreted as self-adjoint with respect to this scalar product:
\begin{align}
(Bw|v)_0 =(w|Bv)_0 =(w_2|\sqrt{ L}v_1)+(w_1|\sqrt{
  L}v_2).\end{align}
The space $\cW_0$ is also endowed with the (indefinite)
Klein-Gordon
charge form
\begin{align}
(w|v)_\KG = (w|Qv)_{0} :=(w_1|v_2)+(w_2|v_1),\quad
Q=\begin{bmatrix}
   0 &\one \\ \one & 0
  \end{bmatrix}
.\label{kgb}
\end{align}
Thus $\cW_0$ is a Krein space with a distinguished  Hilbert space structure.

We can compute the evolution operator and the spectral projection of
$B$ onto the positive and negative part of the spectrum:
\begin{align}
\e^{-\ii tB}:=&\begin{bmatrix}\cos t\sqrt{ L}&-\ii\frac{\sin
    t\sqrt{ L}}{\sqrt{ L}}\\-\ii\sqrt{ L}\sin
    t\sqrt{ L}&\cos t\sqrt{ L}\end{bmatrix},\\
  \Pi^{(\pm)}:=&\,\one_{\mathbb{R}_+}(\pm B)=\frac12\begin{bmatrix}\one&\pm\frac{\one}{\sqrt{ L}}\\\pm\sqrt{ L}&\one\end{bmatrix}.
\label{funda}\end{align}
Note that the evolution $\e^{-\ii tB}$ preserves the Klein-Gordon
charge form  \eqref{kgb}.  $\cR(\Pi^{(\pm)})$ are maximal uniformly
positive/negative subspaces with respect to the Klein-Gordon charge form.
Then we can define the propagators on the level of the Cauchy data as follows:
\begin{align*}
   E^\PJ(t,s)    & := \e^{-\ii(t-s) B}, \\
   E^{\vee/\wedge}(t,s)   & := \pm\theta\big(\pm(t-s)\big)
   \e^{-\ii(t-s) B}, \\
   E^{(\pm)}(t,s)  & := \e^{-\ii(t-s) B}\Pi^{(\pm)}, \\
   E^{\F/\overline{\F}}(t,s)  & :=  \e^{-\ii(t-s) B}\bigl(\theta(t-s)\Pi^{(\pm)} -\theta(s-t)\Pi^{(\mp)}\bigr).
\end{align*}
At least formally, $ E^\vee,  E^\wedge,  E^\F,
 E^{ \overline\F}$ are inverses and $ E^\PJ,  E^{(+)},
 E^{(-)}$ are bisolutions of $\partial_t+\ii  B$.
They are $2\times 2$ matrices:
\begin{align*}
   E^\bullet(t,s)=
          \begin{bmatrix}
             E^\bullet_{11}(t,s) &  E^\bullet_{12}(t,s) \\ E^\bullet_{21}(t,s) &  E^\bullet_{22}(t,s)
          \end{bmatrix}
  .
\end{align*}
We set
\begin{align}
  G^\bullet & := \ii \alpha E_{12}^\bullet\alpha,\quad \bullet=\PJ,\vee,\wedge, \F,
  \overline{\F}, \\ \notag
  G^{(\pm)} & := \pm \alpha E_{12}^{(\pm)}\alpha,
\end{align}
obtaining propagators for a general static stable case.
Thus
\begin{align}\label{pro1}
  G^\PJ(t,
  \mathbf{x};s,\mathbf{y})& =\alpha(\mathbf{x})\frac{\sin(t-s)\sqrt{
  L}}{\sqrt{ L}} (
                            \mathbf{x},\mathbf{y})\alpha(\mathbf{y}),\\\label{pro2}
    G^{\vee/\wedge}(t,
  \mathbf{x};s,\mathbf{y})&=\pm\theta\big(\pm(t-s)\big)\alpha(\mathbf{x})\frac{\sin(t-s)\sqrt{
  L}}{\sqrt{ L}} (
                            \mathbf{x},\mathbf{y})\alpha(\mathbf{y}),
\\
   G^{(\pm)}(t,
  \mathbf{x};s,\mathbf{y})&=\alpha(\mathbf{x})
  \frac{\e^{\mp\ii(t-s)\sqrt{
 L}}}{2\sqrt{ L}} (
                            \mathbf{x},\mathbf{y})\alpha(\mathbf{y}),
                            \\
   G^{\F/\overline{\F}}(t,   \mathbf{x};s,\mathbf{y})
  & =\pm\ii\alpha(\mathbf{x})
  \Bigg(\theta(t-s)\frac{\e^{\mp\ii(t-s)\sqrt{
 L}}}{2\sqrt{ L}} +
                            \theta(s-t)\frac{\e^{\pm\ii(t-s)\sqrt{
 L}}}{2\sqrt{ L}}\Bigg) (
                            \mathbf{x},\mathbf{y})\alpha(\mathbf{y})\label{wewe1}.
\end{align}

Note that all the identities~\eqref{eq:relations_props}  hold, where
  we tacitly assume that $\Omega_\alpha=\Omega_\beta$ corresponds to the natural
  state given by the fundamental decomposition
  \eqref{funda}.
In this setting, the Wightman function is often called the {\em
  positive frequency bisolution}, and the anti-Wightman function the
{\em negative frequency bisolution}, since they are obtained from the
spectral decomposition of the generator of the dynamics  into
positive and negative frequencies.

Note also that the specialty condition is true:
\begin{align}\label{specialty}
G^\F + G^{\overline{\F}} &= G^{\lor} + G^{\land}.
\end{align}

Let us describe now the approach based on the Hilbert space $L^2(M)$.
We assume that $ L$ is essentially self-adjoint on
$C_\mathrm{c}^\infty(\Sigma)$ in the sense of $L^2(\Sigma
)$.
Then it is easy to see
$\partial_t^2+ L$ is essentially self-adjoint on
$C_\mathrm{c}^\infty(M)$ in the sense of $L^2(M
)$.
Assume that for some $0<c,C$ we have $c\leq\alpha(\mathbf{x})\leq C$.
Then $\alpha(\mathbf{x})$ is  bounded invertible operator on $L^2(M)$, and
using this we can show that
$-\Box+Y(\mathbf{x})$ is essentially self-adjoint on
$\alpha(\mathbf{x})C_\mathrm{c}^\infty(M)$.
As proven in
\cite{DS18}, under some minor additional technical conditions
we can then define $G_\op^\F$ and $G_\op^{ \overline\F}$, and they coincide
with $G^\F$ and $G^{ \overline\F}$  computed from the evolution equation.

Note that the stability condition $ L\geq0$ was an important ingredient of the  analysis based on the evolution equation. Suppose
now
that $ L$  is
 not positive, but
only self-adjoint, which can be called the {\em tachyonic case}.
In the tachyonic case, we do not have the distinguished
  scalar product \eqref{scalar}. The
  formulas  \eqref{pro1} and
\eqref{pro2} for the classical propagators
$G^\PJ,G^\vee,G^\wedge$ are still true.
However, the evolution approach does not allow us to define
$G^{(\pm)}$,  $G^\F$ or $G^{ \overline\F}$.
The operator-theoretic $G_\op^\F$ or $G_\op^{ \overline\F}$, defined
as usual by \eqref{point1} and \eqref{point2}, will often exist.
However, they will not be given by the formula \eqref{wewe1}.
The specialty condition \eqref{specialty} is no longer true
in the tachyonic case.

For instance, in the  Minkowski space,  with $Y(x)=m^2<0$,
$G^\F_{\rm op}$ and $ G^{\overline{\F}}_{\rm op}$ are
well-behaved tempered distributions while the forward and backward
propagators have exponential growth as $t\to\pm\infty$ inside
  the forward, resp. backward cone.
A detailed discussion can for example be found in \cite{DeSi}.

\subsection{Classical propagators from evolution equations}
\label{Classical propagators from evolution equations}

Let us now consider a generic (not necessarily stationary) globally
hyperbolic spacetime $M$.
In order to compute non-classical (actually, also classical) propagators, it is useful to convert the Klein--Gordon equation into a 1st order evolution equation on the phase space describing Cauchy data.
To this end, we fix a decomposition
$M=]t_-,t_+[
\times\Sigma$, where $-\infty\leq t_-<t_+\leq+\infty$.
We assume that $\{t\}\times\Sigma$ is Riemannian for all
$t\in ]t_-,t_+[$
and $\partial_t$ is always timelike.
We will use Latin letters for spatial indices.
We introduce
\begin{align*}
  h=&\,[h_{ij}]=[g_{ij}],  \quad h^{-1}=[h^{ij}],
  \\ \beta^j :=&\,
  g_{0i}h^{ij},  \quad \alpha^2 := g_{0i}h^{ij}g_{j0}-g_{00},
  \\ |h|=&\,|\det h|=\det h, \quad |g|=|\det g|.
\end{align*}
We assume that $[h_{ij}]$, and hence also $[h^{ij}]$, are positive definite and
$\alpha^2>0$.
In coordinates, the metric can be written as
\begin{align}\label{ADM}
  g_{\mu\nu} \dd x^\mu \dd x^\nu & = -\alpha^2 \dd t^2 + h_{ij} (\dd x^i + \beta^i \dd t) (\dd x^j + \beta^j \dd t),
\end{align}
We have  $|g|=\alpha^2|h|$.  \eqref{ADM} is
  sometimes called  a {\em metric in the ADM form}, for Arnowitt, Deser and
  Misner \cite{BN}.

 $Y$ is a real valued function on $M$. To be on the safe side,
  let us assume
assume that $Y$ is smooth (which actually is not needed for the
existence of various propagators).
The Klein--Gordon operator in the half-density formalism  can now be written
\begin{align}
  \notag -\Box+Y(x)  & =
  |g|^{-\frac14} (\partial_t- \partial_i \beta^i)
  \frac{|g|^\frac12}{\alpha^2} (\partial_t-\beta^j
 \partial_j)|g|^{-\frac14} \\\notag & \quad - |g|^{-\frac14}
  \partial_i  |g|^\frac12 h^{ij} \partial_j |g|^{-\frac14}
  + Y.
\end{align}
Instead of the operator $\Box$ on $L^2(M)$, it is more convenient to work with the operator
\begin{equation*}
  \tilde\Box := \alpha \Box \alpha.
\end{equation*}

We have
\begin{align*}
  -\tilde\Box+ \alpha^2Y
  &= \gamma^{-\frac12} (\partial_t -\partial_i \beta^i) \gamma
    (\partial_t -\beta^j \partial_j ) \gamma^{-\frac12} \\&\quad -
  \gamma^{-\frac12} \partial_i\alpha^2 \gamma h^{ij} \partial_j \gamma^{-\frac12} + \alpha^2 Y \\
  &
= (\partial_t + \ii W^*) (\partial_t +\ii W) + L,
\end{align*}
where we introduced
\begin{align*}
  \gamma &:= \alpha^{-2} |g|^\frac12 = \alpha^{-1} |h|^\frac12, \\
  W &:=  \frac{\ii}{2} \gamma^{-1} \gamma\!,_{t} + \ii \gamma^{\frac12} \beta^i \del_i \gamma^{-\frac12}, \\
  L &:= -\partial_i^{\gamma\, *} \tilde{h}^{ij} \partial_j^{\gamma} + \tilde{Y},
\end{align*}
and we use the shorthands
\begin{align*}
  \tilde{h}^{ij} := \alpha^2 h^{ij}, \quad
  \tilde{Y} := \alpha^2 Y, \quad
  \partial_i^{\gamma} := \gamma^\frac12 \partial_i \gamma^{-\frac12},
  \quad \gamma\!,_{t} := \partial_t\gamma.
\end{align*}
Clearly, \eqref{tildkg} is still true, and hence propagators for $\tilde\Box$
induce corresponding propagators for $\Box$.

For each $t \in \mathbb{R}$, we  define
\begin{equation}
  B(t) =
        \begin{bmatrix}
          B_{11}(t) & B_{12}(t) \\B_{21}(t) & B_{22}(t)
        \end{bmatrix}
  := \begin{bmatrix} W(t) & \one \\ L(t) & W(t)^* \end{bmatrix}.\label{sca2}
\end{equation}
Setting $w_1(t) = \tilde u(t)$ and $w_2(t) = (\ii\partial_t - W(t))\tilde u(t)$,
we find that
\begin{equation}\label{sca3}
  \bigl( \partial_t + \ii B(t) \bigr) \begin{bmatrix} w_1(t) \\ w_2(t) \end{bmatrix} = 0
\end{equation}
if and only if $\tilde u$ is a solution of the Klein-Gordon equation
$\tilde\Box\tilde u = 0$. Therefore we occasionally call $\partial_t + \ii B(t)$
the \emph{first-order Klein--Gordon operator}. The half-densities $w_1(t)$ and
$w_2(t)$ may be called the \emph{Cauchy data} for $\tilde{u}$ at time $t$.

It is natural to introduce the \emph{classical Hamiltonian}
\begin{align*}
  H(t) & = QB(t) =
  \begin{bmatrix}
    L(t) & W^*(t) \\ W(t) & \one
  \end{bmatrix},\qquad Q=
  \begin{bmatrix}
    0 & \one \\\one & 0
  \end{bmatrix}
  .
\end{align*}
The operator $L(t)$ is a Hermitian operator on
$C_\mathrm{c}^\infty(\Sigma)$ in the sense of the Hilbert space
$L^2(\Sigma)$ and $H(t)$ is
 Hermitian  on
$C_\mathrm{c}^\infty(\Sigma)\oplus C_\mathrm{c}^\infty(\Sigma)$ in the sense of
$L^2(\Sigma)\oplus L^2(\Sigma)$.

We would like to apply the theory of non-autonomous evolution
equations due mostly to Kato and described in detail in Appendix C of
\cite{DS19}.
The space of Cauchy data
$C_\mathrm{c}^\infty(\Sigma)\oplus C_\mathrm{c}^\infty(\Sigma)$  has a
natural indefinite Klein-Gordon form
\begin{align}
(w|v)_\KG = (w|Qv) :=(w_1|v_2)+(w_2|v_1).
\label{kgb.}
\end{align}
However, we will need also a positive scalar product.
To this end we fix an appropriate  positive operator $L$ on $L^2(\Sigma)$ with $\cN(L)=0$  and
introduce the scalar product on  the Cauchy data
\begin{align}\label{sca1}
(w|v)_L :=(w_1|\sqrt{ L}v_1)+\big(w_2\big|\tfrac{1}{\sqrt{L}}v_2\big).
\end{align}
The Krein space given by the completion of
$C_\mathrm{c}^\infty(\Sigma)\oplus C_\mathrm{c}^\infty(\Sigma)$ in the
scalar product \eqref{sca1} will be denoted $\cW_\KG$.

The choice of the operator $L$ should be adapted to the family of
operators $B(t)$.
In typical situations, the operator $L(t)$ that appears in \eqref{sca2}
is positive at least  for some $t_0\in]t_-,t_+[$, and then one could to take
$L:=L(t_0)$. In any case, as we know, there exists a considerable freedom
of choosing $L$. A detailed discussion of conditions that one should
impose on $[g_{ij}]$ and $Y$ is contained in Section 2 and Appendix B
of \cite{DS19}, and also in \cite{DS22}. Under these conditions,
the evolution equation leads to a dynamics $R(t,s)$, which is
a two-parameter family of bounded operators on $\cW_\KG$ satisfying
\begin{align} R(t,t)&=\one, \notag \\
R(t,u)&=R(t,s)R(s,u), \notag\\
  \bigl(\partial_t+\ii B(t)\bigr)R(t,s)&=0,\notag\\
 \partial_s R(t,s)  - \ii R(t,s) B(t)&=0.
                                             \end{align}

The dynamics is a $2\times2$ matrix of operators acting on functions
on $\Sigma$:
\begin{align}
  R(t,s)=
        \begin{bmatrix}
          R_{11}(t,s) & R_{12}(t,s) \\R_{21}(t,s) & R_{22}(t,s)
        \end{bmatrix}
\end{align}
 with distributional kernels $R_{ij}(t,\mathbf{x};s,\mathbf{y})$.
 The classical propagators in the Cauchy data formalism are:
 \begin{align}
E^\PJ(t,
  \mathbf{x};s,\mathbf{y})&=\ii R(t,
                            \mathbf{x};s,\mathbf{y}),\\\label{unprob1}
 E^\vee(t,
  \mathbf{x};s,\mathbf{y})&=\ii\theta(t-s) R(t,
                            \mathbf{x};s,\mathbf{y}),\\\label{unprob2}
  E^\wedge(t,
  \mathbf{x};s,\mathbf{y})&=-\ii\theta(s-t) R(t,
                            \mathbf{x};s,\mathbf{y}).
                            \end{align}

  The usual classical propagators  well-known from the
literature, e.g. constructed in \cite{bar}, can be obtained by setting
\begin{align}
  G^\bullet(t,
  \mathbf{x};s,\mathbf{y})&=\ii \alpha(t,\mathbf{x})E_{12}^\bullet
        (t ,                   \mathbf{x};s,\mathbf{y})\alpha(s,\mathbf{y})
,\qquad \bullet=\PJ,\vee,\wedge.
                            \end{align}
 Note that they are usually
constructed by other mehods (e.g. by a local construction involving an
expansion in terms of Riesz potentials\cite{bar}). The
construction presented above has a considerable advantage over other
methods: it allows
for  low regularity of the metric and scalar potential. Its
apparent disadvantage is the need to impose assumptions global in the
slice $\Sigma$. However, using the finite speed of propagation, this
problem can be easily bypassed.

Note that the multiplication with step functions in the evolution
equation approach in \eqref{unprob1} and \eqref{unprob2}, and also
later in \eqref{proper3} and \eqref{proper4}, is unproblematic. In
fact, $R(t,s)$ is  an operator-valued function which is {\em strongly
continuous} in $t,s$. Therefore,  $R(t,\mathbf{x};s,\mathbf{y})$ is a
{\em continuous} function with values in distributions. Multiplication
of  $R(t,\mathbf{x};s,\mathbf{y})$ by $\theta(t-s)$ or $\theta(s-t)$
yields then a distribution.

\subsection{Non-classical propagators
on asymptotically stationary  spacetimes}
Assume now that $]t_-,t_+[=\RR$, so that the spacetime is
  $\RR\times\Sigma$ and the
  Klein-Gordon equation is
(a) asymptotically stationary and (b) asymptotically stable, i.e.,
\begin{align}
  &\text{(a) the strong resolvent limits
  } \lim_{t\pm\infty}B(t)=:B_\pm\quad\text{exist};
  \\
  &\text{(b) } H_\pm:=QB_\pm\geq0\quad \text{in the sense of  $L^2(\Sigma)\oplus L^2(\Sigma)$.}
\end{align}
  Assume that $0$ is not an eigenvalue of $B_+$ and $B_-$. Define the
  ``out/in particle/antiparticle projections'':
\begin{align}\label{sca4}
  \Pi_\pm^{(+)}:=&\one_{]0,\infty[}( B_\pm),\\\label{sca5}
  \Pi_\pm^{(-)}:=&\one_{]0,\infty[}(-B_\pm).\end{align}
As in Subsection \ref{Classical propagators from evolution equations},
we consider the Klein-Gordon form
\eqref{kgb.}. We choose the Krein structure so that both
$S_+:=\Pi_+^{(+)}-\Pi_+^{(-)}$ and $S_-:=\Pi_-^{(+)}-\Pi_-^{(-)}$ are
admissible involutions.

We can transport \eqref{sca4} and \eqref{sca5} by the evolution to any time $t$:
\begin{align}
  \Pi_\pm^{(+)}(t):=&\lim_{s\to\pm\infty}R(t,s)\Pi_\pm^{(+)}R(s,t),\\
  \Pi_\pm^{(-)}(t):=&\lim_{s\to\pm\infty}R(t,s)\Pi_\pm^{(-)}R(s,t).\end{align}
We can now define the ``out/in positive/negative frequency
bisolutions in the Cauchy data formalism'':
\begin{align}
  E_\pm^{(+)}(t,s)=  \Pi_\pm^{(+)}(t)R(t,s),\\
   E_\pm^{(-)}(t,s)=  \Pi_\pm^{(-)}(t)R(t,s).\end{align}
 Note that $\cR(\Pi^{(\pm)})$ and $\cR(\Pi^{(\pm)}(t))$ are maximal uniformly
positive/negative subspaces of the Krein space $\cW_\KG$.

We will need also projections $\Pi_{+-}^{(\pm)}(t)$ and
  $\Pi_{-+}^{(\pm)}(t)$ defined by specifying their range and
  nullspace:\footnote{Note that our notation is
 different from the convention in \cite{DS19,DS22}.}
    \begin{align}
\cR\big(\Pi_{+-}^{(+)}(t)\big)
&=\cN\big(\Pi_{+-}^{(-)}(t)\big)
=\cR\big(\Pi_+^{(+)}(t)\big),
\\ \notag
\cR\big(\Pi_{+-}^{(-)}(t)\big)
&=\cN\big(\Pi_{+-}^{(+)}(t)\big)
=\cR\big(\Pi_-^{(-)}(t)\big),
\\ \notag
\cR\big(\Pi_{-+}^{(+)}(t)\big)
&=\cN\big(\Pi_{-+}^{(-)}(t)\big)
=\cR\big(\Pi_-^{(+)}(t)\big),
\\ \notag
\cR\big(\Pi_{-+}^{(-)}(t)\big)
&=\cN\big(\Pi_{-+}^{(+)}(t)\big)
=\cR\big(\Pi_+^{(-)}(t)\big).
  \end{align}
    Note that
    \begin{align}
\Pi_{+-}^{(\pm)}+\Pi_{+-}^{(\mp)}&=\one.
      \end{align}
Now we can define the in-out Feynman and the out-in anti-Feynman Green functions in the
Cauchy data formalism:
\begin{align}\label{proper3}
  E_{+-}^\F(t,s)&=\theta(t-s) \Pi_{+-}^{(+)}(t)R(t,s)-\theta(s-t)
                  \Pi_{+-}^{(-)}(t)R(t,s),\\\label{proper4}
    E_{-+}^{ \overline\F}(t,s)&=\theta(t-s) \Pi_{-+}^{(-)}(t)R(t,s)-\theta(s-t)
                  \Pi_{-+}^{(+)}(t)R(t,s).
  \end{align}
Next we set
\begin{align}
  G_\pm^{(+)}(t,
  \mathbf{x};s,\mathbf{y})&=  \alpha(t,\mathbf{x})E_{\pm,12}^{(+)}(t,
                            \mathbf{x};s,\mathbf{y})  \alpha(s,\mathbf{y})
                            ,\\
    G_\pm^{(-)}(t,
  \mathbf{x};s,\mathbf{y})&=-  \alpha(t,\mathbf{x})E_{\pm,12}^{(-)}(t,
                            \mathbf{x};s,\mathbf{y})  \alpha(s,\mathbf{y})
                            ,\\
    G_{+-}^\F(t,
  \mathbf{x};s,\mathbf{y})&=\ii  \alpha(t,\mathbf{x}) E_{+-,12}^\F(t,
                            \mathbf{x};s,\mathbf{y})  \alpha(s,\mathbf{y})
                            ,\\
      G_{-+}^{ \overline\F}(t,
  \mathbf{x};s,\mathbf{y})&=\ii  \alpha(t,\mathbf{x})E_{-+,12}^{ \overline\F}(t,
                            \mathbf{x};s,\mathbf{y})  \alpha(s,\mathbf{y})
                            .
                            \end{align}

                            $G_-^{(\pm)}$ are  two-point functions of
                            the ``in-vacuum'' $\Omega_-$
                            and $G_+^{(\pm)}$ are  two-point functions
                            of the ``out-vacuum'' $\Omega_+$.
                            Both are Hadamard states \cite{GW17}.

The \emph{in-out Feynman propagator} $G_{+-}^\F(x,x')$ and the
\emph{out-in  anti-Feynman propagator} $G_{-+}^{\overline{\F}}(x,x')$
are the ``mixed Feynman propagators'' corresponding to those
states. In fact it is easy to see that if
$(\Omega_{+}|\Omega_{-})\neq0$ then
\begin{align}\label{they1}
 G_{+-}^\F(x,y)
 &= \ii \frac{\big(\Omega_{+}|\T\big(\hat\phi(x)\hat\phi(y)\big)\Omega_{-}\big)
 }{(\Omega_{+}|\Omega_{-})}, \\\label{they2}
  G_{-+}^{\overline{\F}}(x,y)
 &= -\ii \frac{\big(\Omega_{-}|\overline{\T}\big(\hat\phi(x)\hat\phi(y)\big)\Omega_{+}\big)
 }{(\Omega_{-}|\Omega_{+})}.
\end{align}

Assume in addition that $\alpha(x)$ and $\alpha^{-1}(x)$ are bounded on $M$.
One can then heuristically derive \cite{DS18,DS19}, and under some technical
assumptions rigorously prove \cite{V20,NT23}, that \eqref{they1} and
\eqref{they2}  coincide with the
operator-theoretic propagators:
\begin{align}
 G^\F_{\rm op} (x,y)= G_{+-}^\F(x,y),\\
 G^{\overline{\F}}_{\rm op}(x,y) =  G_{-+}^{\overline{\F}}(x,y).
\end{align}

\section{FLRW spacetimes}
\label{sec:FLRW}
\subsection{1+0-dimensional spacetimes}
\label{ssc:1dim}
$1+0$-dimensional spacetimes form an important class of
spacetimes for which we can understand various propagators
rather completely.

The  Klein-Gordon operator on $\RR^{1,0}$ can be
written as a {\em one-dimensional Schrödinger operator} (with the wrong sign
in front of the second derivative):
\begin{align}
 K:=-\Box+Y(t) =  \partial_t^2 + Y(t).
\end{align}
We will assume that
\begin{align}
  Y(t)=-V(t)+m^2,&\quad\lim_{t\to\pm\infty}V(t)=0,
\end{align}
so that we can write
\begin{align}
  H:=-\partial_t^2+V(t),\quad
  K= -H+m^2.\label{oned}
      \end{align}
Thus to discuss propagators on $1+0$-dimensional spacetimes one
needs to understand the theory of Green functions of
the one-dimensional Schrödinger operator $H$. A standard reference
for the subject is \cite{Y10}. In the following subsection, we present
this  well-known theory following  \cite{DeGeorgescu}
in a style adjusted to the QFT applications that we have in mind.

\subsection{Green functions of one-dimensional Schrödinger operators}
\label{sss:1dim_SchrOp}
Suppose that $k\in\mathbb{C}$ and  we are given two solutions $\psi_1,\psi_2$ of the
equation
\begin{align}
\label{eq:1d_eigeneq}
  (H+k^2)\psi(t)=0.
\end{align}
Their {\em Wronskian}
\begin{align}
 \sW(\psi_1,\psi_2) := \psi_1(t) \psi_2'(t)- \psi_1'(t) \psi_2(t)
\end{align}
does not depend on $t$. \eqref{eq:1d_eigeneq} possesses a dinstinguished
bisolution defined by
\begin{align}
 G^\leftrightarrow(-k^2;t,s) := \frac{\psi_1(t)\psi_2(s)-\psi_2(t)\psi_1(s)}{ \sW(\psi_1,\psi_2)}.
\end{align}
Note that
\begin{align}
 \label{eq:1d_in_cond}
 G^\leftrightarrow(-k^2;t,t)=0 \word{and}
 \partial_s G^\leftrightarrow(-k^2;t,s)\Big|_{s=t} = -\partial_t G^\leftrightarrow(-k^2;t,s)\Big|_{t=s}
=1.
\end{align}
It is easy to see that $G^\leftrightarrow(-k^2;t,s) $ is independent of the choice of
$\psi_1$ and $\psi_2$.
We call $G^\leftrightarrow(-k^2;t,s) $ the
\emph{canonical bisolution}. It is the analog of the Pauli-Jordan
propagator.

We then can define the \emph{forward and backward Green functions} via
\begin{align}
\label{eq:1d_forward}
 G^\rightarrow(-k^2;t,s) &:= \theta(t-s) G^\leftrightarrow(-k^2;t,s),
 \\ \label{eq:1d_backward}
 G^\leftarrow(-k^2;t,s) &:= -\theta(s-t) G^\leftrightarrow(-k^2;t,s).
\end{align}
Using \eqref{eq:1d_in_cond}, one readily verifies that
$G^\rightarrow$ and $G^\leftarrow$ are indeed Green functions. Needless
to say, they are the analogs of the retarded and advanced propagators.

Now let $\Re(k)>0$. The \emph{Jost solutions} $\psi_\pm(k,t)$
are the unique solutions of \eqref{eq:1d_eigeneq} with the
asymptotic behavior
\begin{align}
 \psi_\pm(k,t) \sim \ee^{\mp k t} \word{as} t\to\pm\infty.
\end{align}
The \emph{Jost function} is
\begin{align}
 \omega(k) := \sW\big( \psi_+(k,\cdot), \psi_-(k,\cdot)\big).
\end{align}
Then, it is well-known that the unique fundamental solution with appropriate decay behavior
as $|t|\to\infty$, that is, the integral kernel of the resolvent
$G(-k^2):=(H+k^2)^{-1}$, is
\begin{align}\label{resol}
 G(-k^2;t,s) := \frac{1}{\omega(k)}
 \Big( \theta(t-s) \psi_+(k,t) \psi_-(k,s)
 + \theta(s-t) \psi_-(k,t)\psi_+(k,s)\Big) .
\end{align}

Now let $m>0$. Setting $k=\pm\ii m$ in \eqref{resol} we see that the
distributional boundary values of the resolvent
on the spectrum are then given by
\begin{align}
 G(m^2\mp\ii0;t,s) = \frac{\theta(t-s) \psi_+(\pm\ii m,t) \psi_-(\pm\ii m,s)
 + \theta(s-t) \psi_-(\pm\ii m,t)\psi_+(\pm\ii m,s)}{\omega(\pm\ii m)}.
\end{align}
Thus we computed all four basic Green functions of the Klein-Gordon
equation given by \eqref{oned}:
\begin{align}
  \text{retarded propagator:}\qquad&G^\rightarrow(m^2;t,s),\\
  \text{advanced propagator:}\qquad&G^\leftarrow(m^2;t,s),\\
  \text{Feynman propagator:}\qquad& G(m^2-\ii0;t,s),\\
  \text{anti-Feynman propagator:}\qquad& G(m^2+\ii0;t,s).
                                        \end{align}
The Klein-Gordon scalar product essentially coincides with the
  Wronskian:
  \begin{align}
    (\psi_1|\psi_2)_\KG=\sW( \overline\psi_1,\psi_2).\end{align}

One can now ask when  the Klein-Gordon equation given
by the operator \eqref{oned} on a $1+0$-dimensional spacetime is special, i.e., when the following identity holds:
\begin{align}
\label{eq:1d_special}
 G(m^2-\ii0) + G(m^2+\ii0)
 =G^\rightarrow(m^2) + G^\leftarrow(m^2)?
\end{align}
To answer this question, it is useful to introduce the concept of
\emph{reflectionlessness}.
\begin{defn}
 Let $A(\pm\ii m)$ and $B(\pm\ii m)$ denote the coefficients of the
 scattering matrix, i.e.,
 \begin{align}
  \psi_+(\pm\ii m, t ) = A(\pm\ii m) \psi_-(\mp\ii m,t )
  + B(\pm\ii m) \psi_+(\mp\ii m,t ).
 \end{align}
The potential $Y(t)$ is called \emph{reflectionless at energy $m^2$} if
$B(\pm \ii m)=0$.
\end{defn}
We have the following theorem.
\begin{thm}
 The potential $Y(t)$ is reflectionless if and only if the
 Klein-Gordon equation given by \eqref{oned}
 is special, i.e., if and only if \eqref{eq:1d_special}
 is true.
\end{thm}
\begin{proof}
 We have
 \begin{align}
 \notag
 &\quad G(m^2-\ii0)+G(m^2+\ii0)
  \\ \label{eq:1d_sum_forward}
  &= \theta(t-s) \Bigg(
  \frac{\psi_+(\ii m,t) \psi_-(\ii m,s) }{\omega(\ii m)}
  +\frac{\psi_+(-\ii m,t) \psi_-(-\ii m,s) }{\omega(-\ii m)}\Bigg)
  \\ \label{eq:1d_sum_backward}
  &\quad
  + \theta(s-t) \Bigg(
  \frac{\psi_-(\ii m,t) \psi_+(\ii m,s) }{\omega(\ii m)}
  +\frac{\psi_-(-\ii m,t) \psi_+(-\ii m,s) }{\omega(-\ii m)}\Bigg).
 \end{align}
Moreover,
\begin{align}
 \omega(\pm\ii m)
 &= \pm A(\pm \ii m) \sW( \psi_-(-\ii m), \psi_-( \ii m))
    + B(\pm\ii m) \sW( \psi_+(\mp\ii m), \psi_-( \pm\ii m)) .
\end{align}
Then the part \eqref{eq:1d_sum_forward} becomes
\begin{align}
 \theta(t-s)  &\Bigg(
 \frac{A(\ii m) \psi_-(-\ii m,t) \psi_-(\ii m,s)
 + B(\ii m) \psi_+(-\ii m,t)\psi_-(\ii m,s)}{
 A( \ii m) \sW( \psi_-(-\ii m), \psi_-( \ii m))
    + B(\ii m) \sW( \psi_+(-\ii m), \psi_-( \ii m)) }
    \\ \notag
&\quad - \frac{A(-\ii m) \psi_-(\ii m,t) \psi_-(-\ii m,s)
 + B(-\ii m) \psi_+(\ii m,t)\psi_-(-\ii m,s)}{
  A(- \ii m) \sW( \psi_-(-\ii m), \psi_-( \ii m))
    - B(-\ii m) \sW( \psi_+(\ii m), \psi_-( -\ii m)) }
 \Bigg)
\end{align}
Since $A(\pm \ii m)\neq0$, this is $G^\rightarrow$ if
and only if $B(\pm \ii m)=0$.  Similar for \eqref{eq:1d_sum_backward}.
\end{proof}

\subsection{Mode decomposition of FLRW spacetimes}
Consider a {\em Friedmann-Lemaître-Robertson-Walker (FLRW) spacetime},
that is, a spacetime $M= \bR\times\Sigma$ with the line element
\begin{align}
 \dd s^2 = - \dd t^2 + a(t)^2 \dd \Sigma^2,
\end{align}
 where $\dd\Sigma^2$ is the line element of a fixed $d-1$-dimensional
complete Riemannian manifold $\Sigma$. The Klein-Gordon operator is
\begin{align}\label{FLRW}
 -\square_g +m^2
 = \partial_t^2 + (d-1) \frac{\dot{a}(t)}{a(t)} \partial_t
 -  \frac{\Delta_\Sigma}{a(t)^2} + m^2,
\end{align}
where the dot indicates a derivative with respect to $t$.
Then
\begin{align}
\label{eq:KG_FLWR_gauged}
 a^{\tfrac{d-1}{2}} \big(-\square_g +m^2\big) a^{-\tfrac{d-1}{2}}
 = \partial_t^2 -\frac{d-1}{2} \Big(
 \frac{\ddot{a}}{a} + \frac{d-3}{2} \Big(\frac{\dot{a}}{a}\Big)^2
 \Big)  -  \frac{\Delta_\Sigma}{a(t)^2} + m^2.
\end{align}

It is well-known that $-\Delta_\Sigma$ is self-adjoint, and
  by the spectral theorem we can diagonalize $-\Delta_\Sigma$, and
then to
restrict \eqref{eq:KG_FLWR_gauged} to a (generalized) eigenfunction
(a ``mode'') of
$-\Delta_\Sigma$ with eigenvalue $\lambda$.
Thus, for each such mode, \eqref{eq:KG_FLWR_gauged} becomes
$-H_\lambda+m^2$, where
\beq H_\lambda:=-\partial_t^2+V_\lambda(t)\label{refle}
\eeq
is the one-dimensional
Schrödinger operator with potential
\begin{align}
 V_\lambda(t) = \frac{d-1}{2} \Big(
 \frac{\ddot{a}}{a} + \frac{d-3}{2} \Big(\frac{\dot{a}}{a}\Big)^2
 \Big)  - \frac{\lambda}{a(t)^2}.
\end{align}
Using Subsection \ref{sss:1dim_SchrOp}, we can then write all propagators as the integral over all modes.

As a consequence,  the Klein-Gordon equation given by
  \eqref{FLRW} is special if and only if
\eqref{refle}
is reflectionless
at energy $m^2$ for all $\lambda$ in the spectrum of $-\Delta_\Sigma$.



\section{De Sitter space}
\label{ssc:dS}

Our next example is the $d$-dimensional {\em de Sitter space} $\dS_d$.
De Sitter space is  an important example of a non-stationary spacetime
and one of the simplest examples to model a universe with an
accelerated expansion. It exhibits a particularly
rich structure and, being a symmetric space, all its
invariant propagators can be given explicitly in
terms of special functions.

We will describe four different approaches to investigate propagators on
$\dS_d$. The first  is based on Wick rotation (analytic
continuation) from the sphere $\mathbb{S}^d$. One obtains the
so-called Euclidean state, considered to be the most physical
invariant state on $\dS_d$. The second approach
is the off-shell approach based on the resolvent of the d'Alembertian
on $L^2(\dS_d)$. Somewhat surprisingly, it
leads to non-physical two-point functions. The third approach is
 the on-shell approach based on $\cW_\KG$. It leads to the well-known family of de Sitter invariant
two-point functions corresponding to the so-called $\alpha$-\emph{vacua}.
One can then compute invariant correlation functions
between \emph{two different} $\alpha$-vacua. Finally, we may
interpret $\dS_d$ as a special case of a FLRW spacetime and apply
the methods
of Section \ref{sec:FLRW}.

 Note that the first three approaches directly lead
to simple expressions for  invariant propagators. The last approach breaks manifest de Sitter
invariance, and to obtain  invariant expressions, one needs to sum
over all modes using rather complicated addition formulas for special
functions.

There is a very large literature about propagators on de Sitter space.
Particularly useful for our considerations were
\cite{ABDMPS19,Allen85,BD69,BMS02,BEM98,BGM94,BM96,CT68,CritchleyPhD,FH14,FSS13,HMM11,Hollands12,Hollands13,Lokas,Moschella06,SS76,SSV01,RumpfdS,AvisIS78}. In these references, one finds different
approaches to investigate propagators on de Sitter space.

Many of them use mode sums to construct propagators --
sometimes explicitly like in \cite{FSS13,ABDMPS19,M85,BMS02}, sometimes abstractly like in \cite{Allen85}. The papers
\cite{BEM98,BGM94,BM96}
have an axiomatic approach much in the spirit of G\r{a}rding
and Wightman. Only the reference \cite{RumpfdS} uses the
operator-theoretic approach to define the Feynman propagator
in  $d=4$ dimensions.

\subsection{Geometry of de Sitter space}
\label{sss:dS_geom}
The $d$-dimensional de Sitter space $\dS_d$ is defined by an embedding into
$d+1$-dimensional Minkowski space $\bR^{1,d}$. Let $[\cdot|\cdot]$
denote the pseudo-scalar product on $\bR^{1,d}$ defined by
\begin{align}\label{minko}
 [x|x'] = -x^0 {x'}^0 + \sum_{i=1}^d x^i {x'}^i.
\end{align}
 Then the $d$-dimensional
de Sitter space is the one-sheeted hyperboloid
\begin{align}
 \dS_d := \{ x\in\bR^{1,d}\;|\; [x|x]=1\}.
\end{align}

Let us introduce some notation that will frequently appear
throughout this section. For $x,x' \in\dS_d \hookrightarrow \bR^{1,d}$,
we define
\begin{align}
 &\textup{the invariant quantity}  &&Z \equiv Z(x,x') := [x|x'],
 \\ \notag
 &\textup{the antipodal point to } x{\rm :} && x^A:=-x,
 \\ \notag
 &\textup{the time variable} && t  \equiv t(x,x') := x^0-{x'}^0,
 \\ \notag
 &\textup{the ``antipodal time'' variable}  && t^A := t(x^A,x') := -(x^0+{x'}^0).
\end{align}
While $t$ and $t^A$ are two independent variables, we have $Z(x^A,x') = -Z(x,x')=-Z$.

De Sitter space has various regions:
\begin{align}
 &Z>1: &&\quad x \text{ and } x' \text{ are timelike separated,}
 \\ \notag
&Z=1: &&\quad x \text{ and } x'
    \text{ are separated by a null-geodesic,}
 \\ \notag
 &Z<1: &&\quad x \text{ and } x' \text{ are not
 connected by a causal curve.}
\end{align}
The last region includes the subregions
\begin{align}
 &Z=-1: &&\quad x^A \text{ and } x'
 \text{ are separated by a null-geodesic},
 \\ \notag
 &Z<-1: &&\quad x^A \text{ and } x' \text{ are timelike separated}.
\end{align}

One may further divide the regions $Z>1$ and $Z<-1$ into future and past
dependent on whether $t$, resp. $t^A$ are positive or negative.
Thus, if we fix a point $x'\in\dS_d$, then we can partition $\dS_d$
into 5 regions:
\begin{align}
 \dS_d=V^+\cup V^-\cup A^+\cup A^-\cup S
\end{align}
as depicted in Figure \ref{fig:dS}.

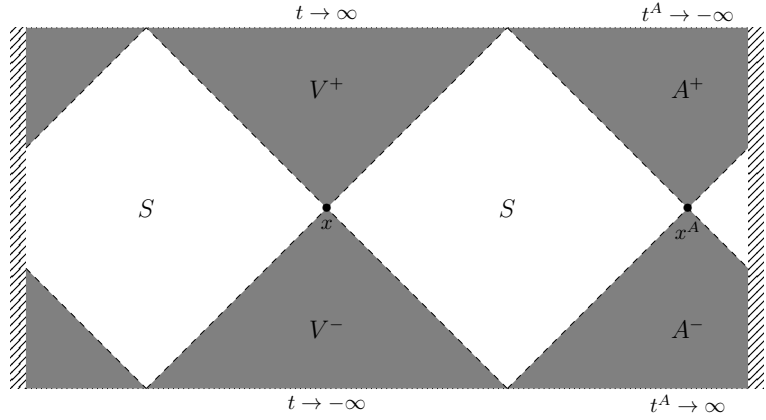
\begin{figure}[ht]
\caption{Conformal diagram of de Sitter space with the reference
point $x$ and the regions
$V^{\pm} := \{ Z(x,x')>1 \;|\; t(x,x')\gtrless0 \}$,
$A^{\pm} := \{ Z(x,x')<-1 \;|\; t(x^A,x')\lessgtr0 \}$
and $S :=\{|Z(x,x')|<1\}$. The left and right side of the diagram
are glued together and each point represents a $d-2$-sphere.}\label{fig:dS}
{\scalebox{0.8}{\begin{tikzpicture}[scale=0.5, domain=-3:5]

\begin{scope}
\path[fill=gray, opacity=0.2] (0,0) -- (6,-6) -- (-6,-6) -- cycle;
\end{scope}
\begin{scope}
\path[fill=gray, opacity=0.2] (0,0) -- (6,6) -- (-6,6) -- cycle;
\end{scope}
\draw[color=black, dashed] (-6,-6) -- (0,0) -- (-6,6);
\draw[color=black, dashed] (6,-6) -- (0,0) -- (6,6);

\begin{scope}
\path (8,-4) -- (8,4) -- (0,0) -- cycle;
\end{scope}
\begin{scope}
\path[fill=gray, opacity=0.1] (6,6) -- (12,0) -- (14,2) -- (14,6) -- cycle;
\end{scope}
\begin{scope}
\path[fill=gray, opacity=0.1] (6,-6) -- (12,0) -- (14,-2) -- (14,-6) -- cycle;
\end{scope}
\begin{scope}
\path[fill=gray, opacity=0.1] (-10,2) -- (-6,6) -- (-10,6)  -- cycle;
\end{scope}
\begin{scope}
\path[fill=gray, opacity=0.1] (-10,-2) -- (-6,-6) -- (-10,-6)  -- cycle;
\end{scope}

\begin{scope}
\path[fill=gray, opacity=0.2, pattern= north east lines] (-10,6) -- (-10,-6) -- (-10.5,-6) -- (-10.5,6) -- cycle;
\end{scope}
\begin{scope}
\path[fill=gray, opacity=0.2, pattern= north east lines] (14,6) -- (14,-6) -- (14.5,-6) -- (14.5,6) -- cycle;
\end{scope}

\draw[color=black, dashed] (6,6) -- (14,-2);
\draw[color=black, dashed] (6,-6) -- (14,2);
\draw[color=black, dashed] (-10,-2) -- (-6,-6);
\draw[color=black, dashed] (-10,2) -- (-6,6);

\draw[dotted,color=black] (-10,-6) -- (14,-6);
\draw[dotted,color=black] (-10,6) -- (14,6);

\draw[color=black] (0,0) node {\footnotesize $\bullet$};
\draw[color=black] (0,-0.1) node[below] {\footnotesize $x$};
\draw[color=black] (12,0) node {\footnotesize $\bullet$};
\draw[color=black] (12,-0.15) node[below] {\footnotesize $x^A$};

\node at (0,4) {$V^{+}$};
\node at (0,-4) {$V^{-}$};
\node at (12,-4) {$A^{-}$};
\node at (12,4) {$A^+$};
\node at (6,0) {$S$};
\node at (-6,0) {$S$};
\node at (0,6.5) {\footnotesize $t\to\infty$};
\node at (0,-6.5) {\footnotesize $t\to-\infty$};
\node at (12,-6.5) {\footnotesize $t^A\to\infty$};
\node at (12,6.5) {\footnotesize $t^A\to-\infty$};
\end{tikzpicture}}}
\end{figure}

The de Sitter space possesses a global system of coordinates
\begin{align}
\label{eq:dS_global_coords}
 x^0 =  \sinh\tau, \quad x^i=\cosh\tau\; \Omega^i,\;i=1,\dots,d,
 \word{where}  \tau\in\bR,\;\Omega\in\bS^{d-1}\hookrightarrow\bR^{d}.
\end{align}
In these coordinates we have
$\dd s^2=-\dd \tau^2+\cosh^2(\tau)\dd\Omega^2$ and
\begin{align}
\label{eq:Z_dS_global_coords}
 Z = -\sinh\tau\sinh\tau' +\cosh\tau\cosh\tau'\cos\theta,
\end{align}
where $\theta$ is the angle between $\Omega$ and $\Omega'$.
 If $x=(0,1,0\dots)$, then $Z=\cosh\tau'\cos\theta$.

  Both the {\em (full) de Sitter group} $O(1,d)$ and
  the  {\em restricted de Sitter group} $SO_0(1,d)$, that is,
  the connected component of the identity in $O(1,d)$, act on $\dS_d$.
The Klein-Gordon equation restricted to invariant solutions
and written in terms of $Z$ reduces to the Gegenbauer equation,
a form of the  hypergeometric equation  \cite{BirrellDavies,Allen85,SS76,GS68,M85,CT68} whose properties we
discuss in Appendix \ref{app:gegenbauer}.

In the literature one
often restricts analysis to subsets of $\dS_d$, such as the {\em
  Poincaré patch} or the
{\em static patch}, which allow for coordinate systems with special
properties.
In our paper we consider only the ``{\em global patch}'', that is the
full de Sitter space. Otherwise, we would have to consider
boundary conditions for the d'Alembertian at the boundary of our patch
(which would break the de Sitter invariance and presumably
be non-physical).

For more information about de Sitter space, consult the 
overviews \cite{Moschella06,SSV01} .

\subsection{The sphere}

The de Sitter space can be viewed as a Wick-rotated sphere. Therefore,
in this subsection we recall some facts about the sphere and
 the Green function of the spherical Laplacian.

Consider the $d+1$ dimensional Euclidean space equipped with the
scalar product
\beq(x|x')=\sum_{i=1}^{d+1}x^ix^{\prime i}.\eeq
The $d$-dimensional (unit) sphere is defined as
\begin{align}
 \mathbb{S}^d := \{ x\in\bR^{d+1}\;|\; (x|x)=1\}.
\end{align}
For $\Re(\nu)>0$ or $\nu\in\ii\bR_{\geq0}\setminus \ii\big(
\tfrac{d-1}{2}+\bN_0\big)$, let us consider the resolvent of the
spherical Laplacian
$G^\s(-\nu^2):=(-\Delta^\s+(\frac{d-1}{2})^2+\nu^2)^{-1}$.
Its integral kernel $G^\s(-\nu^2;x,x')$ can be expressed in
terms of the invariant
quantity $(x|x')$ (see e.g. \cite{DGR23a,DGR23b}, and
  \cite{CDT18,Szmytkowski07}, where Legendre functions are used) as:
 \begin{align}\label{eq:G_sphere}
  G^\s(-\nu^2;x,x')
   =C_{d,\nu}
     {\bf S}_{\frac{d-2}{2},\ii\nu }\big(-(x|x')\big),
 \end{align}
where
\begin{align}\label{cdnu}
C_{d,\nu}:= \frac{\Gamma\big(\tfrac{d-1}{2}+\ii\nu \big)
\Gamma\big(\tfrac{d-1}{2}-\ii\nu \big)
     }{(4\pi)^{\frac{d}{2}}},
\end{align}
and  ${\bf S}_{\alpha,\lambda }(z)$ is the Gegenbauer function
described in Appendix \ref{app:gegenbauer}.

\subsection{Propagators related to  the Euclidean state}

We now turn to
the $d$-dimensional de Sitter space for $d\geq2$.
We will analyze bi-
and fundamental solutions of the Klein-Gordon equation
\beq(-\Box+m^2)\phi(x)=0\label{kgo1}\eeq
in de Sitter space,
which are  invariant under the full or restricted de Sitter group.
Note that $m$ might contain a coupling to the
scalar curvature. Hence it is sometimes called {\em effective
  mass}. Anyway, we prefer to use the parameter $\nu$ defined by
\begin{align}
 \nu := \sqrt{m^2-\big(\tfrac{d-1}{2}\big)^2}\in\bC .
\end{align}
Thus \eqref{kgo1} is replaced with
\begin{align}
\Big(-\Box+\big(\tfrac{d-1}{2}\big)^2+\nu^2\Big)\phi(x)=0.
\label{kgo}
\end{align}
We will  allow for complex $\nu^2$,  choosing the principal sheet
of the square root, that is $\nu\in\{\Re(\nu)>0\}$. The case
of positive $\nu^2$ has analogous properties to that  of positive
$m^2$ in Minkowski space. In the case $\nu^2<0$ we assume that
$\nu\in\ii\bR_{\geq0}$. It is more intricate than the case $\Re(\nu)>0$
and contains various subcases with different exotic
properties. It is somewhat
analogous  to the tachyonic case in Minkowski space.

On a generic spacetime the concept of the Wick rotation is not
uniquely defined.
However, on the de Sitter space embedded in $\mathbb{R}^{1,d}$
there is a natural kind of a Wick rotation, which we will use: the
replacement of $x^{d+1}$ with $\pm\ii x^0$. We note first that
 \begin{align}
  (x|x') =  1-\frac{(x-x'|x-x')}{2} \word{for} x,x' \in \bS^d.
 \end{align}
The replacement of $x^{d+1}-{x'}^{d+1}$ with $ (x^0-{x'}^0)
\ee^{\pm\ii\phi}$, $\phi\in[0,\tfrac{\pi}{2}]$, yields
 \begin{align}
 &(x^{d+1}-x'^{d+1})^2 \to (x^0-x'^0 )^2 \ee^{\pm2\ii\phi}
 \stackrel{\phi\to\tfrac{\pi}{2}}{\longrightarrow}
 -(x^0-x'^0 )^2\pm\ii0
\\ \notag
  \Rightarrow\;&
  (x|x') \to [x|x']\mp\ii0.
 \end{align}
 Moreover, we need to insert a prefactor $\pm \ii$ coming from
 the change of the integral measure.

Let $\Re(\nu)>0$ or $\nu\in\ii\bR_{\geq0}\setminus \ii\big(
\tfrac{d-1}{2}+\bN_0\big)$.
The Feynman and anti-Feynman propagators in the
$d$-dimensional de Sitter space obtained by Wick rotation of the
Green function \eqref{eq:G_sphere} on the sphere are given by
 \begin{align}
 \label{eq:GF_GFbar_global_dS}
 G^{\F/\overline{\F}}_{0}(x,x')
&= \pm \ii C_{d,\nu}\;
     {\bf S}_{\frac{d}{2}-1,\ii\nu }\big(-Z\pm\ii0\big),
\end{align}
where $C_{d,\nu}$ is given by \eqref{cdnu} and $Z:= [x|x'].$
 We easily check that
\eqref{eq:GF_GFbar_global_dS} are Green functions of the Klein-Gordon
equation on $\dS_d$.

 The sum of the Euclidean Feynman and anti-Feynman propagator
 has a causal  support, for
  ${\bf S}_{\alpha,\lambda }(z)$ is
 holomorphic on $\bC \setminus ]-\infty,-1]$,  and therefore
 \begin{align}\label{vani}
   G^{\F}_{0} + G^{\overline{\F}}_{0}
   = \ii C_{d,\nu}\;\Big(
     {\bf S}_{\frac{d}{2}-1,\ii\nu }\big(-Z+\ii0\big)
     -{\bf S}_{\frac{d}{2}-1,\ii\nu }\big(-Z-\ii0\big)\Big)
 \end{align}
 vanishes for $Z<1$.

 As we will see later, $G^{\F}_{0}$ and $G^{\overline{\F}}_{0}$
 are not the operator-theoretic Feynman and anti-Feynman propagators.
 However, we can still apply to them the procedure described in
 Subsection \ref{Special spacetimes}. This leads to  the classical
 propagators
\begin{align}
   G^{\lor/\land} (x,x')
  &= \ii \theta\big(\pm(x^0-{x'}^0)\big)
 C_{d,\nu}\;
     \Big(
     {\bf S}_{\frac{d}{2}-1,\ii\nu }\big(-Z+\ii0 \big)
   -{\bf S}_{\frac{d}{2}-1,\ii\nu }\big(-Z-\ii0  \big)
   \Big)
     ,
     \label{eq:ret_adv_props_dS1}\\
   G^{\PJ} (x,x')
  \label{eq:PJ_prop_dS1}
  &=\ii \sgn\big(x^0-{x'}^0\big)
 C_{d,\nu}
     \Big(
     {\bf S}_{\frac{d}{2}-1,\ii\nu }\big(-Z+\ii0 \big)
   -{\bf S}_{\frac{d}{2}-1,\ii\nu }\big(-Z-\ii0  \big)
    \Big),\end{align}
as well as to the positive/negative frequency solutions
\begin{align}
   G^{(\pm)}_{0} (x,x')
   &= C_{d,\nu}\;
       {\bf S}_{\frac{d}{2}-1,\ii\nu }\Big(
     -Z\pm\ii0 \sgn\big(x^0-{x'}^0\big)     \Big).
  \label{eq:pos_neg_props_dS}
 \end{align}
$ G^{(\pm)}_{0} $ have the Hadamard property and are two-point
functions of a state  called the \emph{Euclidean state} $\Omega_0$
(sometimes also called the \emph{Bunch-Davies state})
\cite{Allen85,BunchDavies78,CT68,GS68,M85,SS76}.

Note that the propagators associated to the Euclidean vacuum satisfy
all relations \eqref{eq:relations_props} with $\alpha=\beta=0$.
  The classical propagators \eqref{eq:ret_adv_props_dS1} and
\eqref{eq:PJ_prop_dS1}
are universal: they do not depend on the Euclidean vacuum, therefore
we do not decorate them with the subscript $0$.

\subsection{Bisolutions and Green functions}

The family of invariant propagators on the de Sitter space is quite
rich and is not limited to those related to the Euclidean state,
discussed in the previous subsection. In order to prepare for their
analysis, in this subsection we will descibe invariant solutions of
the Klein-Gordon equation on de Sitter space.

From the analysis of previous subsection we easily see that the
following functions are bisolutions invariant with respect to the
full de Sitter group:
\begin{align}\label{symm}
  G^\sym_{0}(x,x')
  :=&\;  G^{(+)}_{0}(x,x') + G^{(-)}_{0}(x,x') \\
  \notag
  =&\; C_{d,\nu} \Big(
     {\bf S}_{\frac{d}{2}-1,\ii\nu }\big(-Z+\ii0 \big)
   +{\bf S}_{\frac{d}{2}-1,\ii\nu }\big(-Z-\ii0  \big)
     \Big),\\\label{symm-A}
  G^{\sym,A}_{0}(x,x') := & G^\sym_{0}(x^A,x')=G^\sym_{0}(x,x^{\prime A})\\
    =&\; C_{d,\nu} \Big(
     {\bf S}_{\frac{d}{2}-1,\ii\nu }\big(Z+\ii0 \big)
   +{\bf S}_{\frac{d}{2}-1,\ii\nu }\big(Z-\ii0  \big)
       \Big),\notag
\end{align}
The following functions are bisolutions invariant with respect to the
 restricted de Sitter group:
\begin{align}
     G^{\PJ} (x,x'):=&\ii\big(  G^{(+)}_{0}(x,x') - G^{(-)}_{0}(x,x')  \big)
  \label{eq:PJ_prop_dS1-}\\
  =&\ii \sgn\big(t\big)
 C_{d,\nu}
     \Big(
     {\bf S}_{\frac{d}{2}-1,\ii\nu }\big(-Z+\ii0 \big)
   -{\bf S}_{\frac{d}{2}-1,\ii\nu }\big(-Z-\ii0  \big)
      \Big),\notag
  \\
         G^{\PJ,A} (x,x') := &      G^{\PJ} (x^A,x') =-  G^{\PJ}
                               (x,x^{\prime A})
  \label{eq:PJ_prop_dS1-A}\\\notag
  =&\ii \sgn\big(t^A\big)
 C_{d,\nu}
     \Big(
     {\bf S}_{\frac{d}{2}-1,\ii\nu }\big(Z+\ii0 \big)
   -{\bf S}_{\frac{d}{2}-1,\ii\nu }\big(Z-\ii0  \big)
    \Big).
\end{align}
 Indeed, we already know that $G_0^{(\pm)}$ are bisolutions, hence so
are \eqref{symm} and
\eqref{eq:PJ_prop_dS1-}. It is also clear that replacing $x$ with
$x^A$, used in \eqref{symm-A} and \eqref{eq:PJ_prop_dS1-A} leads to
invariant bisolutions. We expect that the following is true:
\begin{con}\label{thm:dS_bisol}
 For any $\nu\in\bC$ such that
$\tfrac{d-1}{2}\pm\ii\nu \notin\{0,-1,-2,\dots\}$,
$\{G_0^{\sym},G_0^{\sym,A}\}$ is a basis of
the space of fully de Sitter invariant bisolutions, and
  $\{G_0^{\sym},G_0^{\sym,A},G^\PJ,G^{\PJ,A}\}$ is a basis of the
  space of bisolutions invariant under the restricted de Sitter group.
\end{con}
Note that the Gegenbauer function ${\bf S}_{\frac{d}{2}-1,\ii\nu}(w)$
is an entire function of $\nu$. If we were only interested in bisolutions,
we could drop the restriction $\tfrac{d-1}{2}\pm\ii\nu\notin\{0,-1,-2,\dots\}$
in Thm. \ref{thm:dS_bisol}, which is only necessary due to the poles of the
prefactor $C_{d,\nu}$ at such $\nu$. However, we eventually want to relate
bisolutions to Green functions by time-ordering, and therefore we normalize
them properly.

Functions invariant with respect to the full
de Sitter group can always be written in terms of the invariant quantity
$Z$ alone.
The Klein-Gordon equation restricted to invariant
solutions and written in terms of $Z$ reduces to the Gegenbauer equation
(cf. e.g. \cite{GS68,Allen85,BMS02})
\begin{align}
\Big((1-Z^2)\del_Z^2 - d Z \del_Z - \nu^2 - \big(\tfrac{d-1}{2}\big)^2\Big) f(Z)=0.
\end{align}
Therefore, all bisolutions and Green functions invariant wrt the full
de Sitter group can
be expressed in terms of Gegenbauer functions.

If we only demand invariance under the restricted de Sitter
group, the regions $V^+$ and $V^-$ as well as $A^+$ and $A^-$ need to
be treated as independent. Hence for $|Z|>1$, propagators  invariant under
the restricted de Sitter group may depend on $\sgn(t)$ resp. $\sgn(t^A)$.

Assuming the validity of Conjecture \ref{thm:dS_bisol}, the
general bisolution is
\begin{align}
\label{eq:gen_bisol}
  G^{{\rm bisol}}_{\underline{a}} =
 :=&\, \ii a_1 G_0^\sym + a_2 G^\PJ
 + \ii a_3 G_0^{\sym,A} + a_4 G^{\PJ A}\\\notag
 =&\, \ii C_{d,\nu}\Big(
 \big(a_1 + a_2\sgn(t)\big) {\bf S}_{\tfrac{d-2}{2},\ii\nu}(-Z+\ii0)
  \\ \notag
  &\qquad +\big(a_1 - a_2\sgn(t)\big) {\bf S}_{\tfrac{d-2}{2},\ii\nu}(-Z-\ii0)
  \\ \notag &\qquad
+ \big(a_3 - a_4\sgn(t^A) \big) {\bf S}_{\tfrac{d-2}{2},\ii\nu}(Z+\ii0)
\\ \notag
  &\qquad
  + \big(a_3 + a_4\sgn(t^A) \big) {\bf S}_{\tfrac{d-2}{2},\ii\nu}(Z-\ii0)
 \Big)
\end{align}
and the general fundamental solution is
\begin{align}
\label{eq:gen_fund_sol}
  G_{\underline{a}}
 :=&\,  G^{\F}_{0} + G^{{\rm bisol}}_{\underline{a}}
 = \ii C_{d,\nu} {\bf S}_{\tfrac{d-2}{2},\ii\nu}(-Z+\ii0)
 + G^{{\rm bisol}}_{\underline{a}}  .
\end{align}

\subsection{Resolvent of the d'Alembertian and
operator-theoretic propagators}
\label{sss:dS_op_inv}

The d'Alembertian $-\Box$ is essentially self-adjoint  on
$C_\mathrm{c}^\infty(\dS_d)$ in the sense of $L^2(\dS_d)$.
This follows from a general theory of invariant differential operators on symmetric spaces  \cite{Rossmann78,vdBan84} and the fact that de Sitter
space can be seen as the quotient of Lie groups $O(1,d)/O(1,d-1)$. In this subsection we will compute its resolvent and operator-theoretic
Feynman and anti-Feynman propagators.
In the four-dimensional case, this has been studied \cite{RumpfdS}.

Outside of the spectrum of $-\Box+\big(\tfrac{d-1}{2}\big)^2$, its
resolvent (Green operator) will be denoted by
\begin{align}
G(-\nu^2):=\Big(-\Box+\big(\tfrac{d-1}{2}\big)^2+\nu^2\Big)^{-1}.
\end{align}
We will write   $G(-\nu^2;x,x')$ for its integral kernel.

In the following statement we will compute $G(-\nu^2;x,x')$. This
computation, short and, we believe, quite elegant, is based on Conjecture
\ref{thm:dS_bisol}, which does not have a complete proof in our paper.
 Therefore, strictly speaking, all statements in this subsection are
 not fully proven in our paper, even if we call them ``theorems''.

One can justify Thm. \ref{resofeyn} independently, following the (rather
complicated) arguments of \cite{FSS13} involving global coordinates and
summation formulas for Gegenbauer functions.
We will not discuss these arguments in this paper.

\begin{thm} Let $\Re\nu>0$.\\
{\bf Odd $d$.}
    The resolvent is given by
    \begin{align}
& \quad  G(-\nu^2;x,x')\label{eq:dS_op_inv_odd1} \\ \notag
 &=  \frac{\pm\Gamma\big(\tfrac{d-1}{2}\pm\ii\nu \big)}{2^{2\pm\ii\nu}
 (2\pi)^{\frac{d-1}{2}}\sinh\pi\nu}
 \Big({\bf Z}_{\tfrac{d-2}{2},\pm\ii\nu}(-Z-\ii 0)
-{\bf Z}_{\tfrac{d-2}{2},\pm\ii\nu}(-Z+\ii 0) \Big), \quad\Im\nu\lessgtr0.
\end{align}
Therefore, for  $\nu>0$, the Feynman and anti-Feynman propagators are
\begin{align}
\label{eq:dS_op_inv_odd1.}
 G^{\F/\overline{\F}}_{{\rm op}}(x,x')
 &=  \frac{\pm\Gamma\big(\tfrac{d-1}{2}\pm\ii\nu \big)}{2^{2\pm\ii\nu}
 (2\pi)^{\frac{d-1}{2}}\sinh\pi\nu}
 \Big({\bf Z}_{\tfrac{d-2}{2},\pm\ii\nu}(-Z-\ii 0)
-{\bf Z}_{\tfrac{d-2}{2},\pm\ii\nu}(-Z+\ii 0) \Big).
\end{align}

{\bf Even $d$.} The resolvent is given by
\begin{align} & \quad G (-\nu^2;x,x') \label{eq:dS_op_inv_even1}\\ \notag
 &=  -\frac{\Gamma\big( \tfrac{d-1}{2}\pm\ii\nu \big)
 }{2^{2\pm\ii\nu}(2\pi)^{\tfrac{d-1}{2}} \cosh\pi\nu}
\Big({\bf Z}_{\tfrac{d-2}{2},\pm\ii\nu}(-Z+\ii 0)
+{\bf Z}_{\tfrac{d-2}{2},\pm\ii\nu}(-Z-\ii 0) \Big), \quad\Im\nu\lessgtr0.
\end{align}
Therefore, for  $\nu>0$, the operator-theoretic Feynman and anti-Feynman propagators are
\begin{align}
\label{eq:dS_op_inv_even1.}
 &\quad  G^{\F/\overline{\F}}_{{\rm op}} (x,x') \\ \notag
 &=   -\frac{\Gamma\big( \tfrac{d-1}{2}\pm\ii\nu \big)
 }{2^{2\pm\ii\nu}(2\pi)^{\tfrac{d-1}{2}} \cosh\pi\nu}
\Big({\bf Z}_{\tfrac{d-2}{2},\pm\ii\nu}(-Z+\ii 0)
+{\bf Z}_{\tfrac{d-2}{2},\pm\ii\nu}(-Z-\ii 0) \Big).
\end{align}
\label{resofeyn}
\end{thm}
\begin{proof}[Proof (assuming the validity of Conj. \ref{thm:dS_bisol}).]
Let us first compute the Green operator $G(-\nu^2)$ for
$\nu^2\in\mathbb{C}\setminus\mathbb{R}$.
Clearly, its integral kernel is a Green function invariant
under the \emph{full} de Sitter group.
Its integral kernel (as the integral
kernel of a bounded operator) must not grow too fast as $Z\to \pm
\infty$. By Conjecture  \ref{thm:dS_bisol}, the formula
\eqref{eq:gen_fund_sol} describes the family of all fully de Sitter
invariant Green functions.

To start, we thus use the connection formula
\eqref{formu3} to write the
general fundamental solution \eqref{eq:gen_fund_sol} in terms of
the Gegenbauer functions ${\bf Z}_{\alpha,\pm\lambda}(-Z\pm\ii0)$,
which have a determined behavior as $|Z|\to\infty$. Since we
require invariance under the full de Sitter group, we must have
$a_2=a_4=0$. This yields
\begin{align}
\label{eq:dS_inv_gen_fundsol}
  \frac{\sinh\pi\nu}{2^{\tfrac{d-3}{2}}\sqrt{\pi} C_{d,\nu}}
 G_{\underline{a}}
 &=
  \frac{2^{\ii\nu}}{\Gamma\big(\tfrac{d-1}{2}+\ii\nu\big)}
{\bf Z}_{\tfrac{d-2}{2},-\ii\nu}(-Z+\ii 0)
\Big( 1+a_1  +a_3\ee^{\ii\pi\big(\tfrac{d-1}{2}-\ii\nu\big)} \Big)
\\ \notag
&\quad  + \frac{2^{\ii\nu}}{\Gamma\big(\tfrac{d-1}{2}+\ii\nu\big)}
{\bf Z}_{\tfrac{d-2}{2},-\ii\nu}(-Z-\ii 0)
\Big(a_1 + a_3 \ee^{-\ii\pi\big(\tfrac{d-1}{2}-\ii\nu\big)}\Big)
-(\nu\leftrightarrow-\nu).
\end{align}
We have
${\bf Z}_{\tfrac{d-2}{2},\pm\ii\nu}(Z) \sim c Z^{-\frac{d-1}{2}\mp\ii\nu}$
as $|Z|\to\infty$, while the measure on $L^2(\dS_d,\sqrt{|g|})$ behaves
as $cZ^{d-2}$ as $|Z|\to\infty$.\footnote{This can be verified using the
global coordinates \eqref{eq:dS_global_coords}, in which $Z$ is given by
\eqref{eq:Z_dS_global_coords}.} Thus, the resolvent should, for $|Z|>1$, only contain
\begin{align}
\label{eq:conditions_resolvent}
{\bf Z}_{\tfrac{d-2}{2},\ii\nu}(|Z|) \word{if} \Im(\nu)<0 \word{and}
{\bf Z}_{\tfrac{d-2}{2},-\ii\nu}(|Z|) \word{if} \Im(\nu)>0,
\end{align}
for otherwise it could not be the integral kernel of a bounded operator
on $L^2(\dS_d,\sqrt{|g|})$.
The parameters that correspond to such a decay behavior are
different in  even and odd dimensions:

\paragraph{Odd dimensions.}
In odd dimensions, $\tfrac{d-1}{2}$ is an integer, and we obtain
\begin{align}
\label{eq:a+_odd}
 \textup{Solution } \sim {\bf Z}_{\tfrac{d-2}{2},\pm\ii\nu}(Z)
 \textup{ for } |Z|>1: &\quad
a_1 = \pm \frac{\ee^{\mp\pi\nu}}{2\sinh\pi\nu},
 \quad
 a_3= \pm \frac{(-1)^{\tfrac{d+1}{2}}}{2\sinh\pi\nu}.
\end{align}
\paragraph{Even dimensions.} In even dimensions,
$\tfrac{d-1}{2}$ is a half-integer but not an integer. We obtain
\begin{align}
\label{eq:even_inu_plusinfty}
 \textup{Solution } \sim {\bf Z}_{\tfrac{d-2}{2},\pm\ii\nu}(Z)
 \textup{ for } |Z|>1: &\quad
a_1= -\frac{\ee^{\mp\pi\nu}}{2\cosh\pi\nu}, \quad
 a_3 = - \ii \frac{(-1)^{\tfrac{d}{2}}}{2\cosh\pi\nu}.
\end{align}
These values of $a_1$ and $a_3$ yield the formulas for the resolvents.
The operator-theoretic Feynman and anti-Feynman propagators are the limits
of the resolvents on the spectrum from below resp. above.
\end{proof}

We will give an interpretation of the operator-theoretic (anti-)Feynman
propagators in terms of time-ordered two-point
functions between two states in Section \ref{sss:dS_inout}.
However, from their formulas, we can already see the surprising fact that
they are different from the propagators in the Euclidean state $\Omega_0$,
which is the only de Sitter-invariant Hadamard state.

One can ask when the Klein-Gordon operator on de Sitter space is special.
The situation is quite remarkable:

\begin{thm}
Let $\nu>0$. Then
 \begin{align}
G^{\F}_{{\rm op}} + G^{\overline{\F}}_{{\rm op}}
 &= G^{\lor} + G^{\land}, \quad\text{\bf for odd }d;\\
  \word{but}
 G^{\F}_{{\rm op}} + G^{\overline{\F}}_{{\rm op}}
  &\neq  G^{\lor} + G^{\land}, \quad\text{\bf for even }d.
\end{align}
\end{thm}
\begin{proof}
 We use the connection formula \eqref{formu1}  to rewrite
 $G^{\F}_{{\rm op}}$ and  $G^{\overline{\F}}_{{\rm op}}$
 in terms of ${\bf S}_{\frac{d-2}{2},\pm\ii\nu}(\cdot)$
 and compare to the formulas \eqref{eq:ret_adv_props_dS1}.
 Actually, in odd dimensions, the result follows immediately if one
 uses \eqref{formu3} instead of \eqref{formu1}.
\end{proof}

Let us finally consider the ``tachyonic'' region of parameters in the
de Sitter space. Instead of the parameter $\nu$, it will be
  convenient to use $\mu:=-\ii\nu$.

\begin{thm}
\begin{enumerate}
 \item {\bf Odd $d$}. The spectrum of $-\Box+\big(\tfrac{d-1}{2}\big)^2$
  equals
\begin{align}
 ]-\infty,0]\cup\big\{\mu^2\ |\ \mu\in\bN_0\},
\end{align}
and for $\mu\in[0,\infty[\setminus \bN_0$, the resolvent is given by
 \begin{align}
 \label{eq:resolvent_tachyonic_odd}
  &\quad G (\mu^2;x,x') \\ \notag
&=  -\ii \frac{\Gamma\big(\tfrac{d-1}{2}+\mu \big)}{2^{2+\mu}
 (2\pi)^{\frac{d-1}{2}}\sin\pi\mu}
 \Big({\bf Z}_{\tfrac{d-2}{2},\mu}(-Z+\ii 0)
-{\bf Z}_{\tfrac{d-2}{2},\mu}(-Z-\ii 0) \Big).
 \end{align}

 \item {\bf Even $d$}. The spectrum of $-\Box+\big(\tfrac{d-1}{2}\big)^2$
  equals
\begin{align}
 ]-\infty,0]\cup\big\{\mu^2\ |\ \mu\in\bN_0+\tfrac12\},
\end{align}
and for $\mu\in[0,\infty[\setminus \big(\bN_0+\tfrac12\big)$, the
resolvent is given by
 \begin{align}
 \label{eq:resolvent_tachyonic_even}
  &\quad G (\mu^2;x,x') \\ \notag
&=   -\frac{\Gamma\big( \tfrac{d-1}{2}+\mu \big)
 }{2^{2+\mu}(2\pi)^{\tfrac{d-1}{2}} \cos\pi\mu}
\Big({\bf Z}_{\tfrac{d-2}{2},\mu}(-Z+\ii 0)
+{\bf Z}_{\tfrac{d-2}{2},\mu}(-Z-\ii 0) \Big).
 \end{align}
\end{enumerate}
  \end{thm}

\begin{proof}
 Let $\mu>0$. If the limits of \eqref{eq:dS_op_inv_odd1} as $\nu$ approaches the imaginary line
 exist, they coincide:
 \begin{align}
  \lim_{\epsilon\to0} G((-(\ii\mu+\epsilon)^2;x,x')
  =  \lim_{\epsilon\to0} G((-(-\ii\mu+\epsilon)^2;x,x').
 \end{align}
 The results of these limits are the integral kernels of the resolvents
 in the  ``tachyonic'' case \eqref{eq:resolvent_tachyonic_odd}. Similar
 for \eqref{eq:dS_op_inv_even1} and \eqref{eq:resolvent_tachyonic_even}.

 For even $d$, the limit diverges for $\mu\in\bN_0+\tfrac12$ due to the presence of
 $\cos\pi\mu$ in the denominator of \eqref{eq:resolvent_tachyonic_even}. This is
 not a removable singularity. For $Z<-1$, we have
 \begin{align}
 {\bf Z}_{\tfrac{d-2}{2},\mu}(-Z+\ii 0)
={\bf Z}_{\tfrac{d-2}{2},\mu}(-Z-\ii 0)= {\bf Z}_{\tfrac{d-2}{2},\mu}(|Z|),
\end{align}
and this does not vanish identically.

  For odd $d$, the limit diverges for $\mu\in\bN_0$ due to the presence of
 $\sin\pi\mu$ in the denominator of \eqref{eq:resolvent_tachyonic_odd}.
 Although less obvious than in the even-dimensional case, this is also not a
 removable singularity. Due to \eqref{eq:Z+-}, we have
 \begin{align}
  {\bf Z}_{\tfrac{d-2}{2},\mu}(-Z+\ii 0)
-{\bf Z}_{\tfrac{d-2}{2},\mu}(-Z-\ii 0) =0,\quad |Z|>1,\;\mu\in\bN_0.
 \end{align}
 But using the connection formula \eqref{formu1}, we find for $|Z|<1$ and
 $\mu\in\bN_0$,
 \begin{align}
   {\bf Z}_{\tfrac{d-2}{2},\mu}(-Z+\ii 0)
-{\bf Z}_{\tfrac{d-2}{2},\mu}(-Z-\ii 0)
&= \frac{ -\ii\sgn(Z) 2^{\mu+\frac{d+1}{2}}}{\Gamma\big(\tfrac{d-1}{2}+\mu\big)
(1-Z^2)^{\frac{d-2}{2}}} {\bf S}_{\frac{d-2}{2},\mu}(-Z).
 \end{align}
This does not vanish identically.
\end{proof}

\subsection{Alpha vacua}
\label{sss:dS_alpha}

For the rest of the section on de Sitter space, we
restrict ourselves to the case of real and positive $\nu >0$.

The Euclidean vacuum is not the only de Sitter invariant state on
de Sitter space. There exists a whole family of such states, called
\emph{alpha vacua} \cite{Allen85,BMS02,M85}. We describe these
states using the Krein space language introduced in Section
\ref{ssc:props_curved} and then explain the relation to the approach
based on mode expansions, which is commonly used in the physics
literature \cite{Allen85,BMS02}.

\subsubsection{Alpha vacua in the Krein space picture}
Let $\cW_\KG$ be the Krein space of solutions of the
Klein-Gordon equation, which has a fundamental decomposition
corresponding to positive and negative frequencies with respect
to the Euclidean vacuum.
That is,
\begin{align} \label{eq:dS_Eucl_fund_deco}
\cW_\KG=\cZ_0^{(+)}\oplus \cZ_0^{(-)}, \quad
\cZ_0^{(-)}=\overline{\cZ_0^{(+)}},
\end{align}
where $\cZ_0^{(\pm)}:=\cR(\Pi_0^{(\pm)})$ are the ranges of
the orthogonal projections $\Pi_0^{(\pm)}$, whose Klein-Gordon
kernels are the bisolutions $\pm G_0^{(\pm)}$. The fundamental
decomposition \eqref{eq:dS_Eucl_fund_deco} will serve as a reference
decomposition of $\cW_\KG$.

Using the explicit representations \eqref{eq:pos_neg_props_dS}, it is
easy to see that
\begin{align}
\label{oute}
G_0^{(+)}(x^A,x^{\prime  A})=\overline{G_0^{(+)}(x,x')}=G_0^{(-)}(x,x').
  \end{align}
Introducing the map $( J^A\varphi)(x) :=\varphi(x^A)$, \eqref{oute}
implies
\begin{align}
  J^A\Pi_0^{(+)} J^A=\Pi_0^{(-)},\quad J^A(\cZ_0^{(\pm)})=\cZ_0^{(\mp)}.
\end{align}

Now let $\alpha\in\bC$ with $|\alpha|<1$. We define a Bogoliubov
transformation $R_\alpha$ on $\cW_\KG$ (i.e., a real pseudounitary map) via
\begin{align}\label{bogoliu}
(  R_\alpha\varphi)(x)=&\, \frac{1}{\sqrt{1-|\alpha|^2}} \varphi(x)
                      + \frac{\overline{\alpha}}{\sqrt{1-|\alpha|^2}}
                      \varphi(x^A),\quad \varphi\in\cZ_0^{(+)};\\
  (  R_\alpha \overline\varphi)(x)=&\, \frac{1}{\sqrt{1-|\alpha|^2}} \overline{\varphi(x) }
                      + \frac{\alpha}{\sqrt{1-|\alpha|^2}}
                        \overline{\varphi(x^A)},\quad
\overline\varphi\in\cZ_0^{(-)}.\label{bogoliu2}
\end{align}
In other words, as a $2\times2$ matrix on
$\cZ^{(+)}_0\oplus\cZ_0^{(-)}$,
\begin{align}\label{bogoliu.}
  R_\alpha=& \begin{bmatrix} \one
                      & \frac{\overline{\alpha}}{\sqrt{1-|\alpha|^2}}J^A
                      \\
 \frac{\alpha}{\sqrt{1-|\alpha|^2}}                       J^A&\frac{1}{\sqrt{1-|\alpha|^2}}
  \end{bmatrix}
\end{align}
The projections $R_\alpha \Pi_0^{(\pm)} R_\alpha^{-1}$ define another
fundamental decomposition of $\cW_\KG$, hence another Fock vacuum,
called the $\alpha$-vacuum. Their two-point functions are given by the
Klein-Gordon kernels of $\pm R_\alpha \Pi_0^{(\pm)} R_\alpha^{-1}$.
Using \eqref{bogoliu.} and
$  G_0^{(\pm)}(x^A,x')=  G_0^{(\mp)}(x,x^{\prime A})$ we obtain
\begin{align}
\label{eq:dS_Gpm_albe0}
&\quad  G^{(\pm)}_{\alpha}(x,x') \\ \notag
 &= \frac{1}{1-|\alpha|^2}
   \Big(  \tfrac{1+|\alpha|^2}{2} G_0^{\rm sym}(x,x')
   \mp \ii \tfrac{1-|\alpha|^2}{2} G^\PJ(x,x')
  +\tfrac{\alpha+\overline{\alpha}}{2} G^{\sym,A}_{0}(x,x')
   - \ii \tfrac{\alpha-\overline{\alpha}}{2}   G^{\PJ,A}_{0}(x,{x'})
  \Big).
   \end{align}
From \eqref{eq:dS_Gpm_albe0}, we obtain the well-known expressions
for the Feynman and anti-Feynman propagator
\cite{Allen85,BMS02}:\footnote{Note that the
two references have different conventions for the parameter
$\alpha$, and in addition, both conventions are different from ours.
In particular, \cite{Allen85} uses two real labels $\alpha$, $\beta$
that are both described by a single $\alpha\in\bC$ in our notation.}
\begin{align}
\label{eq:dS_GF_al...}
 &\quad G^{\F/\overline{\F}}_{\alpha}(x,x')
  \\ \notag &= G_0^{\F/\overline{\F}}(x,x')
  \pm\frac{\ii }{1-|\alpha|^2}
  \Big(
   |\alpha|^2G_0^{\rm sym}(x,x')
  +\ii \tfrac{\alpha+\overline{\alpha}}{2} G^{\sym,A}_{0}(x,x')
  -\ii \tfrac{\alpha-\overline{\alpha}}{2}   G^{\PJ,A}_{0}(x,{x'})
  \Big).
\end{align}
It is known that only the $\alpha$-vacuum satisfying
the Hadamard condition is the Euclidean vacuum, that is, corresponding
to $\alpha=0$ (see \cite{Allen85} and references therein).
This can also be read off the expansion of the Gegenbauer function
around the singularity.

From the point of view of perturbative QFT, the usefulness of alpha
vacua for $\alpha\neq0$ is therefore questionable. It is not clear how
one can renormalize quantities that are local and
non-linear in the fields
\cite{BFH05}. However, they are reasonable objects in \emph{linear}
QFT and possibly also in an effective field-theory. We shall see that
the operator-theoretic propagators correspond
to field expectation values in specific alpha vacua.

\subsubsection{Alpha vacua and mode expansions}
In the literature $\alpha$-vacua are often introduced as follows
\cite{Allen85,BMS02}. First one expands the real scalar Klein-Gordon
field $\hat{\phi}(x)$ into modes with respect to the Euclidean vacuum,
\begin{align}
\label{eq:dS_mode_sum}
\hat \phi(x) &= \sum_n \varphi_n(x)\hat a_n^\ast +
\overline{\varphi_n(x)} \hat a_n.
\end{align}
 Here, $\hat a_n$ and $\hat a_n^\ast$ are  annihilation and creation operators and $\varphi_n(x)$ are mode functions  that satisfy the
 orthogonality relations \eqref{eq:modes_orthogonal} with the Dirac
 delta replaced by the Kronecker delta. This is essentially a choice
 of an orthonormal basis of the space $\cZ_0^{(+)}$. The positive
 frequency solution can then be written as a mode sum,
 \begin{align}
 G^{(+)}_0(x,x')
 &=\sum_{n} \overline{\varphi_n(x)} \varphi_n(x' ) .
\end{align}

Next, using the explicit form of the modes, one shows \cite{Allen85,BMS02}
that the modes associated to the Euclidean vacuum can be chosen to satisfy
\begin{align}
\label{eq:convention_Eucl_modes}
 \varphi_n(x) = \overline{\varphi_n(x^A)}.
\end{align}

Then one defines the Bogoliubov transformation \eqref{bogoliu}
by its action on the modes,
\begin{align}
\label{eq:MA_trf}
 \varphi_{\alpha,n}(x) :=&\, \frac{1}{\sqrt{1-|\alpha|^2}} \varphi_n(x)
 + \frac{\overline{\alpha}}{\sqrt{1-|\alpha|^2}} \overline{\varphi_n(x)},
\end{align}
and the positive frequency solution associated to the alpha
vacuum with parameter $\alpha$ is given by
 \begin{align}
 G^{(+)}_\alpha(x,x')
 &=\sum_{n} \overline{\varphi_{\alpha,n}(x)} \varphi_{\alpha,n}(x' ) .
 \end{align}

 Needless to say, the construction using the mode expansion and the
 construction based on \eqref{bogoliu.} are equivalent.
 In particular, $\varphi_{\alpha,n}=R_\alpha\varphi_n$.

\subsubsection{Correlation functions between two different alpha vacua}
Suppose now that $\alpha,\beta$ be two complex parameter with
$|\alpha|,|\beta|<1$ and consider a pair of Bogoliubov
transformations $R_\alpha$, $R_\beta$ and a pair of Fock vacua
$\Omega_\alpha$, $\Omega_\beta$. Using modes, we can write
\begin{align}
\label{eq:MA_trf-}
 \varphi_{\beta,n}(x) :=&\,
 N_{\alpha,\beta} \varphi_{\alpha,n}(x)
 + M_{\alpha,\beta}
 \overline{\varphi_{\alpha,n}(x)},
 \\ \notag
 N_{\alpha,\beta} =&\,
                    \frac{1-\overline{\beta}\alpha}{\sqrt{(1-|\alpha|^2)(1-|\beta|^2)}},
\\ \notag
M_{\alpha,\beta} =&\,\frac{\overline{\beta}-\overline{\alpha}}{\sqrt{(1-|\alpha|^2)(1-|\beta|^2)}} .
\end{align}
Note that this definition is a special case of the more general form
\eqref{eq:Bogo}. It relates to the latter equation via
\begin{align}
 N_{\alpha,\beta}=N_{\alpha,\beta}(n),\quad
 M_{\alpha,\beta}\delta_{n,m}=\Lambda_{\alpha,\beta}(n,m).
\end{align}
Therefore, we may use \eqref{eq:G_albe_def} to obtain the mixed
two-point functions
\begin{align}
  &\quad  G^{(\pm)}_{\alpha,\beta}(x,x')
 \label{eq:dS_Gpm_albe}
\\ \notag
 &= \frac{1}{1-\overline{\beta}\alpha}
   \Big(  \tfrac{1+\alpha\overline{\beta}}{2} G_0^{\rm sym}(x,x')
   \mp \ii \tfrac{1-\alpha\overline{\beta}}{2} G^\PJ(x,x')
  +\tfrac{\alpha+\overline{\beta}}{2} G^{\sym,A}_{0}(x,x')
   - \ii \tfrac{\alpha-\overline{\beta}}{2}   G^{\PJ,A}_{0}(x,{x'})
   \Big).
   \end{align}
The corresponding Feynman and anti-Feynman propagator are
\begin{align}
\label{eq:dS_GF_albe}
 &\quad G^{\F/\overline{\F}}_{\alpha,\beta}(x,x')
  \\ \notag &= G_0^{\F/\overline{\F}}(x,x')
  \pm\frac{\ii }{1-\overline{\beta}\alpha}
  \Big(
   \alpha\overline{\beta} G_0^{\rm sym}(x,x')
  +\ii \tfrac{\alpha+\overline{\beta}}{2} G^{\sym,A}_{0}(x,x')
  -\ii \tfrac{\alpha-\overline{\beta}}{2}   G^{\PJ,A}_{0}(x,{x'})
  \Big).
\end{align}

\subsection{``In'' and ``out'' vacua}
\label{sss:dS_inout}

The de Sitter space is not asymptotically stationary. Therefore, the
usual definition of ``in'' and ``out'' vacua is not
applicable. Nevertheless, one can define a pair of de Sitter invariant states that deserve
to be called the ``in'' and ``out'' vacuum. In this subsection we will
 compute the corresponding propagators.

Every bisolution of the Klein-Gordon equation is  a linear
combination of appropriately regularized functions ${\bf
  Z}_{\frac{d-2}{2},\ii\nu}(Z)$ and
${\bf
  Z}_{\frac{d-2}{2},-\ii\nu}(Z)$. They behave for large $Z$
proportionally to $Z^{-\frac{d-1}{2}-\ii\nu}$,
resp. $Z^{-\frac{d-1}{2}+\ii\nu}$. We are looking for two-point
functions, which in the ``causal asymptotic region'', that is for $Z\to\infty$ and $t\to\pm\infty$, have a {\em
  definite behavior}, that is, they behave either as
$cZ^{-\frac{d-1}{2}-\ii\nu}$,
or as $cZ^{-\frac{d-1}{2}+\ii\nu}$.

Note that the propagators have also the ``antipodal asymptotic region'':
$Z\to-\infty$, $t^A\to\pm\infty$. It will be interesting to determine
their behavior in that region as well.

The following theorem describes all de Sitter invariant two-point
functions with a definite behavior in the causal asymptotic region.

\begin{thm}
\label{thm:dSinout}
  \begin{enumerate}
    \item {\bf Odd dimensions.}  There exists a unique $\alpha$-vacuum with the propagators behaving as
      \begin{align}
      \label{eq:dS_timelike_behav}
        G_\alpha^{(\pm)}\sim cZ^{-\frac{d-1}{2}\pm\ii\nu},&\quad
                                                             Z\to+\infty,\quad
                                                   t\to -\infty;
\\\text{ and }\quad
G_\alpha^{(\pm)}\sim cZ^{-\frac{d-1}{2}\mp\ii\nu},&\quad
                                                             Z\to+\infty,\quad
                                                             t\to+\infty. \label{eq:dS_timelike_behav2}
      \end{align}
      These functions vanish for $Z<-1$ and their parameter $\alpha$ is
\begin{align}
\label{eq:alpha_inout_odd}
\alpha_-=\alpha_+= \alpha_\as := (-1)^{\tfrac{d+1}{2}} \ee^{-\pi\nu}
 = \ee^{-\pi\nu \pm \ii\pi\tfrac{d+1}{2}}.
\end{align}
                 This vacuum could be called the ``in'' vacuum or the
                 ``out'' vacuum. We will call it the {\em asymptotic vacuum}.       We will write $\as$ instead of
                                    $\alpha_\as$ in the subscripts of
                                    propagators and two-point
                                    functions. The two point functions
                                    of these states are
\begin{align}
 &\quad
 \frac{\ii \sinh\pi\nu}{2^{\tfrac{d-3}{2}}\sqrt{\pi} C_{d,\nu}}
 G^{(\pm)}_{\as}(x,x')
 \\ \notag  &= \frac{2^{-\ii\nu} \theta(\pm t)
 }{\Gamma\big( \tfrac{d-1}{2}-\ii\nu \big)}
\Big( {\bf Z}_{\tfrac{d-2}{2},\ii\nu}(-Z-\ii 0)
- {\bf Z}_{\tfrac{d-2}{2},\ii\nu}(-Z+\ii 0) \Big)
\\ \notag
&\quad  + \frac{2^{\ii\nu}\theta(\mp t)
}{\Gamma\big(\tfrac{d-1}{2}+\ii\nu\big)}
\Big( {\bf Z}_{\tfrac{d-2}{2},-\ii\nu}(-Z-\ii 0)
- {\bf Z}_{\tfrac{d-2}{2},-\ii\nu}(-Z+\ii 0)
\Big),
\end{align}
and their Feynman and anti-Feynman propagators coincide
with the operator-theoretic ones from \eqref{eq:dS_op_inv_odd1}:
\begin{align}
 G^\F_\as(x,x') = G^\F_{{\rm op}}(x,x'),\quad
 G^{\overline{\F}}_\as(x,x')
 = G^{\overline{\F}}_{{\rm op}}(x,x').
\end{align}
\item{\bf Even dimensions.}
There exist two $\alpha$-vacua that satisfy
\eqref{eq:dS_timelike_behav} and \eqref{eq:dS_timelike_behav2}.
One of the two values is
\begin{align}
 \label{eq:alpha_inout_even}
 \alpha_- =\ii\ee^{-\pi\nu}(-1)^{\frac{d}{2}}= \ee^{-\pi\nu +\ii\pi \frac{d+1}{2}}
\end{align}
and its positive/negative frequency solutions vanish for $Z<-1$, $t^A<0$.
It will be called the ``in'' vacuum.

 The other value is
 \begin{align}
 \alpha_+ =-\ii\ee^{-\pi\nu}(-1)^{\frac{d}{2}}= \ee^{-\pi\nu -\ii\pi \frac{d+1}{2}}
 = - \alpha_-
 \end{align}
 and its positive/negative frequency solutions vanish for $Z<-1$, $t^A>0$.
It will be called the ``out'' vacuum.

 We will write $-$,
resp. $+$
instead of $\alpha_-$ and $\alpha_+$ in subscripts.
The two-point functions of these states are
\begin{align}
&\quad \frac{\ii\sinh\pi\nu}{2^{\tfrac{d-3}{2}}\sqrt{\pi} C_{d,\nu}}
 G^{(\pm)}_-(x,x') \\ \notag
&=
 \frac{2^{-\ii\nu}\theta(\pm t)
 }{\Gamma\big( \tfrac{d-1}{2}-\ii\nu \big)}
\Big({\bf Z}_{\tfrac{d-2}{2},\ii\nu}(-Z-\ii 0)
-{\bf Z}_{\tfrac{d-2}{2},\ii\nu}(-Z+\ii 0)   \Big)
\\ \notag
&\quad + \frac{2^{\ii\nu}\theta(\mp t)
}{\Gamma\big(\tfrac{d-1}{2}+\ii\nu\big)}
\Big( {\bf Z}_{\tfrac{d-2}{2},-\ii\nu}(-Z-\ii 0)
- {\bf Z}_{\tfrac{d-2}{2},-\ii\nu}(-Z+\ii 0) \Big)
\\ \notag
&\quad +
\frac{ (-1)^{\tfrac{d}{2}}\theta(t^A) }{2^{\tfrac{d-3}{2}} \sqrt{\pi}}
\Big(   {\bf S}_{\tfrac{d-2}{2},\ii\nu}(Z+\ii 0)
-{\bf S}_{\tfrac{d-2}{2},\ii\nu}(Z-\ii 0)\Big),
\end{align}
and
\begin{align}
&\quad \frac{\ii\sinh\pi\nu}{2^{\tfrac{d-3}{2}}\sqrt{\pi} C_{d,\nu}}
 G^{(\pm)}_+(x,x') \\ \notag
&=
 \frac{2^{-\ii\nu}\theta(\pm t)
 }{\Gamma\big( \tfrac{d-1}{2}-\ii\nu \big)}
\Big({\bf Z}_{\tfrac{d-2}{2},\ii\nu}(-Z-\ii 0)
-{\bf Z}_{\tfrac{d-2}{2},\ii\nu}(-Z+\ii 0)   \Big)
\\ \notag
&\quad + \frac{2^{\ii\nu}\theta(\mp t)
}{\Gamma\big(\tfrac{d-1}{2}+\ii\nu\big)}
\Big( {\bf Z}_{\tfrac{d-2}{2},-\ii\nu}(-Z-\ii 0)
- {\bf Z}_{\tfrac{d-2}{2},-\ii\nu}(-Z+\ii 0) \Big)
\\ \notag
&\quad +
\frac{ (-1)^{\tfrac{d}{2}}\theta(-t^A) }{2^{\tfrac{d-3}{2}} \sqrt{\pi}}
\Big(   {\bf S}_{\tfrac{d-2}{2},\ii\nu}(Z+\ii 0)
-{\bf S}_{\tfrac{d-2}{2},\ii\nu}(Z-\ii 0)\Big).
\end{align}
The in-out Feynman  and the out-in anti-Feynman propagator coincide
with the operator-theoretic Feynman and anti-Feynman  propagator
\eqref{eq:dS_op_inv_even1.}, resp.:
\begin{align}
\label{eq:outin_F_even}
 G^{\F}_{+-}
 &=  G^{\F}_\op ;\quad
   G^{ \overline\F}_{-+}
 = G^{ \overline\F}_\op .
\end{align}
\end{enumerate}
\end{thm}

\begin{remark} The concrete values for $\alpha$ corresponding to ``in'' and ``out''
states are well-known \cite{M85,BMS02} but typically derived by asymptotic
properties of the modes. We derive them in the following using a ``global
picture''.
\end{remark}
\begin{proof}[Proof of Thm \ref{thm:dSinout}]
 We use \eqref{eq:dS_Gpm_albe0} to express a
 generic $G^{(\pm)}_{\alpha}$ in terms of ${\bf Z}_{\frac{d-2}{2},\ii\nu}$
 and ${\bf Z}_{\frac{d-2}{2},-\ii\nu}$:
\begin{align}
\label{eq:gen_bisol_Z}
 &\quad 2\ii(1-|\alpha|^2) \frac{\sinh\pi\nu}{2^{\tfrac{d-3}{2}}\sqrt{\pi} C_{d,\nu}}
 G^{(\pm)}_{\alpha}
\\ \notag
 &=
 -\frac{2^{-\ii\nu}}{\Gamma\big( \tfrac{d-1}{2}-\ii\nu \big)}
{\bf Z}_{\tfrac{d-2}{2},\ii\nu}(-Z+\ii 0)
\Big( \big(1+|\alpha|^2 \pm (1-|\alpha|^2)\sgn(t)\big)
\\ \notag & \hspace{32ex}
+ \big(2 \Re(\alpha) + 2\ii \Im(\alpha)\sgn(t^A) \big) \ee^{\ii\pi\big(\tfrac{d-1}{2}+\ii\nu\big)} \Big)
\\ \notag &\quad -\frac{2^{-\ii\nu}
}{\Gamma\big( \tfrac{d-1}{2}-\ii\nu \big)}
{\bf Z}_{\tfrac{d-2}{2},\ii\nu}(-Z-\ii 0)
\Big(\big(1+|\alpha|^2 \mp (1-|\alpha|^2)\sgn(t)\big)
\\ \notag & \hspace{32ex}
+\big(2 \Re(\alpha) - 2\ii \Im(\alpha)\sgn(t^A) \big)\ee^{-\ii\pi\big(\tfrac{d-1}{2}+\ii\nu\big)}\Big)
\\ \notag &\quad
- (\nu\leftrightarrow-\nu).
\end{align}
The analysis of the asymptotic behavior of the latter function differs
in odd and even dimensions. We only display the derivation of the more
complicated even-dimensional case. The odd-dimensional case can be worked out
analogously:

\paragraph{Even dimensions}
The conditions on the asymptotic behavior read:
\begin{subequations}
 \label{eq:alpha_conditions_even}
\begin{align}
\label{eq:alpha_behav1+even}
 \textup{Solution }  &G^{(\pm)}_\alpha\sim {\bf Z}_{\tfrac{d-2}{2},\ii\nu}(Z)
 \textup{ for } Z>1: \\ \notag
 &\mp (-1)^{\tfrac{d-2}{2}} 2 \ii \sgn(t) \Re(\alpha)
=  \ee^{\mp \sgn(t)\pi\nu} - |\alpha|^2 \ee^{\pm \sgn(t) \pi\nu} ,
 \\ \label{eq:alpha_behav1-even}
 \textup{Solution }  &G^{(\pm)}_\alpha\sim {\bf Z}_{\tfrac{d-2}{2},-\ii\nu}(Z)
 \textup{ for } Z>1: \\ \notag
 &\mp (-1)^{\tfrac{d-2}{2}} 2 \ii \sgn(t) \Re(\alpha)
=  \ee^{\pm \sgn(t)\pi\nu} - |\alpha|^2 \ee^{\mp \sgn(t) \pi\nu} ,
 \\ \label{eq:alpha_behav2+even}
 \textup{Solution }  &G^{(\pm)}_\alpha\sim {\bf Z}_{\tfrac{d-2}{2},\ii\nu}(-Z)
 \textup{ for } Z<-1: \\ \notag
 & (-1)^{\tfrac{d-2}{2}} (1+|\alpha|^2) =
- 2\ii \Re(\alpha) \sinh\pi\nu + 2 \Im(\alpha) \sgn(t^A) \cosh\pi\nu,
 \\ \label{eq:alpha_behav2-even}
 \textup{Solution }  &G^{(\pm)}_\alpha\sim{\bf Z}_{\tfrac{d-2}{2},-\ii\nu}(-Z)
 \textup{ for } Z<-1: \\ \notag
 & (-1)^{\tfrac{d-2}{2}} (1+|\alpha|^2)
= 2\ii \Re(\alpha) \sinh\pi\nu + 2 \Im(\alpha) \sgn(t^A) \cosh\pi\nu.
\end{align}
\end{subequations}
We immediately read off $\Re(\alpha)=0$. Then, by \eqref{eq:alpha_behav1+even}
and \eqref{eq:alpha_behav1-even}, the existence of a definite behavior in
the region $Z>1$ implies $|\alpha|=\ee^{-\pi\nu}$. Hence
$\alpha=\ee^{\ii\pi(n+\tfrac12)-\pi\nu}$ with $n\in \bZ$.
Then  \eqref{eq:alpha_behav2+even} and \eqref{eq:alpha_behav2-even} simplify
to
\begin{align}
 (-1)^{\tfrac{d-2}{2}-n}
=   \sgn(t^A) .
\end{align}
$n=\frac{d-2}{2}$ yields a solution that vanishes for $Z<-1$ and $t^A>0$ but has
indeterminate behavior as $Z<-1$ and $t^A\to-\infty $, while $n=\frac{d}{2}$
yields a solution that vanishes for $Z<-1$ and $t^A<0$ but has
indeterminate behavior as $Z<-1$ and $t^A\to+\infty $. We obtain the values for
$\alpha_+$ and $\alpha_-$.

Inserting the obtained values for $\alpha$ into \eqref{eq:gen_bisol_Z} yields the explicit formulas for
$G^{(\pm)}_{\rm \pm}$:
this rather cumbersome computation involves the
connection formula \eqref{formu3}, the identity
\eqref{eq:Z+-} and repeated use of identities of the type
$1\pm\sgn(\cdot)=2\theta(\pm \cdot)$.
The (anti-)Feynman
propagators are obtained from \eqref{eq:dS_GF_albe} and also using the connection
formulas.
\end{proof}

\subsection{Symmetric Scarf Hamiltonian}
We will discuss in the next subsection another approach to
the Klein-Gordon equation on the de Sitter space.
In this approach we will use the one-dimensional Schrödinger Hamiltonian  on
$L^2(\mathbb{R})$ of the form
\begin{align}
\label{eq:H_scarf}
 H_\alpha^\mathrm{S} := - \partial_\tau^2 - \frac{\alpha^2-\tfrac14}{\cosh(\tau)^2}.
\end{align}
It is sometimes called
\emph{symmetric Scarf  Hamiltonian} \cite{DeWroch}.
It is well-known that this Hamiltonian for some values of
  parameters is reflectionless. For completeness, let us verify this.

First  we check that $H_\alpha^\mathrm{S} + \lambda^2$ is
equivalent to the Gegenbauer equation after the consecutive change of
variables $\sinh \tau = w$, $\ii w= v$:
\begin{align}
 &\quad \cosh(\tau)^{-\alpha-\tfrac12} \big(H_\alpha^\mathrm{S} + \lambda^2\big)
 \cosh(\tau)^{\alpha+\tfrac12}
 \\ \notag
 &= - \partial_\tau^2  - (2\alpha+1) \tanh(\tau) \partial_\tau
        - \big(\alpha+\tfrac12\big)^2 + \lambda^2
 \\ \notag
 &= -(1+w^2) \partial_w^2 - 2(\alpha+1) w \partial_w
 - \big(\alpha+\tfrac12\big)^2 + \lambda^2
 \\ \notag
 &= (1-v^2) \partial_v^2 - 2(\alpha+1) v \partial_v
 - \big(\alpha+\tfrac12\big)^2 + \lambda^2.
\end{align}
For $\Re(\lambda)>0$, the Jost solutions can thus be expressed
in terms of the Gegenbauer $Z$-function:
\begin{align}
\label{eq:Scarf_Jost}
 \psi_\pm(\lambda,\tau)
 &= 2^{\mp\lambda}\Gamma(1\pm\lambda)
 \ee^{\ii\tfrac{\pi}{2}(\tfrac12+\alpha\pm\lambda)} \cosh(\tau)^{\alpha+\tfrac12} Z_{\alpha,\pm\lambda}
 (\pm \ii \sinh\tau),
\end{align}
such that
\begin{align}
 \psi_\pm(\lambda,\tau)\sim \ee^{\mp\lambda \tau},
 \quad \pm \tau \to \infty.
\end{align}

The Gegenbauer functions on the righthand-side of
\eqref{eq:Scarf_Jost} have purely imaginary arguments. They are
to be interpreted as living on the cut plane $\bC\setminus\big(
]-\infty,-1] \cup [1,\infty[\big)$ instead of the usual
$\bC\setminus]-\infty,1]$. $\psi_+(\lambda,\cdot)$ is expressed in
terms of the analytic continuation of $Z_{\alpha,\lambda}(w)$
defined on the standard sheet $\bC\setminus]-\infty,1]$ \emph{to the
upper half-plane}, while $\psi_-(\lambda,\cdot)$ is expressed in terms
of the analytic continuation of $Z_{\alpha,-\lambda}(w)$
defined on the standard sheet $\bC\setminus]-\infty,1]$ \emph{to the
lower half plane}. Using the connection formulas \eqref{formu1} and
\eqref{formu3}, and the fact that ${\bf S}_{\alpha,\lambda}$ is
holomorphic on $]-1,1[$, one can derive a connection formula for the
two holomorphic continuations:
\begin{align}
 &\quad{\bf Z}_{\alpha,\lambda}(w+\ii0)
 \\ \notag
 &= \frac{\ii\cos\pi\alpha \,\ee^{-\ii\pi(\alpha+\lambda)}
 {\bf Z}_{\alpha,\lambda}(w-\ii0) }{\sin\pi\lambda}
 -\frac{\ii 2^{2\lambda} \ee^{-\ii\pi\alpha} \pi
 {\bf Z}_{\alpha,-\lambda}(w-\ii0)
 }{\Gamma\big(\tfrac12+\alpha+\lambda\big)
 \Gamma\big(\tfrac12-\alpha+\lambda\big) \sin\pi\lambda},
 \quad w\in]-1,1[.
\end{align}
In particular, ${\bf Z}_{\alpha,\lambda}(w+\ii0)$ is proportional
to ${\bf Z}_{\alpha,-\lambda}(w-\ii0)$ if and only if $\cos\pi\alpha=0$,
i.e., if and only if $\alpha\in\bZ+\tfrac12$.

Consequently, the symmetric Scarf Hamiltonian is reflectionless for
all energies $\nu^2$ iff
$\alpha\in\bZ+\tfrac12$.

\subsection{Partial wave decomposition}
\label{sss:dS_Scarf}

Using the global system of coordinates
\eqref{eq:dS_global_coords}, the de Sitter space can be viewed as a
FLRW space, and can be identified with $\mathbb{R}\times\mathbb{S}^{d-1}$.
In these coordinates,
the (gauged) Klein-Gordon operator takes the form
\begin{align}
\label{eq:KG_gauged_scarf}
&\quad \cosh(\tau)^{\tfrac{d-1}{2}}
 \big(-\square_g + m^2\big) \cosh(\tau)^{-\tfrac{d-1}{2}}
 \\ \notag
&= \partial_\tau^2 -\frac{d-1}{2} \Big(
 1 +  \frac{(d-3)\sinh(\tau)^2}{2\cosh(\tau)^2}
 \Big)  -  \frac{\Delta_{\bS^{d-1}}}{\cosh(\tau)^2} + m^2
  \\ \notag
&= \partial_\tau^2  +
\frac{\big(\tfrac{d-2}{2}\big)^2-\tfrac14
-\Delta_{\bS^{d-1}}}{\cosh(\tau)^2} +\nu^2
\end{align}
The spectrum of $-\Delta_{\bS^{d-1}}$ is
$\{ l(l+d-2)\;|\;l\in\bN_0\}$. Hence, restricted to eigenfunctions
with eigenvalue $l(l+d-2)$, the above operator becomes
$-H_\alpha^\mathrm{S}+\nu^2$, where $H_\alpha^\mathrm{S}$ is the
\emph{symmetric Scarf Hamiltonian}
with $\alpha=l+\tfrac{d-2}{2}$.
The symmetric Scarf potential is reflectionless for
all energies $\nu^2\in\bR$ and
$\alpha\in\frac12+\mathbb{Z}$. This corresponds to odd
dimensions. Thus  for each mode the in-state coincides with the out-state.
In even dimensions $\alpha\in\mathbb{Z}$, and then for each mode the in-state
is different from the out-state.


\section{Anti-de Sitter space and its universal cover}
\label{sec:AdS}
Our final examples of Lorentzian manifolds are the $d$-dimensional anti-de Sitter space
$\AdS_d$ and its universal covering  $\tAdS_d$.

$\AdS_d$ is pathological from several points of view.
First of all, it has time loops, which makes it unsuitable as a model
of a spacetime. It does not make much sense to speak about propagators
on $\AdS_d$.

The cyclicity of time can be cured by replacing the proper
anti-de Sitter space by its universal cover $\tAdS_d$. It is still
not globally hyperbolic, because of a
boundary with a spacelike normal at spacelike
infinity. However the latter problem is not very serious, and various
propagators can be defined on $\tAdS_d$.

Therefore, most of this section will be devoted to $\tAdS_d$. We will
apply two methods to define propagators: through the resolvent of the
d'Alembertian on $L^2(\tAdS_d)$, and by considering the evolution of
the Cauchy data. The latter approach is facilitated by the fact that
$\tAdS_d$ is static. The absence of global hyperbolicity is not a
problem for the first approach. For the second approach it manifests
itself by the need to set boundary conditions at the spatial infinity
for  $m^2$ below a certain value.


Various propagators of massive scalar fields on
$\tAdS_d$ have been intensively studied. Among the vast literature,
we mention the references \cite{DullemontvB85, AvisIS78, RumpfdS,
IW04, BurgessL85,AMP18, AAK21, BEGMP12, BEM02,DF16,DFJ18}, which
are particularly useful for understanding the analytic structure.
Similar to the de Sitter example, the only of these references using
the operator-theoretic view on the Feynman propagator is \cite{RumpfdS}
(here in two dimensions). The references \cite{BEGMP12,BEM02} have an
axiomatic approach. Appendix A of \cite{AMP18} is particularly helpful
to understand the analytic structure of propagators on the universal cover.
 Subsection \ref{sss:AdS_bc}, where we present the approach
based on the evolution of Cauchy data, is based on the seminal work
\cite{IW04}.

\subsection{Geometry of anti-de Sitter space}
\label{ssc:AdS_geometry}
 The $d$-dimensional anti-de Sitter space $\AdS_d$ can be
defined as an embedded submanifold of $\bR^{2,d-1}$:
\begin{align}
\label{eq:AdS_def}
 \AdS_d = \{ x \in \bR^{2,d-1}\;|\; \langle x | x\rangle =-1 \},
\end{align}
where
\begin{align}
 \langle x|x' \rangle &:=  -x^0 {x'}^0  - x^d {x'}^d
+ \sum_{i=1}^{d-1} x^i {x'}^i =: Z(x,x') \equiv Z.
\end{align}

A coordinate system covering all of $\AdS_d$ is given by
\begin{align}
\label{eq:AdS_global_coords.}
 &x^0 =  \cosh\rho\cos\tau, \quad x^i=\sinh\rho\; \Omega^i,\;
 \quad  x^d =\cosh\rho\sin\tau,
 \\ \notag
 \textup{where}\quad  &\tau\in[-\pi,\pi[,\;\rho\in\bR_{\geq0},
 \;\Omega\in\bS_{d-2}\hookrightarrow\bR^{d-1}
 \word{and} i=1,\dots,d-1.
\end{align}
In these coordinates, the line element reads
\begin{align}
\label{eq:metric_AdS_global}
 \dd s^2 = -\cosh(\rho)^2\dd\tau^2 + \dd \rho^2
 + \sinh(\rho)^2 \dd \Omega^2.
\end{align}
Note the famous cyclicity of time,
$x(\tau+2\pi k,\rho,\Omega)=x(\tau,\rho,\Omega)$ for all $k\in\bZ$.
Therefore, $\AdS_d$ has
closed timelike curves and is not globally hyperbolic.

$\AdS_d$ is equipped with an involution $x\mapsto -x$.
This involution maps the coordinates $(\tau,\rho,\Omega)$ to
$(\tau+\pi,\rho,-\Omega)$.

Another system of coordinates is obtained by replacing $\rho$ with
$u\in[0,\tfrac{\pi}{2}[$, where
$\sinh\rho=\tan u$. In these coordinates, the line element
\eqref{eq:metric_AdS_global} becomes
\begin{align}\label{orthis}
 \dd s^2
 =& \frac{ -\dd \tau^2 + \dd u^2 + \sin(u)^2
  \,\dd\Omega^2}{\cos(u)^2}.
\end{align}

In the coordinates \eqref{eq:AdS_global_coords.} and \eqref{orthis}, we find
\begin{align}\label{lineel}
 Z
&   = -\cosh\rho\cosh\rho' \cos(\tau-\tau')
                         + \sinh\rho\sinh\rho' \cos\theta\\
& = -\frac{ \cos(\tau-\tau') }{\cos u\cos u'}
                         + \frac{\sin u\sin u' \cos\theta}{\cos u\cos u'}.
\end{align}
where $\theta$ is the angle between $\Omega$ and $\Omega'$.

Let us fix the vector $x'=(1,0\dots,0)$. Then $-\langle
x|x'\rangle=\frac{\cos\tau}{\cos u}$ and we can partition $\AdS_d$
into the following regions:
\begin{subequations}
\begin{align}
  V_0&:=\{|\tau|<u\},\\
V_2&:=\{\pi-\tau<u\}\cup\{\pi+\tau<u\},\\
  V_1&:=\{\min(\tau,\pi-\tau)>u\},\\
  V_{-1}&:=\{\min(-\tau,\pi+\tau)>u\}.
\end{align}
\end{subequations}
Note that
\begin{subequations}
\begin{align}
  -Z >1, &\quad\tau\in[-\tfrac\pi2,\tfrac\pi2]
  &&\text{ on }V_0,\\
  Z > 1, &\quad \tau\in[-\pi,-\tfrac\pi2]\cup
  [\tfrac\pi2,\pi]
              &&\text{ on }V_2,\\
  |Z| <1, &\quad \tau\in[0,\pi] &&\text{ on }V_1,\\
  |Z| <1, &\quad \tau\in[-\pi,0] &&\text{ on }V_{-1}.
  \end{align}
\end{subequations}

The Klein-Gordon equation on anti de Sitter space reads
  \beq(-\Box+m^2)\phi(x)=0.\label{kgo1a}\eeq
Instead of $m$ we will use the parameter $\nu$
\begin{align}
\label{eq:AdS_spectral_parameter}
 \nu := \sqrt{m^2+\big(\tfrac{d-1}{2}\big)^2},
\end{align}
where as usual we use the principal branch of the square root.
Thus \eqref{kgo1a} is replaced with
\beq\Big(-\Box-\big(\tfrac{d-1}{2}\big)^2+\nu^2\Big)\phi(x)=0.\label{kgoa}\eeq

The Klein-Gordon equation on anti de Sitter space
restricted to invariant
solutions and written in terms of $Z$ reduces to the Gegenbauer
equation, where the sign in front of $\nu^2$ is opposite from de
Sitter space:
\begin{align}
\Big((1-Z^2)\del_Z^2 - d Z \del_Z +\nu^2 - \big(\tfrac{d-1}{2}\big)^2\Big) f(Z)=0.
\end{align}

\subsection{Universal cover of anti de Sitter space}

$\AdS_d$  has
the topology of $\bS_1\times\bR^{d-1}$. Therefore it has a universal covering
space
\begin{align}
\label{quo}
\tAdS_d\to\AdS_d.
\end{align}
 In the literature, this universal cover is sometimes called
anti-de Sitter space instead \cite{HawkingEllis}. We will, however,
use the name  anti-de Sitter space for the embedded
submanifold \eqref{eq:AdS_def}, adding the adjective
``proper'' whenever we think it is necessary to avoid confusion.

It is easy to describe  $\tAdS_d$ in
coordinates: we just assume that $\tau\in\mathbb{R}$, and keep the
line element \eqref{eq:metric_AdS_global} or \eqref{orthis}.
 $\tAdS_d$ is a static Lorentzian manifold. It is still not
globally hyperbolic, since there are geodesics, which in finite time
escape to its boundary.

Let us fix the vector $x'=(1,0\dots,0)$. Then we can partition $\tAdS_d$
into the following regions:
\begin{align}
  V_{2n}&:=\{|\tau-n\pi|<u\},\\
  V_{2n+1}&:=\{\min(\tau-n\pi,(n+1)\pi-\tau)>u\}.
\end{align}
Note that
\begin{align}
  -(-1)^n Z &>1,\quad\tau\in[(n-\tfrac12)\pi,(n+\tfrac12)\pi]
  &&\text{ on }V_{2n}\\
    |Z|&<1,\quad\tau\in[n\pi,(n+1)\pi]
    &&\text{ on }V_{2n+1}.
  \end{align}

The spaces $\AdS_d$ and $\tAdS_d$ with their various regions are
depicted in Figure \ref{fig:AdS}.

\begin{figure}[ht]
\floatbox[{\capbeside\thisfloatsetup{capbesideposition={right,center}}}]{figure}[\FBwidth]
{\caption{(a) Anti-de Sitter space in the coordinates $u\in[0,\tfrac{\pi}{2}[$ and $\tau\in]-\pi,\pi]$ from \eqref{orthis} and its partition into the regions $V_0$, $V_2$, $V_1$ and $V_{-1}$.
Each point represents a $d-2$-sphere of the coordinates
$\Omega$. The lines $\tau=\pi$ and $\tau=-\pi$ are glued together,
reflecting the cyclicity of time. An
observer can reach spatial infinity ($u=\tfrac{\pi}{2}$, indicated by the dashed line) in finite
time, which makes it necessary to impose boundary conditions
when solving the Cauchy problem for certain masses, see Section
\ref{sss:AdS_bc}. (b) The universal cover of
anti-de Sitter space in the same coordinates, where however $\tau$
ranges over all of $\bR$, removing the cylicity of time.
The  boundary at $u=\tfrac{\pi}{2}$ is still present.}\label{fig:AdS}}
{\scalebox{0.8}{\begin{tikzpicture}[scale=0.5, domain=-3:5]
\begin{scope}
\path[fill=gray, opacity=0.2] (0,-8) -- (4,-4) -- (0,0) -- cycle;
\end{scope}
\begin{scope}
\path[fill=gray, opacity=0.2] (0,0) -- (4,4) -- (0,8) -- cycle;
\end{scope}

\begin{scope}
\path[fill=gray, opacity=0.2, pattern= north east lines] (0,8) -- (4,8) -- (4,8.5) -- (0,8.5) -- cycle;
\end{scope}
\begin{scope}
\path[fill=gray, opacity=0.2, pattern= north east lines] (0,-8) -- (4,-8) -- (4,-8.5) -- (0,-8.5) -- cycle;
\end{scope}

\draw[->,color=black] (-0.5,0) -- (5,0);
\draw[shift={(4.5,0)},color=black] (0pt,2pt) -- (0pt,-2pt) node[below] {\footnotesize $\frac{\pi}{2}$};
\draw[->,color=black] (0,-9) -- (0,9);
\draw (5,0.2) node[anchor=north west] {$u$};
\draw (-0.3,10) node[anchor=north west] {$\tau$};
\draw[shift={(0,-8)},color=black] (2pt,0pt) -- (-2pt,0pt) node[left] {\footnotesize $ - \pi$};
\draw[shift={(0,-4)},color=black] (2pt,0pt) -- (-2pt,0pt) node[left] {\footnotesize $ - \frac{\pi}{2}$};
\draw[shift={(0,4)},color=black] (2pt,0pt) -- (-2pt,0pt) node[left] {\footnotesize $  \frac{\pi}{2}$};
\draw[shift={(0,8)},color=black] (2pt,0pt) -- (-2pt,0pt) node[left] {\footnotesize $  \pi$};

\draw[color=black] (0pt,-10pt) node[left] {\footnotesize $0$};

\draw[color=black, dotted] (0,-8) -- (4,-4) -- (0,0) -- (4,4) -- (0,8);
\draw[color=black, dotted] (0,-8) -- (4,-8);
\draw[color=black, dotted] (0,8) -- (4,8);
\draw[dashed] (4,8) -- (4,-8);
\node at (1.5,-4) {$V_{-1}$};
\node at (1.5,4) {$V_{1}$};
\node at (2.5,0.5) {$V_{0}$};
\node at (2.5,-7) {$V_{2}$};
\node at (2.5,7) {$V_{2}$};
\node at (2,-10) {(a)};

\begin{scope}
\path[fill=gray, opacity=0.2] (10,-8) -- (14,-4) -- (10,0) -- cycle;
\end{scope}
\begin{scope}
\path[fill=gray, opacity=0.2] (10,0) -- (14,4) -- (10,8) -- cycle;
\end{scope}

\begin{scope}
\path[fill=gray, opacity=0.2] (10,8) -- (14,12) -- (12,14) -- (10,14) -- cycle;
\end{scope}
\begin{scope}
\path[fill=gray, opacity=0.2] (10,-8) -- (11,-9) -- (10,-9) -- (10,-8) -- cycle;
\end{scope}

\draw[->,color=black] (9.5,0) -- (15,0);
\draw[shift={(14.5,0)},color=black] (0pt,2pt) -- (0pt,-2pt) node[below] {\footnotesize $\frac{\pi}{2}$};
\draw[->,color=black] (10,-9.5) -- (10,14.5);
\draw (15,0.2) node[anchor=north west] {$u$};
\draw (9.7,15.5) node[anchor=north west] {$\tau$};
\draw[shift={(10,-8)},color=black] (2pt,0pt) -- (-2pt,0pt) node[left] {\footnotesize $ - \pi$};
\draw[shift={(10,-4)},color=black] (2pt,0pt) -- (-2pt,0pt) node[left] {\footnotesize $ - \frac{\pi}{2}$};
\draw[shift={(10,4)},color=black] (2pt,0pt) -- (-2pt,0pt) node[left] {\footnotesize $  \frac{\pi}{2}$};
\draw[shift={(10,8)},color=black] (2pt,0pt) -- (-2pt,0pt) node[left] {\footnotesize $  \pi$};
\draw[shift={(10,12)},color=black] (2pt,0pt) -- (-2pt,0pt) node[left] {\footnotesize $  \frac{3\pi}{2}$};

\draw[color=black] (10,-10pt) node[left] {\footnotesize $0$};

\draw[color=black, dotted] (11,-9) -- (10,-8) -- (14,-4) -- (10,0) -- (14,4) -- (10,8) -- (14,12) -- (12,14);
\draw[dashed] (14,-9) -- (14,14);
\node at (11.5,-4) {$V_{-1}$};
\node at (11.5,4) {$V_{1}$};
\node at (12.5,0.5) {$V_{0}$};
\node at (12.5,-8) {$V_{-2}$};
\node at (12.5,8) {$V_{2}$};
\node at (13.25,13.5) {$V_{4}$};
\node at (11.5,12) {$V_{3}$};
\node at (12,-10) {(b)};
\end{tikzpicture}}
}
\end{figure}
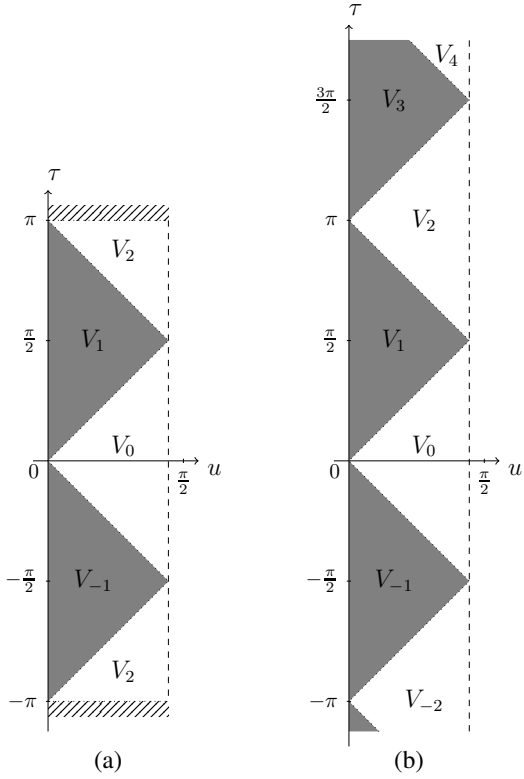

\subsection{Wick rotation}
\label{ssc:AdS_Wick}

Anti-de Sitter space is closely related to the hyperbolic space
\begin{align}
 \bH^d := \{ x \in \bR^{1,d} \;|\; [x|x]=-1 \},
\end{align}
where
\begin{align}
 [x|x'] = -x^0 {x'}^0 + \sum_{i=1}^d x^i {x'}^i,
\end{align}
as in \eqref{minko}.
Let $\Delta^\h$ be the Laplace-Beltrami operator on $\mathbb{H}^d$.
Set
\begin{align}
G^\h(-\nu^2):=\Big(-\Delta^\h-(\tfrac{d-1}{2})^2+\nu^2\Big)^{-1}.
\end{align}

For $\Re(\nu)>0$, the integral kernel of $G^\h(-\nu^2)$ can be expressed
in terms of the invariant quantity $[x|x']$ and the Gegenbauer function
${\bf Z}_{\alpha,\lambda}(w)$  as (see e.g.
\cite{DGR23a,DGR23b}, and for an equivalent expression in
terms of associated Legendre functions \cite{CDT18})
 \begin{align}
 G^\h\big(-\nu^2;x,x'\big)
     =\frac{\sqrt\pi\Gamma(\frac{d-1}{2}+\nu)
     }{\sqrt2(2\pi)^{\frac{d}{2}}2^{\nu}}
     {\bf Z}_{\frac{d}{2}-1,\nu }\big(-[x|x']\big).
     \label{eq:G_hyp}
 \end{align}

Let us try to introduce a kind of a Wick rotation from $\bH^d$ to
anti-de Sitter space by replacing $x^d$ with $\pm \ii x^d$.
We have\begin{align}
  [x|x'] &=-1-\frac{[x-x'|x-x'] }{2},\quad x,x'\in\bH^d,\\ \notag
  Z := \langle x|x'\rangle &=-1-\frac{\langle x-x'|x-x'\rangle }{2},
  \quad x,x'\in\AdS_d.
 \end{align}
 Thus, similar to the case of de Sitter space, we have to replace
 $-[x|x']$ in the argument of the Gegenbauer function
 in \eqref{eq:G_hyp} by $-\big(\langle x|x' \rangle\mp\ii0\big)
 =-\langle x|x' \rangle\pm\ii0$
 and insert a prefactor $\pm \ii$ coming from the change of the
 integral measure. In this way, we obtain
\begin{align}
\pm \ii \frac{\sqrt\pi\Gamma(\frac{d-1}{2}+\nu)
     }{\sqrt2(2\pi)^{\frac{d}{2}}2^{\nu} }
     {\bf Z}_{\frac{d}{2}-1,\nu }
         \big(-Z\pm\ii0\big)
         \label{eq:Wick_G_AdS}
\end{align}
as candidates for  (anti-)Feynman propagators on $\AdS_d$.

On the proper anti-de Sitter space $\AdS_d$ the latter expression
cannot be a Green function. In fact, due to
the identity \eqref{eq:Z+-}, the application of the Klein-Gordon
operator to \eqref{eq:Wick_G_AdS} yields a nonzero distribution supported
at $\{Z=-1\}\cup\{Z=1\}$ (the diagonal and the antipode of the diagonal).

This problem dissapears on the universal cover $\tAdS_d$ of
anti-de Sitter space. The expression \eqref{eq:Wick_G_AdS}, properly
continued to further regions, yields a Green function of the
Klein-Gordon operator, as we shall see in the next subsection.

The following four functions are bisolutions of the Klein-Gordon equation:
 \begin{align} \label{eq:tAdS_bisol_W0}
  \sim
     {\bf Z}_{\frac{d}{2}-1,\nu }
         \big(-Z\pm\ii0\sgn(\tau)\big) \word{and}
  \sim
     {\bf Z}_{\frac{d}{2}-1,-\nu }
         \big(-Z\pm\ii0\sgn(\tau)\big).
 \end{align}
 We expect that the following is true:
\begin{con} On $W_0$, the functions
\eqref{eq:tAdS_bisol_W0} form a basis of bisolutions of the
Klein-Gordon equation invariant wrt the restricted anti de Sitter
group.
\label{concon}\end{con}

\subsubsection{Resolvent of the d'Alembertian}

The essential self-adjointness of the d'Alembertian $-\Box$ on
$C_\mathrm{c}^\infty(\tAdS_d)$ is not covered by the references
 \cite{Rossmann78,vdBan84}. However, we expect that the methods of
 above references can be extended to  $\tAdS_d$, so that
 one can show  that the d'Alembertian is
 indeed essentially self-adjoint on $C_\mathrm{c}^\infty(\tAdS_d)$.

The main aim of this subsection is a computation of the integral
 kernel of the resolvent of the d'Alembertian on $\tAdS_d$.
As before, it is convenient to set
\begin{align}
 \label{eq:Gnu_AdS_tilde}
 G(-\nu^2) :=\Big(-\Box-\big(\tfrac{d-1}{2}\big)^2+\nu^2\Big)^{-1},
\end{align}
and denote the integral kernel of $G(-\nu^2)$ by
 $G(-\nu^2;x,x')$.
Below we will compute  $G(-\nu^2;x,x')$.
We state this computation  as a theorem. However, the arguments that we
present, quite  simple and convincing,
 use Conjecture \ref{concon}, which we have not proved. A natural
 strategy for a complete proof   would involve global coordinates and
 summation formulas for
Gegenbauer functions, and is much more complicated. It will not be given in this paper.

To describe $G(-\nu^2;x,x')$ explicitly, it is convenient
for $n\in\bZ$ to introduce
  open regions
\begin{align}
 W_n &:= \Big(
 V_{2n-1} \cup V_{2n} \cup V_{2n+1}\big)^{\rm cl}\Big)^{\circ},\quad
  n\in\bZ,
\end{align}
with $V_n$ as defined in Subsection \ref{ssc:AdS_geometry} and
with ${\rm cl}$ denoting the closure and $\circ$ the interior.
We have $W_{n}\cap W_{n+1}=   V_{2n+1}$ and
\begin{align}
\tAdS_d=\bigcup_{n\in\bZ}W_n.
\end{align}

\begin{thm}
 For $\nu^2\in\bC\setminus\bR$ and $\Re(\nu)>0$, the integral
 kernel of the resolvent \eqref{eq:Gnu_AdS_tilde} is given on
 $W_n$ by the formula
 \begin{align}
 \label{eq:tAds_resolvent}
 &G(-\nu^2;x,x')
 = \dfrac{\sqrt\pi\Gamma(\frac{d-1}{2}+\nu)
 }{\sqrt2(2\pi)^{\frac{d}{2}}2^{\nu} }
 \\ \notag
  &\quad\times\begin{cases}
               \ii \e^{-\ii |n|(\frac{d-1}{2}+\nu)\pi}
 {\bf Z}_{\frac{d}{2}-1,\nu }
         \Big(-(-1)^n Z
         +(-1)^n\ii0 s\Big), &\quad \Im\nu<0; \\ \notag
   -\ii\e^{\ii |n|(\frac{d-1}{2}+\nu)\pi}
 {\bf Z}_{\frac{d}{2}-1,\nu }
         \Big(-(-1)^n Z
         -(-1)^n\ii0 s\Big), &\quad    \Im\nu>0,
              \end{cases}
\end{align}
Here  $s$ can be represented by $s=\sgn(\sin(|\tau-\tau'|))$,
or
\begin{align}(x,x')\in
   V_{2n-1}&\Rightarrow s=(-1)^n\sgn(2n-1),\\ \notag
  (x,x')\in V_{2n}&\Rightarrow s=0,\\ \notag
(x,x')\in  V_{2n+1}&\Rightarrow s=(-1)^{n+1}\sgn(2n+1).
\end{align}
(Note that in $ V_{2n}$ we may set $s=-1$ or $s=1$
because the function is univalent).
\label{resol-anti}\end{thm}

\begin{proof}[Proof (assuming the validity of Conj. \ref{concon}).]
 We split the proof of \eqref{eq:tAds_resolvent} in two steps. First,
 we show that \eqref{eq:tAds_resolvent} is a fundamental solution
 with appropriate decay behavior as $|Z|\to\infty$ \emph{and}
 $|\tau|\to\infty$. Second, we argue that adding any bisolution,
 which is a (non-zero) linear combination of
 \eqref{eq:tAdS_bisol_W0} has exponential growth as $\tau\to +\infty$
 or $\tau\to-\infty$.

 On $W_0$, consider \eqref{eq:Wick_G_AdS}. On the overlap
 $V_1 = W_0 \cap W_2$, we have
 \begin{align} \label{eq:tAdS_intermediate}
      {\bf Z}_{\frac{d}{2}-1,\nu }
         \big(-Z\pm\ii0\big)
         &= \ee^{\mp \ii \pi \big(\tfrac{d-1}{2}+\nu\big)}
         {\bf Z}_{\frac{d}{2}-1,\nu }
         \big(-(-Z)\pm(-1)\ii0\big).
\end{align}
On the chart $W_2$, the integral kernel of the resolvent must be
a bisolution and it must on $V_1$ agree with
\eqref{eq:tAdS_intermediate}. Therefore, the $\ii0$ should switch
the sign from $V_{1}$ to $V_{3}$. On
$V_1$, we have $\tau\in]0,\pi[$. Hence
\begin{align}\label{eq:tAdS_intermediate2}
 \eqref{eq:tAdS_intermediate}
         = \ee^{\mp \ii \pi \big(\tfrac{d-1}{2}+\nu\big)}
         {\bf Z}_{\frac{d}{2}-1,\nu }
         \big(-(-Z)\pm(-1)\ii0 \sgn(\sin(|\tau-\tau'|))\big)
        \word{on} V_1
\end{align}
and \eqref{eq:tAdS_intermediate2} is the appropriate continuation
of \eqref{eq:tAdS_intermediate} to $W_2$.

Now notice that $\sgn(\sin(|\tau-\tau'|))=-1$ on $V_3$.
Therefore, in this region,
\begin{align}
 &\quad \ee^{\mp \ii \pi \big(\tfrac{d-1}{2}+\nu\big)}
         {\bf Z}_{\frac{d}{2}-1,\nu }
         \big(-(-Z)\pm(-1)\ii0 \sgn(\sin(|\tau-\tau'|))\big)
 \\ \notag
 &= \ee^{\mp \ii \pi \big(\tfrac{d-1}{2}+\nu\big)}
         {\bf Z}_{\frac{d}{2}-1,\nu }
         \big(-(-Z)\pm \ii0 \big)
  \\ \notag
 &= \ee^{\mp 2\pi \ii \big(\tfrac{d-1}{2}+\nu\big)}
         {\bf Z}_{\frac{d}{2}-1,\nu }
         \big(-(-1)^2 Z\pm (-1)^2 \ii0 \sgn(\sin(|\tau-\tau'|))\big).
\end{align}
 Inductively, we obtain \eqref{eq:tAds_resolvent}
for $n\geq0$. The continuation to negative $n$ works analogously.
Since only
${\bf Z}_{\frac{d}{2}-1,\nu }$ appears, both formulas have an
appropriate decay behavior as $|Z|\to\infty$ for any sign of
$\Im(\nu)$. However, the exponential prefactor
\begin{align}
 \ee^{\mp |n|\pi\ii \big(\tfrac{d-1}{2}+\nu\big)}
\end{align}
decays only for $\Im(\nu)\lessgtr0$ as  $|n|\to\infty$
(or equivalently, as $|\tau|\to\infty$).

Assuming Conjecture \ref{concon} we see that these are the only
fundamental solutions with appropriate decay behavior. Thus a basis
of bisolutions that decay as $|Z|\to\infty$ is on $W_0$ given by
\begin{align}
 &{\bf Z}_{\frac{d}{2}-1,\nu }
         \big(-Z+\ii0)\big)
         +{\bf Z}_{\frac{d}{2}-1,\nu }
         \big(-Z-\ii0\big)
         \\ \notag
    \word{and}
     &\sgn(\tau) \Big( {\bf Z}_{\frac{d}{2}-1,\nu }
         \big(-Z+\ii0)\big)
         -{\bf Z}_{\frac{d}{2}-1,\nu }
         \big(-Z-\ii0\big) \Big).
\end{align}
Both choices contain $+\ii0$ and $-\ii0$, and it is easy to see
that their continuation to the higher $W_n$ contains terms that
exponentially increase with time at least in one of the directions
$\tau>0$ resp. $\tau<0$.
\end{proof}


\subsubsection{Propagators from the resolvent}
\label{Propagators on anti-de Sitter space from the resolvent}

From the formula for the resolvent we can immediately determine the
operator-theoretic Feynman and anti-Feynman propagators for
for $n\in\mathbb{Z}$  in the regions $W_n$. We have
\begin{align}
 G_{\rm op}^{\F/ \overline\F} (x,x')
 = &\pm\ii
  \dfrac{\sqrt\pi\Gamma(\frac{d-1}{2}+\nu)
 }{\sqrt2(2\pi)^{\frac{d}{2}}2^{\nu} } \e^{\mp\ii |n|(\frac{d-1}{2}+\nu)\pi}
 {\bf Z}_{\frac{d}{2}-1,\nu }
         \Big(-(-1)^nZ
         \pm(-1)^n\ii0 s\Big),
\end{align}
where $s$ is as in Theorem \ref{resol-anti}.

The sum $G_{\rm op}^\F+G_{\rm op}^{ \overline\F}$ has a causal support
 (or
in the terminology of Def. \ref{def-special} the specialty condition holds):
\begin{align}\notag
 &G_{\rm op}^\F (x,x')+  G_{\rm op}^{ \overline\F}(x,x')
 = \ii
  \dfrac{\sqrt\pi\Gamma(\frac{d-1}{2}+\nu)
     }{\sqrt2(2\pi)^{\frac{d}{2}}2^{\nu} } \Bigg(
  \e^{-\ii |n|(\frac{d-1}{2}+\nu)\pi}
 {\bf Z}_{\frac{d}{2}-1,\nu }
         \Big(-(-1)^nZ
         +(-1)^n\ii0  s\Big)\\
&- \e^{\ii |n|(\frac{d-1}{2}+\nu)\pi}
 {\bf Z}_{\frac{d}{2}-1,\nu }
         \Big(-(-1)^nZ
         -(-1)^n\ii0 s\Big)\Bigg).\label{nonono}
\end{align}
In fact, \eqref{nonono}  vanishes for $x\in V_0$. We obtain the retarded and advanced propagator by multiplying it
with $\theta\big(\pm(\tau-\tau')\big)$.
The Pauli-Jordan propagator is then the difference of the retarded
and advanced propagator.
We use \eqref{eq:relDcurved} to define  $  G^{(\pm)} $
 obtaining on the chart $W_n$:
\begin{align}  \label{eq:pos_neg_props_AdS}
  G^{(\pm)} (x,x')
  = &
  \dfrac{\sqrt\pi\Gamma(\frac{d-1}{2}+\nu)
 }{\sqrt2(2\pi)^{\frac{d}{2}}2^{\nu} }
 \e^{\mp\ii n(\frac{d-1}{2}+\nu)\pi }
 \\ \notag &\times
 {\bf Z}_{\frac{d}{2}-1,\nu }
         \Big(-(-1)^nZ
         \pm(-1)^n\ii0 \tilde s\Big).
\end{align}
Here  $\tilde s$ can be represented by $\tilde s=\sgn(\sin(\tau-\tau'))$,
or
\begin{align}(x,x')\in
  V_{2n-1}&\Rightarrow \tilde s=(-1)^n,\\
  (x,x')\in  V_{2n}&\Rightarrow \tilde s=0,\\
(x,x')\in   V_{2n+1}&\Rightarrow \tilde s=(-1)^{n+1}.
\end{align}

Note that for $\nu^2<0$ the specialty condition is no longer
true. Therefore, although we can define $G_\op^\F$ and $G_\op^{ \overline\F}$,
we are not able to obtain other propagators from them.

\subsubsection{Trigonometric Pöschl-Teller
Hamiltonian}

In our further analysis of the anti-de Sitter space we will need properties of
the following 1-dimensional Schr\"odinger operator on $L^2[0,\tfrac{\pi}{2}]$:
\begin{align}\label{PT}
 H_{\alpha,\nu}^{\rm PT} := -\partial_u^2
 +\frac{\alpha^2-\tfrac14}{\sin(u)^2}
 + \frac{\nu^2-\tfrac14}{\cos(u)^2}.
\end{align}
It is called  the \emph{trigonometric Pöschl-Teller
Hamiltonian} \cite{PT33}
and
is one of the 1-dimensional Schr\"odinger operators exactly solvable in
terms of hypergeometric functions.

By an extension of standard
arguments (cf. \cite[Chapter X]{RSII}), one finds that
$H_{\alpha,\nu}^{\rm PT}$, viewed as an operator on $L^2[0,\tfrac{\pi}{2}]$,
is essentially self-adjoint if both $\nu^2\geq1$ and $\alpha^2\geq1$,
it has a positive Friedrichs extension if $\nu^2\geq0$ and $\alpha^2\geq0$,
and  all self-adjoint extensions are unbounded from below if
$\nu^2<0$ or $\alpha^2<0$.

\subsubsection{Propagators from the evolution of Cauchy data}
\label{sss:AdS_bc}

In this subsection we present an approach to propagators on
$\tAdS_d$ different from that
of Subsection \ref{Propagators on anti-de Sitter space from the resolvent}.
It is  based on the  evolution of the Cauchy data. We will use the stationarity of $\tAdS_d$.

The Klein-Gordon operator with effective mass $m$ in the coordinates
\eqref{orthis}
is given by
\begin{align}
\label{eq:AdS_KG}
 -\square_g + m^2
 &= -\frac{1}{\sqrt{|\det g|}} \del_\mu g^{\mu\nu} \sqrt{|\det g|}
 \del_\nu + m^2
 \\ \notag
 &= \cos(u)^2 \Big( \partial_\tau^2 -\frac{\Delta_{\bS^{d-2}}}{\sin(u)^2}
 -\tan(u)^{2-d} \partial_u \tan(u)^{d-2} \partial_u
 + \frac{m^2}{\cos(u)^2}\Big)
\end{align}
with $\Delta_{\bS^{d-2}}$ being the Laplace-Beltrami operator on the
$d-2$-dimensional sphere parametrized by the coordinates $\Omega$.
Gauging \eqref{eq:AdS_KG}  we obtain
\begin{align}
\label{eq:Pöschl1}
 &\quad\tan(u)^{\tfrac{d-2}{2}} \big( -\square_g + m^2  \big)
 \tan(u)^{\tfrac{2-d}{2}}
 \\
 &= \cos(u)^2 \Bigg( \partial_\tau^2 -\partial_u^2
 +\frac{-\Delta_{\bS^{d-2}} + \big(\tfrac{d-3}{2}\big)^2-\tfrac14}{\sin(u)^2}
 + \frac{\nu^2-\tfrac14}{\cos(u)^2}\Bigg) \label{eq:Pöschl2}
\end{align}
with $\nu^2$ as in \eqref{eq:AdS_spectral_parameter}.
For $d\geq3$, the spectrum of $-\Delta_{\bS^{d-2}}$ is
$\{ l(l+d-3)\;|\; l\in\bN_0\}$. For $d=2$, the term proportional to
$\sin(u)^{-2}$ vanishes.

Hence, restricted to eigenfunctions of
$-\Delta_{\bS^{d-2}}$,
\eqref{eq:Pöschl1} becomes,
up to the prefactor $\cos(u)^2$, the trigonometric Pöschl-Teller
Hamiltonian \eqref{PT}
with $\alpha := l+\tfrac{d-3}{2}$ if $d\geq3$ and $\alpha^2=\tfrac14$ if
$d=2$.

To define dynamics in $\tAdS^d$, one needs to fix a
self-adjoint extension of $H_{\alpha,\nu}^{\rm PT}$, i.e., boundary conditions
at spacelike infinity.
A comprehensive analysis of boundary conditions for $H_{\rm PT}$
and their application to anti-de Sitter QFT has been carried out
by Ishibashi and Wald \cite{IW04}.

Notice first that $\alpha^2<1$ if and only if $d=2$ or $d\in\{3,4\}$
and $l=0$. Hence, one might expect that boundary conditions at the
origin need to be fixed in these cases. But one can show that this
is merely an artifact of the choice of coordinates and that no
boundary conditions at $u=0$ are required \cite{IW04}. The important
part is fixing the boundary conditions (i.e., a self-adjoint
extension of $H_{\alpha,\nu}^{\rm PT}$) at spatial infinity $u=\tfrac{\pi}{2}$.

Now for $\nu^2\geq 1$ the operator
$H_{\alpha,\nu}^{\rm PT}$ is essentially self-adjoint, so the
dynamics is uniquely determined. We can compute all propagators---they
agree with those obtained from the operator-theoretic Feynman
propagator. In particular, the specialty condition is true.

For $0\leq\nu^2<1$ we have a one-parameter family of self-adjoint
extensions, depending on the boundary condition at spatial infinity. All of them can be
used to define the propagators. Among them there is a distinguished
boundary condition given by the Friedrichs extension, or equivalently,
by the analytic continuation in the parameter $\nu$.
By the uniqueness of analytic continuation, this leads to propagators
that agree
with those obtained from the operator-theoretic Feynman propagator.

Finally, for $\nu^2<0$ there is a one-parameter family of realizations of
$H_{\alpha,\nu}^{\rm PT}$, and all  are unbounded from below. Each of
them can be used to define an evolution of Cauchy data, and hence the
retarded and advanced propagator.  However, in contrast
to the case $0\leq\nu^2<1$, none of them is distinguished.

\appendix

\section{Projections and Krein spaces}
\label{app:Involutions_and_projections}
The main goal of this appendix is a short presentation
  of basic facts about Krein spaces, which provide a natural
  functional-analytic setting for the Klein-Gordon equation.
There exist comprehensive textbook treatments of spaces with
indefinite inner products \cite{AzizovIokhvidov,Bognar}. Our treatment
is perhaps more concise, concentrating on the concepts directly needed
in our paper. To a large extent we follow \cite{DS22}, with some
simplifications and improvements.

We start with
some useful but not well-known lemmas about projections, involutions and
complementary subspaces, presenting  constructions related
to pairs of complementary subspaces, which go back to Kato
\cite{kato:perturbation}. Then we  describe elements of the theory of Krein spaces.
The main result that we prove
is the proposition saying that every pair consisting of a  maximal uniformly positive and
maximal uniformly negative subspace is complementary, which is
crucial in the construction of the in-out Feynman propagator.

\subsection{Involutions}
\label{sub:Involutions}
Let $\cW$ be a vector space. We do not need topology on $\cW$ for the
moment. We use the term ``invertible'' as a synonym of ``bijective''.
\begin{definition}
  \label{def:complementary}
  We say that a pair $(\cZ_\bullet^{(+)},\cZ_\bullet^{(-)})$ of subspaces of  $\cW$ is \emph{complementary} if
  \begin{equation*}
    \cZ_\bullet^{(+)} \cap \cZ_\bullet^{(-)} = \{0\},
    \quad
    \cZ_\bullet^{(+)}+\cZ_\bullet^{(-)} = \cW.
  \end{equation*}
\end{definition}

  \begin{definition}
  \label{def:complementary_proj}
  We say that a pair of operators
  $(\Pi_\bullet^{(+)},\Pi_\bullet^{(-)})$ on $\cW$ is a pair
  of complementary projections if
  \begin{equation*}
    {(\Pi_\bullet^{(\pm)})}^2=\Pi_\bullet^{(\pm)},
    \quad
    \Pi_\bullet^{(+)}+    \Pi_\bullet^{(-)}=\one.
  \end{equation*}
\end{definition}

\begin{definition}
  An operator $S_\bullet$ on $\cW$ is called an \emph{involution}, if $S_\bullet^2 = \one$.
\end{definition}
Note that there is a 1-1 correspondence between involutions, pairs of
complementary projections and pairs of complementary subspaces:
\begin{equation}
  \label{projo-}
  \Pi^{(\pm)}_\bullet := \frac12 (\one \pm S_\bullet),
  \quad
  \cZ^{(\pm)}_\bullet := \Ran(\Pi^{(\pm)}_\bullet).
\end{equation}



\subsection{Pair of involutions I}
In this subsection we give a criterion for complementarity of two
subspaces, and then we construct the corresponding projections
following Kato \cite{kato:perturbation}.

Suppose that $S_1$ and $S_2$ are two involutions on $\cW$.
Let
\begin{align*}
  \Pi_i^{(\pm)} := \frac12 (\one \pm S_i),\quad \cZ_i^{(\pm)} := \Ran(\Pi_i^{(\pm)}),\quad i=1,2,
\end{align*}
be the corresponding pairs of complementary projections and
subspaces. Define
  \begin{align}
    \Upsilon
    \label{eq:Upsilon_operator4}
     & \mathrel{\phantom{:=}\mathllap{=}}
       \frac14(S_1+S_2)^2.
  \end{align}
Observe that $\Upsilon$ commutes with $\Pi_1^{(+)}$, $\Pi_1^{(-)}$, $\Pi_2^{(+)}$ and $\Pi_2^{(-)}$.

\begin{prop}
  \label{prop:pair_proj}
  The following conditions are equivalent:
  \begin{enumerate}
    \item[(i)] $\Upsilon$ is invertible.
    \item[(ii)] $\Pi_1^{(+)} + \Pi_2^{(-)}$ and $\Pi_2^{(+)} + \Pi_1^{(-)}$ are  invertible.
    \end{enumerate}
  Moreover, if one of the above holds, then the pairs $(\cZ_1^{(+)},\cZ_2^{(-)})$ as well as $(\cZ_2^{(+)},\cZ_1^{(-)})$ are complementary.
\end{prop}
\begin{proof}
 The equivalence of (i) and (ii) follows from
  \begin{align}
    \Upsilon
    & = (\Pi_1^{(+)} + \Pi_2^{(-)}) (\Pi_2^{(+)} + \Pi_1^{(-)})
  \end{align}
 by the following  easy fact:
  If $R,S,T$ are maps such that $R = ST = TS$, then
  $R$ is bijective if and only if both $T$ and $S$ are bijective.

  The last implication follows from the next proposition.
\end{proof}

In the setting of the above proposition we can use $\Upsilon$ to construct two pairs of complementary projections:
\begin{prop}
  \label{prop:compl-proj}
  Suppose that $\Upsilon$ is invertible.
  Then
  \begin{alignat*}{2} \Lambda_{12}^{(+)} & := \Pi_1^{(+)} \Upsilon^{-1} \Pi_2^{(+)} & \quad & \text{is the projection onto $\cZ_1^{(+)}$ along $\cZ_2^{(-)}$}, \\ \Lambda_{12}^{(-)} & := \Pi_2^{(-)} \Upsilon^{-1} \Pi_1^{(-)} & \quad & \text{is the projection onto $\cZ_2^{(-)}$ along $\cZ_1^{(+)}$}, \\ \Lambda_{21}^{(+)} & := \Pi_2^{(+)} \Upsilon^{-1} \Pi_1^{(+)} & \quad & \text{is the projection onto $\cZ_2^{(+)}$ along $\cZ_1^{(-)}$}, \\ \Lambda_{21}^{(-)} & := \Pi_1^{(-)} \Upsilon^{-1} \Pi_2^{(-)} & \quad & \text{is the projection onto $\cZ_1^{(-)}$ along $\cZ_2^{(+)}$}.
  \end{alignat*}

  In particular,
  \begin{equation*}
    \Lambda_{12}^{(+)} + \Lambda_{12}^{(-)} = \one,
    \quad
    \Lambda_{21}^{(+)} + \Lambda_{21}^{(-)} = \one.
  \end{equation*}
\end{prop}
\begin{proof}
  First we check that $\Lambda_{12}^{(+)}$ is a projection:
  \begin{align*}
    \bigl( \Lambda_{12}^{(+)} \bigr)^2 & = \Pi_1^{(+)} \Upsilon^{-1} \Pi_2^{(+)} \Pi_1^{(+)} \Upsilon^{-1} \Pi_2^{(+)} \\ & = \Pi_1^{(+)} \Upsilon^{-1} (\Pi_2^{(+)} \Pi_1^{(+)} + \Pi_1^{(-)} \Pi_2^{(-)}) \Upsilon^{-1} \Pi_2^{(+)} = \Lambda_{12}^{(+)}.
  \end{align*}
  Moreover,
  \begin{align*}
    \Lambda_{12}^{(+)} & = \Pi_1^{(+)} (\Pi_2^{(+)}+\Pi_1^{(-)}) \Upsilon^{-1}=\Upsilon^{-1} (\Pi_1^{(+)}+\Pi_2^{(-)}) \Pi_2^{(+)}.
  \end{align*}
  But $ (\Pi_2^{(+)}+\Pi_1^{(-)}) \Upsilon^{-1}$ and $\Upsilon^{-1} (\Pi_1^{(+)}+\Pi_2^{(-)})$ are invertible.
  Hence $\Ran(\Lambda_{12}^{(+)})=\Ran (\Pi_1^{(+)})$ and $\Ker(\Lambda_{12}^{(+)}) =\Ker(\Pi_2^{(+)})=\Ran(\Pi_2^{(-)})$.
  This proves the statement of the proposition about $\Lambda_{12}^{(+)}$.
  The remaining statements are proven analogously.
\end{proof}

\begin{remark} Note that the notation for projections
  $\Lambda_{12}^{(\pm)}$ and   $\Lambda_{21}^{(\pm)}$ is different
  than in \cite{DS22}. \end{remark}
\subsection{Pair of  involutions II}
\label{sub:Pairs of admissible involutions}

Let $S_i$, $(\Pi_i^{(+)},\Pi_i^{(-)})$, $(\cZ_i^{(+)},\cZ_i^{(-)})$,
$i=1,2$, be
as in the previous subsection.
Set
\begin{equation}
  K := S_2 S_1.\label{kaku}
\end{equation}

\begin{prop}
$K$ is invertible and
  \begin{align}
    \label{chazar0}
    S_1KS_1=S_2 K S_2                                          & =
                                                                 K^{-1}.\end{align}
\end{prop}

In what follows we will use
the decomposition $\cW = \cZ_1^{(+)} \oplus \cZ_1^{(-)}$.
Under the assumption that                                                            $\one+K$
                                                               is
                                                               invertible,
                                                               we define
\begin{equation}
 c := \Pi_1^{(+)} \frac{\one-K}{\one+K} \Pi_1^{(-)}, \quad d := \Pi_1^{(-)} \frac{\one-K}{\one+K} \Pi_1^{(+)}.\label{kaku1}
\end{equation}
where $c$, resp. $d$  are interpreted as  operators from $\cZ_1^{(-)}$
to $\cZ_1^{(+)}$, resp. from $\cZ_1^{(+)}$ to $\cZ_1^{(-)}$.

\begin{prop}
\label{prop:1+K_etc}
The following conditions are equivalent:
  \begin{itemize}
\item[(i)] $\Upsilon$ is invertible (or  Condition (ii) of
      Proposition \ref{prop:pair_proj} is true).
\item[(ii)] $\one+K$ is invertible.
  \end{itemize}
  Suppose that the above conditions are true. As we know from
  Prop. \ref{prop:pair_proj}, the pairs of subspaces
  $(\cZ_1^{(+)}, \cZ_2^{(-)})$ and $(\cZ_2^{(+)}, \cZ_1^{(-)})$
  are then complementary. Here are new formulas for the
  corresponding projections:
  \begin{alignat*}{3}
    \Lambda_{12}^{(+)} & =
    \begin{bmatrix}
      \one & c \\ 0 & 0
    \end{bmatrix}
     \quad & \text{ projects onto $\cZ_1^{(+)}$ along $\cZ_2^{(-)}$},
    \\
    \Lambda_{12}^{(-)}     & =
    \begin{bmatrix}
      0 & -c   \\
      0 & \one
    \end{bmatrix}
    \quad & \text{ projects onto $\cZ_2^{(-)}$ along $\cZ_1^{(+)}$},
    \\
    \Lambda_{21}^{(+)}     & =
    \begin{bmatrix}
      \one & 0 \\
      -d  & 0
    \end{bmatrix}
    \quad & \text{ projects onto $\cZ_2^{(+)}$ along $\cZ_1^{(-)}$},
    \\
    \Lambda_{21}^{(-)}     & =
    \begin{bmatrix}
      0   & 0    \\
      d  & \one
    \end{bmatrix}
    \quad & \text{ projects onto $\cZ_1^{(-)}$ along $\cZ_2^{(+)}$}.
  \end{alignat*}
Besides,
                                                             $\one-dc$
                                                             and
                                                             $\one-cd$
                                                             are
                                                             invertible,
and we have the following formulas:
  \begin{subequations}
    \begin{align}\label{chazar.}
    \Upsilon &  =\frac1{4}(\one+K)(\one+K^{-1})=
    \begin{bmatrix}
     ( \one-cd)^{-1}  & 0         \\
      0         & (\one-d c)^{-1}
    \end{bmatrix}
 ,\\
      K                     & = \mathrlap{
        \begin{bmatrix}
          (\one+cd )(\one-cd )^{-1} & -2c(\one-d c)^{-1}         \\
          -2d (\one-cd )^{-1}       & (\one+d c)(\one-d c)^{-1}
        \end{bmatrix}
        ,}
      \label{chazar2}     \\
      \Pi_1^{(+)}           & =  \begin{bmatrix}        \one & 0 \\        0    & 0      \end{bmatrix}     ,
                             \quad      \Pi_2^{(+)}            =      \begin{bmatrix}
      (\one-cd )^{-1}     & c(\one-d c)^{-1}     \\        -d (\one-cd )^{-1} & -d c(\one-d c)^{-1}      \end{bmatrix},
      \label{chazar22}      \\      \Pi_1^{(-)}           &
      =      \begin{bmatrix}        0 & 0    \\        0 &
      \one      \end{bmatrix}      ,
      \quad      \Pi_2^{(-)}            =      \begin{bmatrix}
      -cd (\one-cd )^{-1} & -c(\one-d c)^{-1} \\        d (\one-cd
      )^{-1}   & (\one-d c)^{-1}      \end{bmatrix}
      ,      \label{chazar23}      \\      S_1                   &
      =      \begin{bmatrix}        \one & 0     \\        0    &
      -\one      \end{bmatrix}      ,
      \quad      S_2                  =       \begin{bmatrix}
      (\one+cd )(\one-cd )^{-1} & 2c(\one-d c)^{-1}           \\
      -2d (\one-cd )^{-1}       & -(\one+d c)(\one-d
      c)^{-1}      \end{bmatrix}      .      \label{chazar25}
    \end{align}
  \end{subequations}
\end{prop}

\begin{proof}
  We have
  \beq
  \Upsilon=\frac14(S_1+S_2)^2=\frac14(\one+K)(\one+K^{-1}).\eeq
  But $(\one+K^{-1})=K^{-1}(\one+K)$. Hence $\one+K$ is invertible iff
  $\one+K^{-1}$ is. Therefore, (i)$\Leftrightarrow$(ii).

  For the remainder of the proof we assume that $\one+K$ is invertible.
We have                                                               \begin{align} \label{chazar}
                                                                        S_1\frac{\one-K}{\one+K} S_1
                                                                        & = - \frac{\one-K}{\one+K}.
                                                                      \end{align} Therefore
                                                                      \beq\Pi_1^{(+)}\frac{\one-K}{\one+K}\Pi_1^{(+)}=\Pi_1^{(-)}\frac{\one-K}{\one+K}\Pi_1^{(-)}=0.\eeq
                                          Hence,
                                          \beq                            \label{chazar1}
      \frac{\one-K}{\one+K}  =
      \begin{bmatrix}
        0   & c \\
        d  & 0
      \end{bmatrix}.\eeq
      This implies
      \beq\frac{1}{\one+K}=\frac12\begin{bmatrix}\one&c\\d&\one\end{bmatrix},\quad
      \frac{1}{\one+K^{-1}}=\frac12\begin{bmatrix}\one&-c\\-d&\one\end{bmatrix}.
\label{chazara}\eeq
Multiplying the two expressions of \eqref{chazara}  yields
\beq\Upsilon^{-1}=\begin{bmatrix}\one-cd&0\\0&\one-dc\end{bmatrix}.\eeq
Hence we proved both identities of \eqref{chazar.}, as well as
invertibility of $\one-cd$ and $\one-dc$.

We check that
  \begin{align}
    \begin{bmatrix}
      \one & c \\ d  & \one
    \end{bmatrix}
    ^{-1}                          =
       \begin{bmatrix}
      (\one-cd )^{-1} & -c(\one-d c)^{-1} \\ -d (\one-cd )^{-1} & (\one-d c)^{-1}
    \end{bmatrix}
       .
  \end{align}
  Now
  \beq K=2     \begin{bmatrix}
      \one & c \\ d  & \one
    \end{bmatrix}
    ^{-1}
    -\begin{bmatrix}\one&0\\0&\one\end{bmatrix}\eeq
    yields \eqref{chazar2}.

    The formulas for $\Pi_1^{(\pm)}$ and $S_1$ are obvious. We obtain
    $S_2$ from $S_2=KS_1$. From $S_2$ we get $\Pi_2^{(\pm)}$.

Now $\Lambda_{12}^{(+)}=\Pi_1^{(+)}\Upsilon^{-1}\Pi_2^{(+)}$ yields
\eqref{chazar22}, etc.
\end{proof}

The operators $c$, $d$ are sometimes called {\em angular operators}.

\subsection{Pair of  self-adjoint involutions in a Hilbert space}

Suppose now that $\cW$ is a Hilbert space and $S_i$, $i=1,2$, is a
pair of self-adjoint involutions. Obviously, the corresponding
projections $\Pi_i^{(+)}$, $\Pi_i^{(-)}$ are then orthogonal.

We will use the orthogonal decomposition $\cW=\cZ_1^{(+)}\oplus
\cZ_1^{(-)}$. In this decomposition we can write
\begin{align}
  \Pi_2^{(+)}=\begin{bmatrix}A&B\\B^*&C\end{bmatrix},
                                     &\qquad\text{where}\quad 0\leq
                                       A\leq\one,\quad 0\leq
                                       C\leq\one.\label{rott}\end{align}
                                     Using
                                     $(\Pi_2^{(+)})^2=\Pi_2^{(+)}$ we obtain
\begin{align}  \big(A-\tfrac12\big)^2=\tfrac14-BB^*,&\qquad   \big(C-\tfrac12\big)^2=\tfrac14-B^*B.\label{roott}
\end{align}

For an operator $K$, $\sigma(K)$ will denote its spectrum. If $K$ is  self-adjoint  we will  write
\beq\inf K=\inf\sigma(K),\quad\sup K=\sup\sigma(K).\eeq
It follows from \eqref{roott} that $\frac14\geq\sup BB^*=\sup B^*B=\|B\|^2$, and hence
\beq0\leq\inf\big(\tfrac14-BB^*\big)=\inf\big(\tfrac14-B^*B\big).\eeq

The following proposition describes the situation where the
angle between the projections $\Pi_1^{(+)}$ and $\Pi_2^{(+)}$ is not more
than $\frac\pi4$:
\begin{prop}
The following conditions are equivalent:
\begin{enumerate}\item $A\geq\frac12$ and $C\leq\frac12$.
  \item $A=\frac12+\sqrt{\frac14-BB^*}$ and $C=\frac12-\sqrt{\frac14-B^*B}$.
  \end{enumerate}\label{qroott}
  \end{prop}

  \proof 1.$\Leftarrow $ 2. is obvious.

  1.$\Rightarrow $ 2. follows from
\eqref{roott}, where by 1.  we need to take the positive
square root. \qed

The following consequence of Prop.
\ref{qroott} will be useful in the theory of Krein spaces:
\begin{lem}
Let $P$ be an orthogonal
projection and $S$  a self-adjoint involution. Let $1\geq\alpha>0$ and
\begin{align}\label{qe1}
  PSP&\geq\alpha P,\\\label{qe2}
  (\one-P)S(\one-P)&\leq0.\end{align}
Then
\beq  (\one-P)S(\one-P)\leq-\alpha(\one-P).\label{qe3}\eeq
and $T:=S(1-P)+PS$ is invertible with \beq\|T^{-1}\|\leq\frac{1}
  {1-\sqrt{1-\alpha^2}}.\label{bound}\eeq
\label{lemma}\end{lem}

\begin{proof} We set $S_1:=2P-\one$, so that $P=\Pi_1^{(+)}$, and
$S_2:=S$. Thus we are in the setting of this subsection. We
write $\Pi_2^{(+)}=\frac{S_2+\one}{2}$  as in \eqref{rott}, so that
\begin{align}
S=\begin{bmatrix}2A-\one& 2 B \\ 2B^*&2C-\one\end{bmatrix}.
\end{align}
Thus,
\begin{align}\label{qe1a}
  PSP&\geq0\ \Leftrightarrow\ A\geq\frac12,\\
  (\one-P)S(\one-P)&\leq0 \Leftrightarrow\ C\leq\frac12.\label{qe2a}\end{align}
Hence  \eqref{qe1} and \eqref{qe2} imply
the conditions of Proposition \ref{qroott}, which allows us to rewrite
\eqref{qe1a} as
\begin{align}
S=\begin{bmatrix}\sqrt{\one-4BB^*}& 2B\\ 2B^*&-\sqrt{\one-4B^*B}\end{bmatrix}.
\end{align}
By \eqref{qe1}, $\sqrt{\one-4BB^*}\geq\alpha$. So
$1-\alpha^2\geq 4BB^*$. This implies $1-\alpha^2\geq 4B^*B$, and hence
$-\sqrt{\one-4B^*B}\leq-\alpha$, which proves \eqref{qe3}.

 Now
\begin{align}
TT^*=\one+3P-SPS-SPSP-PSPS.
\end{align}
Written as a $2\times2$ matrix, it is
\begin{align}
  TT^*&=\begin{bmatrix}\one+12BB^*&-4\sqrt{\one-4BB^*}B\\
    -4B^*\sqrt{\one-4BB^*}&\one-4B^*B
  \end{bmatrix} \notag \\
  =&   \begin{bmatrix} (\one-2\sqrt{BB^*})^2&0\\0&(\one-2\sqrt{B^*B})^2
  \end{bmatrix}+WW^*,
\end{align}
where
        \begin{align}
          W:=&\;\begin{bmatrix}2B(B^*B)^{-\frac14}\sqrt{\one+2\sqrt{B^*B}}\\
            - 2(B^*B)^{\frac14}\sqrt{\one-2\sqrt{B^*B}}\end{bmatrix}, \notag\\
          \text{so that }\quad W^* =&\;
  \begin{bmatrix}2\sqrt{\one+2\sqrt{B^*B}}(B^*B)^{-\frac14}B^*&
       - 2\sqrt{\one-2\sqrt{B^*B}}(B^*B)^{\frac14}\end{bmatrix}.
\end{align}
Now $WW^*\geq0$ and
\beq
\inf (\one-2\sqrt{BB^*})^2 =\inf (\one-2\sqrt{B^*B})^2\geq
(1-\sqrt{1-\alpha^2})^2.\eeq
Thus $TT^*\geq (1-\sqrt{1-\alpha^2})^2$.
An analogous argument (where $P$ is replaced with $\one-P$)
shows $T^*T\geq (1-\sqrt{1-\alpha^2})^2$. This proves
\eqref{bound}.
\end{proof}

\subsection{Hilbertizable spaces}
\label{sub:Hilbertizable spaces}

\begin{definition}
  \label{def:Hilbertizable}
  Let $\cW$ be a complex\footnote{Analogous definitions and results are
  valid for \emph{real} Hilbertizable spaces.} topological vector space.
  We say that it is \emph{Hilbertizable} if it has the topology of a
  Hilbert space for some scalar product $\cinner{\cdot}{\cdot}_\bullet$
  on $\cW$. We will then say that $\cinner{\cdot}{\cdot}_\bullet$ is
  \emph{compatible with (the Hilbertizable structure of)} $\cW$.
  The Hilbert space $\bigl( \cW, \cinner{\cdot}{\cdot}_\bullet \bigr)$
  will be occasionally denoted $\cW_\bullet$. We denote the
  corresponding norm by $\norm{\,\cdot\,}_\bullet$, the orthogonal
  complement of $\cZ\subset\cW$ by $\cZ^{\perp\bullet}$ and the Hermitian
  adjoint of an operator $A$ by $A^{*\bullet}$.
\end{definition}

In what follows $\cW$ is a Hilbertizable space.
Let $\cinner{\cdot}{\cdot}_1$, $\cinner{\cdot}{\cdot}_2$ be two scalar products compatible with $\cW$.
Then there exist constants $0 < c \leq C$ such that
\begin{equation*}
  c \cinner{w}{w}_1 \leq \cinner{w}{w}_2 \leq C \cinner{w}{w}_1.
\end{equation*}
Let $R$ be a linear operator on $\cW$.
We say that it is \emph{bounded} if for some (hence for all) compatible scalar products $\cinner{\cdot}{\cdot}_\bullet$ there exists a constant $C_\bullet$ such that
\begin{equation*}
  \norm{Rw}_\bullet \leq C_\bullet \norm{w}_\bullet.
\end{equation*}
Let $Q$ be a sesquilinear form on $\cW$.
We say that it is \emph{bounded} if for some (hence for all) compatible scalar products $\cinner{\cdot}{\cdot}_\bullet$ there exists $C_\bullet$ such that
\begin{equation*}
  \abs{\cinner{v}{Qw}} \leq C_\bullet \norm{v}_\bullet \norm{w}_\bullet, \quad v,w \in \cW.
\end{equation*}


\subsection{Pseudounitary spaces}

Let $(\cW,Q)$ be a Hilbertizable space equipped with a bounded
Hermitian form,
\beq(v|Qw)= \overline{(w|Qv)},\quad v,w\in\cW.\eeq

\begin{definition}
  Let $\cZ\subset \cW$.
  We define its \emph{$Q$-orthogonal companion} as follows: \[\cZ^{\perp Q} := \{w\in\cW\mid\cinner{w}{Qv}=0,\; v\in\cZ\}.
  \]
\end{definition}
Clearly, $\cZ^{\perp Q}$   is a closed subspace of $\cW$.

  \begin{definition}
    Let $w\in\cW$. We say that $w$ is \emph{positive, negative, resp. neutral} if
    \beq(w|Qw)\geq0,\quad (w|Qw)\leq0,\quad\text{resp.}
    \quad (w|Qw)=0.\eeq
        We say that a subspace $\cZ\subset\cW$ is
   \emph{positive, negative, resp. neutral} if all its elements are positive, negative, resp. neutral elements.
\end{definition}
  \begin{definition}
We say that $(\cW,Q)$ is a {\em pseudounitary space} if $\cW^{\perp Q}=\{0\}$.\end{definition}

\subsection{Krein spaces}
\label{sub:Admissible involutions and Krein spaces}

Let $(\cW,Q)$ be a Hilbertizable space equipped with a bounded Hermitian form.
\begin{definition}
  \label{admissi}
  A bounded involution $S_\bullet$ on $\cW$ will be called
  \emph{admissible} if it preserves $Q$, that is,
  \beq
    \cinner{S_\bullet v}{Q S_\bullet w}    = \cinner{ v}{Q w}, \eeq
  and
  \begin{equation}
    \label{eq:bull_space}
    {\cinner{v}{w}}_\bullet := \cinner{v}{Q S_\bullet w} = \cinner{S_\bullet v}{Q w}
  \end{equation}
  is a scalar product compatible with the Hilbertizable structure of $\cW$.
\end{definition}

\begin{definition}
  A  space $(\cW,Q)$ is called a \emph{Krein space} if it possesses an admissible involution.
\end{definition}

 Clearly, a Krein space is a pseudounitary space.

\begin{remark} In the literature sometimes instead of the term
  ``admissible involution'' one finds ``fundamental symmetry''.
\end{remark}

For any admissible involution $S_\bullet$, we define the corresponding \emph{particle projection} $\Pi_\bullet^{(+)}$ and \emph{particle space} $\cZ_\bullet^{(+)}$, as well as the \emph{antiparticle projection} $\Pi_\bullet^{(-)}$ and \emph{antiparticle space} $\cZ_\bullet^{(-)}$, as in~\eqref{projo-}.
The decomposition $\cW\simeq\cZ_\bullet^{(+)}\oplus\cZ_\bullet^{(-)}$
is often called a {
\em fundamental decomposition}. Note the following relations:
\begin{align*}
  {\cinner{v}{w}}_\bullet & = {\cinner{\Pi_\bullet^{(+)} v}{\Pi_\bullet^{(+)} w}}_\bullet + {\cinner{\Pi_\bullet^{(-)} v}{\Pi_\bullet^{(-)} w}}_\bullet, \\
  \cinner{v}{Q w} & = {\cinner{\Pi_\bullet^{(+)} v}{\Pi_\bullet^{(+)} w}}_\bullet - {\cinner{\Pi_\bullet^{(-)} v}{\Pi_\bullet^{(-)} w}}_\bullet.
\end{align*}

\begin{definition}
  Let $A$ be a bounded operator on $\cW$. Then there exists a unique
  operator $A^{*Q}$ called the {\em $Q$-adjoint of } $A$ such that
  \begin{align}
   \cinner{A^{*Q}v}{Qw}=\cinner{v}{QAw},\quad v,w\in\cW.
  \end{align}
\end{definition}

Let $\cZ\subset \cW$ and let $A$ be an operator on $\cW$.
We have the identities:
\begin{align}
  \cZ^{\perp Q}&=S_\bullet \cZ^{\perp\bullet},\\
  A^{* Q}&=S_\bullet A^{*\bullet}S_\bullet.
\end{align}
With
the help of these identities
it is easy to show various properties of $\perp Q$ and $*\bullet$:

\begin{prop} \label{prop:Qorth_props}
\begin{enumerate}
 \item  If $\cZ$ is a closed subspace, then
          $(\cZ^{\perp Q})^{\perp Q}=\cZ$.
\item If $\cZ_1,\cZ_2$ are complementary subspaces
            in $\cW$, then so are
          $\cZ_1^{\perp Q},\cZ_2^{\perp Q}$.
        \item
  Suppose that $
            (\Pi^{(+)},\Pi^{(-)})$ is a pair of complementary projections.
          Then $(\Pi^{(+)*Q}, \Pi^{(-)*Q})$ is also a pair of complementary projections and
          \begin{align}
            \cR(\Pi^{(\pm)*Q})= \cN(\Pi^{(\mp)*Q})= \cR(\Pi^{(\mp)})^{\perp Q}= \cN(\Pi^{(\pm)})^{\perp Q}.
          \end{align}
        \end{enumerate}\end{prop}

\begin{definition}\label{pseudo-uni}
  Let $R$ be a bounded invertible operator on $(\cW,Q)$.
We say that $R$ is \emph{pseudo-unitary} if
  \begin{equation}\label{prepre}
    \cinner{R v}{Q R w} = \cinner{v}{Q w}.
  \end{equation}
\end{definition}

\subsection{Krein spaces with conjugation}

\begin{definition}
  \label{def:conjugation.}
  An antilinear involution $v \mapsto \varepsilon v$ on a Krein  space
  $(\cW,Q)$
  will be called a {\em conjugation} if it antipreserves $Q$, that is
  \beq(v|Qw)=- \overline{(\varepsilon v|Q\varepsilon w)}\eeq
  and there exists an admissible involution $S_\bullet$ such that
  $\varepsilon S_\bullet \varepsilon=-S_\bullet$.
\end{definition}

 Note that then
  \begin{equation*}
    \cinner{\varepsilon v}{\varepsilon w}_\bullet =
    \overline{\cinner{v}{w}}_\bullet.   \end{equation*}

\begin{definition}
We say that an operator $R$   is  \emph{real} if
  $\overline{R} := \varepsilon R \varepsilon =R$. We say that
  $R$ is \emph{anti-real} if $\overline{R}=-R$, that is, if $\ii R$ is real.
\end{definition}

Krein spaces with conjugations are especially important:
  Suppose that $(\cW,Q)$ is a Krein space with  conjugation.
  Clearly, if $S_\bullet$ is an admissible anti-real involution, then
  \begin{equation*}
     \overline{\Pi_\bullet^{(+)}} = \Pi_\bullet^{(-)}, \quad  \overline{\cZ_\bullet^{(+)}} = \cZ_\bullet^{(-)},
  \end{equation*}
  so that $\cW = \cZ_\bullet^{(+)} \oplus  \overline{\cZ_\bullet^{(+)}}$.


\subsection{Maximal uniformly positive/negative subspaces}

Let $(\cW,Q)$ be a Krein space.
We want to characterize definite subspaces with good properties.
Following \cite{Bognar} we make the following definition.
\begin{definition}\label{def:maxunipos}
  Let $\cZ$ be a subspace of $\cW$.
  \begin{enumerate}
  \item
    We say that it is {\em uniformly
positive/negative} if   for some  scalar product $(\cdot|\cdot)_\bullet$
  compatible with the Hilbertizable structure of $\cW$ there exists
  $\alpha_\bullet>0$ such that
  \begin{align}
  \label{eq:maxunipos...}
   v\in\cZ & \Rightarrow\cinner{v}{Qv}\geq \alpha_\bullet (v|v)_\bullet,
   \word{resp.}
   v\in\cZ  \Rightarrow
   \cinner{v}{Qv}\leq - \alpha_\bullet (v|v)_\bullet.
  \end{align}
\item We say that $\cZ$ is  {\em maximal
  uniformly positive/negative} if it is a maximal subspace with the
  property of uniform positivity/negativity.
 \end{enumerate}
\end{definition}

The following proposition, whose statement partially overlaps with
Thm.~V.5.2. and Cor. V.~7.4. in \cite{Bognar}, relates maximal
uniformly positive/negative spaces to fundamental decompositions
and admissible involutions.

\begin{prop}\label{kaku3}
  Let $\cZ_\bullet^{(+)}$ be a  subspace of $\cW$.
  Set $\cZ_\bullet^{(-)} := \cZ_\bullet^{(+)\perp Q}$.
  The following conditions are equivalent:
  \begin{enumerate}
  \item $\cZ_\bullet^{(+)}$ is maximal uniformly positive.
    \item $\cZ_\bullet^{(+)}$ is maximal uniformly positive and $\cZ_\bullet^{(-)}$ is
      maximal uniformly negative.
          \item The spaces $\cZ_\bullet^{(+)}$ and $\cZ_\bullet^{(-)}$ are complementary, and if $(\Pi_\bullet^{(+)},\Pi_\bullet^{(-)})$ is the corresponding pair of projections, then
          $    S_\bullet := \Pi_\bullet^{(+)}-\Pi_\bullet^{(-)}$
          is an admissible involution.
  \end{enumerate}
  \label{pos3}
\end{prop}

\begin{proof}[Proof of Prop.~\ref{pos3}.]
Assume 3). Then $(\cdot|\cdot)_\bullet:=(\cdot|QS_\bullet\cdot)$ is
compatible and
\beq(v|Qv)_\bullet=\pm(v|v)_\bullet,\quad v\in\cZ_\bullet^{(\pm)}.\eeq
Hence $\cZ_\bullet^{(\pm)}$ are maximal uniformly
positive/negative. This proves 3)$\Rightarrow$2).

2)$\Rightarrow$1) is obvious.

Now assume 1).
Let $S_0$ be an arbitrary admissible involution with
the corresponding scalar product $(\cdot|\cdot)_0$.
First note that $\cZ_\bullet^{(-)}$ is negative. Indeed, suppose that
$v_1\in\cZ^{(-)}_\bullet$ is strictly positive. Then for some $\alpha_1$
\begin{align}
(v_1|Qv_1)\geq\alpha_1(v_1|v_1)_0.
\end{align}
Hence $\Span(\cZ_\bullet^{(+)},v_1)$ is uniformly positive, which
contradicts the maximality
of $\cZ_\bullet^{(+)}$.

Let $P$ by the orthogonal projection (in the sense of $(\cdot|\cdot)_0$) onto
$\cZ_\bullet^{(+)}$. Then an arbitrary element of $\cZ_\bullet^{(+)}$ has the form
$Pv$ and of $\cZ_\bullet^{(-)}$ the form $S_0(\one-P)v$ for some $v\in\cW$.

By the uniform positivity of $\cZ_\bullet^{(+)}$, resp. by negativity of
$\cZ_\bullet^{(-)}$, we have
\begin{align}
  (v|PS_0Pv)_0&=(Pv|S_0Pv)_0=(Pv|QPv)\geq\alpha(Pv|Pv)_0
\end{align}
and
\begin{align}
  (v|(\one-P)S_0(\one-P)v)_0 &=  (S_0(\one-P)v|(\one-P)v)_0\notag\\
  &=(S_0(\one-P)v|QS_0(\one-P)v)\leq0 .
  \end{align}
 Lemma \ref{lemma} then implies the uniform negativity of
 $\cZ_\bullet^{(-)}$:
\begin{align}
  (v|(\one-P)S_0(\one-P)v)_0&\leq -\alpha (v|(\one-P)v)_0 \notag
  \\ &=
 -\alpha  (S_0(\one-P)v|S_0(\one-P)v)_0 .
 \end{align}

Clearly, $0\neq w\in\cZ_\bullet^{(+)}\cap \cZ_\bullet^{(-)}$ has to be
simultaneously positive and negative. Hence $\cZ_\bullet^{(+)}\cap \cZ_\bullet^{(-)}=\{0\}$.

As maximal positive/negative subspaces,
$\cZ_\bullet^{(+)}$ and $ \cZ_\bullet^{(-)}$ are automatically closed.

  Set
  \beq T:=S_0(\one-P)+PS_0.\eeq
For any $w\in\cW$,
  $PS_0w\in\cZ_\bullet^{(+)}$ and $S_0(\one-P)w\in
  \cZ_\bullet^{(-)}$. Hence the range of $T$ is contained in
  $\cZ_\bullet^{(+)}+\cZ_\bullet^{(-)}$.   By Lemma \ref{lemma}, $T$
  is invertible, hence the range of $T$ is $\cW$. Therefore,
  $\cW=\cZ_\bullet^{(+)}+\cZ_\bullet^{(-)}$.

We have proved that $\cZ_\bullet^{(+)}$ and
$\cZ_\bullet^{(-)}$ are complementary. Let $S_\bullet$ be the
corresponding involution. It is obviously bounded. Besides,
\beq
(v|v)_\bullet:= (v|QS_\bullet v)\geq\alpha(v|v)_0.\eeq
Hence $(\cdot|\cdot)_\bullet$ is compatible. This ends the proof of
1)$\Rightarrow$3).
\end{proof}

Here is another proposition about fundamental decompositions.
Note that it does not involve a reference to the topology of $\cW$,
but only to the form $Q$.

\begin{prop}
Let  $\cZ_\bullet^{(+)}$ and
$\cZ_\bullet^{(-)}$ be complementary  subspaces of a Krein space $(\cW,Q)$,
$Q$-orthogonal to one another. Assume  that
$\cZ_\bullet^{(\pm)}$ are positive resp. negative,
contain no neutral elements apart
from $0$ and are complete in the norm $\|v\|_{(\pm)}:=\sqrt{\pm(v|Qv)}$. Then $\cZ_\bullet^{(\pm)}$ is maximal
uniformly positive/negative and
  $\cZ_\bullet^{(-)} := \cZ_\bullet^{(+)\perp Q}$, so that we are
  precisely in the setting described by  Prop.
\ref{kaku3}.\end{prop}

\begin{proof}
Let $S_\bullet$ be the involution defined by
$\cW= \cZ_\bullet^{(+)}\oplus
\cZ_\bullet^{(-)}$.
As usual, we introduce  the corresponding scalar product $(v|w)_\bullet:=(v|QS_\bullet
w)$ and the norm $\|\cdot\|_\bullet$.
Note that $\|v\|_\bullet = \| v\|_{(\pm)}$ if $v\in
\cZ_\bullet^{(\pm)}$.

Let $\|\cdot\|_1$ be any compatible  norm. Clearly, by the
boundedness of $Q$, we have
\beq \|v\|_\bullet\leq C\|v\|_1.\eeq

Consider the identity operator from $\cW$ with $\|\cdot\|_\bullet$
to
$\cW$ with $\|\cdot\|_1$. In both norms $\cW$ is  complete. Then
the identity is bounded. Hence it is closed. The operator is
bijective. Hence by Banach's theorem its inverse is bounded. Therefore
we have
\beq
\|v\|_1\leq c\|v\|_\bullet.\eeq
 Thus, $\cZ_\bullet^{(\pm)}$ are uniformly positive
resp. negative.
\end{proof}

\begin{prop}
  Let $S_1$, $S_2$ be a pair of admissible involutions.
Define $K, c, d$ as in \eqref{kaku} and \eqref{kaku1}.
Then
  $K$ is pseudo-unitary on $(\cW,Q)$ and
  $K$ is positive  with respect to both $(\cdot|\cdot)_1$ and
  $(\cdot|\cdot)_2$.
Besides, $\|c\|<1$  and $c^*=d$ with respect to $(\cdot|\cdot)_1$.
\label{kaku2}\end{prop}
\begin{proof}
  $K$ is pseudo-unitary as the product of two pseudo-unitary transformations.
  The inequality
  \begin{align*}
    (v|K v)_1 & = (S_1 v|Q S_2 S_1 v) = (S_1 v|S_1 v)_2 \geq a (S_1 v|S_1 v)_1 = a (v|v)_1
  \end{align*}
  with $a>0$ shows the positivity of $K$ with respect to
  $(\cdot|\cdot)_1$. Therefore, $\one+K$ is  invertible and $\|\frac{\one-K}{\one+K}\|<1$.
  Hence $\|c\|<1$.
\end{proof}

We finally show that any pair consisting of a maximal uniformly
positive and a maximal uniformly negative subspace is complementary.
(See also Lem. V.7.6. in \cite{Bognar}).
\begin{prop} Suppose that $\cZ_1^{(+)}$ is a maximal uniformly positive space
  and $\cZ_2^{(-)}$ is a maximal uniformly negative space. Then they are
  complementary.
\end{prop}
\begin{proof}
Set $\cZ_1^{(-)}:=\cZ_1^{(+)\perp Q}$ and
$\cZ_2^{(+)}:=\cZ_2^{(-)\perp Q}$.
Let  $S_1$ resp. $S_2$ be the involutions corresponding
to the pairs of complementary subspaces
$(\cZ_1^{(+)},\cZ_1^{(-)})$, resp.\ $(\cZ_2^{(+)},\cZ_2^{(-)})$. They are admissible.
By Prop. \ref{kaku2}, $K=S_2S_1$ is positive. Hence $\one+K$ is
invertible. Thus the result follows from Prop. \ref{prop:1+K_etc}.
\end{proof}

\section{Gegenbauer equation}
\label{app:gegenbauer}
For the convenience of the reader, we present in this appendix
basic statements about Gegenbauer functions needed
 in Sections
 \ref{ssc:dS} and \ref{sec:AdS}.
More details on Gegenbauer functions can be found e.g. in \cite{DGR23a},
on which this section is based.

Here is the {\em Gegenbauer equation}:
\begin{align}
\Bigg((1-w^2)\partial_w^2-2(1+\alpha )w\partial_w
+\lambda ^2-\Big(\alpha +\frac{1}{2}\Big)^2\Bigg)f(w)=0.\label{gege0}
\end{align}
We will express its solutions  in terms of the {\em Olver normalized  Gauss hypergeometric
function}:
\begin{align}
\label{eq:hyper_series}
{\bf F}(a,b;c;z):=\frac{F(a,b;c;z)}{\Gamma(c)}
=\sum_{n=0}^\infty\frac{(a)_n(b)_nz^n}{\Gamma(c+n)n!}.
\end{align}
The defining series converges only in the unit disc, but
${\bf F}(a,b;c;z)$ extends to a holomorphic function on
$\bC \backslash [1 , \infty[$ as well as on a universal
cover of $\bC \backslash \{ 0, 1 \}$.

In what follows complex power functions should be interpreted
as their principal branches (holomorphic on
$\bC \setminus ]- \infty , 0]$). For example $w \mapsto (1-w)^\alpha$
is holomorphic away from $[1,\infty[$. In addition, we will frequently
use the notation
\begin{equation}
 (w^2-1)^\alpha_\bullet := (w-1)^\alpha(w+1)^\alpha.
\end{equation}
The function $(w^2-1)^\alpha_\bullet$ is holomorphic on
$\bC\backslash]-\infty,1]$, whereas $ (w^2-1)^\alpha$ is holomorphic
on $\bC\backslash\big([-1,1]\cup\ii\bR\big)$. One has
$(w^2-1)^\alpha_\bullet = (w^2-1)^\alpha$ only for
$\re(w)>0$. However, $(1-w^2)^\alpha = (1-w)^\alpha (1+w)^\alpha$ for
all $w\not\in]-\infty,-1]\cup[1,\infty[$.

Standard solutions of the Gegenbauer equations are characterized by
simple behavior at one of the three singular points $1, -1 , \infty$.
Due to the $w\mapsto -w$ symmetry of the equation \eqref{gege0},
solutions of the second type are obtained from solutions of the first
type by negating the argument. Therefore we consider 4 functions,
corresponding to 2 behaviors at $1$ and 2 behaviors at $\infty$.
All of them are holomorphic on $\bC \backslash ]- \infty, 1 ]$.

\begin{itemize}
\item The solution characterized by asymptotics $\sim 1$ at $1$:
\begin{alignat}{2}\label{solu1}
    S_{\alpha ,\pm \lambda }(w)&
:= F\Big(\frac12+\alpha+\lambda,\frac12+\alpha-\lambda;\alpha+1;
\frac{1-w}{2}\Big)\\
 &= \left( \frac{2}{w+1} \right)^\alpha
 F\Big(\frac12+\lambda,\frac12-\lambda;\alpha+1;\frac{1-w}{2}\Big).
 \label{solu1_form2}
\end{alignat}
$S_{\alpha, \lambda}$ is distinguished among the four solutions
introduced here by the fact that it is holomorphic on
$\bC \backslash ] - \infty, -1 ] $ rather than only on
$\bC \backslash ]- \infty, 1 ]$.
\item The solution $\frac{2^{2\alpha}}{(w^2-1)^{\alpha}_\bullet}
S_{-\alpha ,\lambda }(w)$ is characterized by asymptotics
$\sim \frac{2^\alpha}{(w-1)^{\alpha}}$ at $1$.
\item The solution characterized by asymptotics
$\sim w^{-\frac12-\alpha-\lambda}$ at $+\infty$:
\begin{alignat}{2}\label{solu3}
 Z_{\alpha ,\lambda }(w)
 &=
 ( w \pm 1)^{-\frac12-\alpha -\lambda } F
\Big(
\frac12+\lambda,\frac12+\lambda+\alpha;1+2\lambda;\frac{2}{1\pm w}\Big)\\\notag
&=
w^{-\frac12-\alpha -\lambda } F\Big(\frac14+\frac\alpha2+\frac\lambda2,
\frac34+\frac\alpha2+\frac\lambda2;1+\lambda;
\frac1{w^2}\Big).
\end{alignat}
\item The solution $Z_{\alpha, - \lambda}(w)$ is characterized by asymptotics
$\sim w^{-\frac12-\alpha+\lambda}$ at $+\infty$.

\end{itemize}

All these 4 functions can be expressed in terms of $
S_{\alpha,\lambda}$, but for typographical reasons it is
convenient to introduce also $ Z_{\alpha,\lambda}$.
We will use Olver's normalization:
\begin{align}
{\bf
  S}_{\alpha,\lambda}(w):=\frac{1}{\Gamma(\alpha+1)}S_{\alpha,\lambda}(w),
  \qquad
{\bf
  Z}_{\alpha,\lambda}(w):=\frac{1}{\Gamma(\lambda+1)}Z_{\alpha,\lambda}(w).
\end{align}
We note the identities
\begin{align}
\label{eq:identities_SZ_signs}
  {\bf S}_{\alpha,\lambda}(w)={\bf S}_{\alpha,-\lambda}(w),&\quad
     {\bf Z}_{\alpha,\lambda}(w)=\frac{ {\bf Z}_{-\alpha ,\lambda }(w)
     }{(w^2-1)^\alpha_\bullet}.
  \end{align}

Here are the connection formulas:
\begin{align}\label{formu2}
  {\bf S}_{\alpha,\lambda}(-w)
=&-\frac{\cos(\pi\lambda)}{\sin(\pi\alpha)}
{\bf S}_{\alpha,\lambda}(w)
+\frac{2^{2\alpha}\pi {\bf S}_{-\alpha,-\lambda}(w)}{\sin(\pi\alpha)
\Gamma(\frac12+\alpha+\lambda)
\Gamma(\frac12+\alpha-\lambda) (1-w^2)^{\alpha}},
\\ \label{formu1}
{\bf Z}_{\alpha,\lambda}(w)
=&-\frac{2^{\lambda-\alpha-\frac12}\sqrt{\pi}
{\bf S}_{\alpha,\lambda}(w) }{\sin(\pi\alpha)
\Gamma(\frac12-\alpha+\lambda)}
+\frac{2^{\lambda+\alpha-\frac12}\sqrt{\pi}}{\sin(\pi\alpha)
\Gamma(\frac12+\alpha+\lambda)}
\frac {{\bf S}_{-\alpha,-\lambda}(w)}{(w^2-1)^{\alpha}_\bullet},
\\ \label{formu3}
{\bf S}_{\alpha,\lambda}(w)
=&  \frac{2^{-\lambda+\alpha-\tfrac12}\sqrt\pi}{\sin\pi\lambda}
\Bigg( -\frac{{\bf Z}_{\alpha,\lambda}(w)
}{\Gamma\big( \tfrac12+\alpha-\lambda \big)}
+ \frac{2^{2\lambda}{\bf Z}_{\alpha,-\lambda}(w)
}{\Gamma\big(\tfrac12+\alpha+\lambda\big)}
\Bigg).
\end{align}
From its definition, it is clear that ${\bf Z}_{\alpha,\lambda}$ satisfies
\begin{align}
\label{eq:Z+-}
{\bf Z}_{\alpha,\lambda} (-w\mp\ii0)
= \ee^{\pm\ii\pi\big(\tfrac12+\alpha+\lambda\big)}
{\bf Z}_{\alpha,\lambda} (w\pm\ii0),\quad
w \in \bR.
\end{align}

For further information on Gegenbauer functions (in various conventions), consult
for example \cite{DGR23a,NIST,GR,WW,Durand76}.


\paragraph{Acknowledgement.}
The work of J.D. and C.G.
 was supported by National Science Center of Poland under the
    grant UMO-2019/35/B/ST1/01651. We would like to
    thank Markus Fröb, Christian Gérard,
    Wojciech Kamiński, Nicola Pinamonti and Karl-Henning Rehren for enlightening discussions.

\paragraph{Data availability statement.}
All relevant data is contained in the manuscript.

\paragraph{Conflicts of interest.} The authors declare no
conflict of interest.


\footnotesize

\begin{thebibliography}{10}
\bibitem{AAK21}
E.~T.~Akhmedov, A.~A.~Artemev and I.~V.~Kochergin.
\newblock{Interacting quantum fields in various charts of anti–de Sitter spacetime}.
\newblock{\em Phy. Rev. D}, \textbf{103}:045009 , 2021.

\bibitem{ABDMPS19}
E.~T. Akhmedov, K.~V. Bazarov, D.~V. Diakonov, U.~Moschella, F.~K. Popov, and
  C.~Schubert.
\newblock {Propagators and Gaussian effective actions in various patches of de
  Sitter space}.
\newblock {\em Phys. Rev. D}, \textbf{100}:105011, 2019.

\bibitem{AMP18}
E.~T.~Akhmedov, U.~Moschella and F.~K.~Popov.
\newblock{Ultraviolet phenomena in AdS self-interacting quantum field theory}.
\newblock{\em JHEP}, \textbf{2018}:183 , 2018.

\bibitem{Allen85}
B.~Allen.
\newblock {Vacuum states in de Sitter space}.
\newblock {\em Phys. Rev. D} \textbf{32}:1, 1985.

\bibitem{ABG}
W. Amrein, A. Boutet de Monvel and V. Georgescu.
  \newblock{$C_0$-Groups, Commutator Methods and Spectral Theory of
    N-Body Hamiltonians}. \newblock{Springer 1996}.

\bibitem{AvisIS78}
S.~J.~Avis, C.~J.~Isham and D.~Storey.
\newblock{Quantum field theory in anti-de Sitter space-time}.
\newblock{\em Phys. Rev. D}, \textbf{18}:3565, 1978.

\bibitem{AzizovIokhvidov}
T.~Ya.~Azizov and I.~S.~Iokhvidov.
\newblock{Linear Operators in Spaces with Indefinite Metric}.
\newblock{\em  Journal of Soviet Mathematics}, \textbf{15}:438--490, 1981.

\bibitem{BF}
  C. B\"ar and K. Fredenhagen. Quantum Field Theory on Curved
  Spacetimes, Lect. Notes in Phys. 786, Springer 2009


\bibitem{bar}
  C.~Bär, N.~Ginoux, F.~Pfäffle.
  \newblock{Wave Equations on Lorentzian Manifolds and Quantization}.
  \newblock ESI Lectures in Mathematical Physics, European Mathematical Society,
    2007.

\bibitem{vdBan84}
E.~P.~van den Ban.
\newblock{Invariant differential operators on a semisimple symmetric space and finite multiplicities in a Plancherel formula}.
\newblock{ \em Ark. Mat.} \textbf{25}:175--187, 1987.

\bibitem{BN} R.~Banerjee and M.~Niedermaier.
Analytic semigroups approaching a Schrödinger group on real foliated metric manifolds.
\newblock{\em Journal of Functional Analysis}, \textbf{289}:110898, 2025.

\bibitem{BirrellDavies}
N.~Birrell and P.~Davies.
\newblock {\em {Quantum Fields in Curved Space}}.
\newblock Cambridge Monographs on Mathematical Physics, 1982.

  \bibitem{bjorken}
  J.~D.~Björken and S.~D.~Drell.
  \newblock{\em Relativistic Quantum Fields}.
  \newblock McGraw-Hill, 1965.

\bibitem{BD69}
G.~Börner and H.~P. Dürr.
\newblock {Classical and Quantum Fields in de Sitter Space}.
\newblock {\em Il Nuovo Cimento}, \textbf{19}:3, 1969.

\bibitem{Bognar}
J.~Bognár.
\newblock{\em Indefinite Inner Product Spaces}.
\newblock{Springer Berlin Heidelberg New York}, 1974.

  \bibitem{bogoliubov}
  N.~N.~Bogoliubov and D.~V.~Shirkov.
  \newblock{\em Introduction to the Theory of Quantized Fields}, 3 edn.
  \newblock John Wiley \& Sons, 1980.

\bibitem{BMS02}
R.~Bousso, A.~Maloney, and A.~Strominger.
\newblock {Conformal vacua and entropy in de Sitter space}.
\newblock {\em Phys. Rev. D}, \textbf{65}:104039, 2002.

\bibitem{BEGMP12}
J.~Bros, H.~Epstein, M.~Gaudin, U.~Moschella, and V.~Pasquier.
\newblock {Anti-de Sitter Quantum Field Theory and a New Class of
  Hypergeometric Identities}.
\newblock {\em Commun. Math. Phys.}, \textbf{309}:255--291, 2012.

\bibitem{BEM98}
J.~Bros, H.~Epstein, and U.~Moschella.
\newblock {Analyticity Properties and Thermal Effects for General Quantum Field
  Theory on de Sitter Space-Time}.
\newblock {\em Commun. Math. Phys.}, \textbf{196}:535--570, 1998.

\bibitem{BEM02}
J.~Bros, H.~Epstein, and U.~Moschella.
\newblock {Towards a General Theory of Quantized Fields on the Anti-de Sitter
  Space-Time}.
\newblock {\em Commun. Math. Phys.}, \textbf{231}:481–528, 2002.

\bibitem{BGM94}
J.~Bros, J.-P. Gazeau, and U.~Moschella.
\newblock {Quantum Field Theory in the de Sitter Universe}.
\newblock {\em Phys. Rev. Lett.}, \textbf{73}:13, 1994.

\bibitem{BM96}
J.~Bros and U.~Moschella.
\newblock {Two-point Functions and Quantum Fields in de Sitter Universe}.
\newblock {\em Rev. Math. Phys.}, \textbf{08}(03):327--391, 1996.

  \bibitem{brunetti-fredenhagen}
  R.~Brunetti and K.~Fredenhagen.
  \newblock{Microlocal analysis and interacting quantum
    field theories: {Renormalization} on physical backgrounds}.
  \newblock{\em Commun. Math. Phys.}, \textbf{208}(3):623--661,
    2000.

\bibitem{BFH05}
R.~Brunetti, K.~Fredenhagen and S.~Hollands.
\newblock{A remark on alpha vacua for quantum field theories on de Sitter space}
\newblock{\em JHEP}, \textbf{05} 063, 2005.

\bibitem{BunchDavies78}
T.~S. Bunch and P.~C.~W. Davies.
\newblock {Quantum field theory in de Sitter space: renormalization by
  point-splitting}.
\newblock {\em Proceedings of the Royal Society A}, \textbf{360}:117, 1978.

\bibitem{BurgessL85}
C.~P.~Burgess and C.~A.~Lütken.
\newblock{Propagators and effective potentials in anti-de Sitter space}.
\newblock{\em Phys. Lett. B},\textbf{153}:137--1441 , 1985.

\bibitem{CDD20}
M.~Capoferri, C.~Dappiaggi and N.~Drago.
\newblock{Global wave parametrices on globally hyperbolic spacetimes}.
\newblock{\em J. Math. Anal. Appl.}, \textbf{490}:124316, 2020.

\bibitem{CV22}
M.~Capoferri and D.~Vassiliev.
\newblock{Global Propagator for the Massless Dirac
Operator and Spectral Asymptotics}.
\newblock{\em Integr. Equ. Oper. Theory }, \textbf{94}:30, 2022.

\bibitem{CT68}
N.~A. Chernikov and E.~A. Tagirov.
\newblock {Quantum theory of scalar field in de Sitter space-time}.
\newblock {\em Annales de l’I. H. P., section A}, \textbf{9}:109--141, 1968.

\bibitem{choquet-bruhat}
 Y.~Choquet-Bruhat.
 \newblock{Hyperbolic partial differential equations on a manifold}.
 \newblock In: Battelle Rencontres, 1967 Lectures Math. Phys., 1968.

\bibitem{CDT18}
H.~S. Cohl, T.~H. Dang, and T.~M. Dunster.
\newblock {Fundamental Solutions and Gegenbauer Expansions of Helmholtz
  Operators in Riemannian Spaces of Constant Curvature}.
\newblock {\em SIGMA}, \textbf{14}:136, 2018.

\bibitem{CritchleyPhD}
R.~Critchley.
\newblock {Field Theory in Curved Spaces (PhD thesis)}.
\newblock University of Manchester, 1976.

\bibitem{CFKS}
H. Cycon, R. Froese, W. Kirsch and B.Simon.
\newblock{Schrödinger Operators.
With Application to Quantum Mechanics and Global Geometry}.
\newblock{Springer 1987}.

\bibitem{DF16}
C.~Dappiaggi and H.~R.~C. Ferreira.
\newblock {Hadamard states for a scalar field in anti–de Sitter spacetime with arbitrary boundary conditions}.
\newblock {\em Phys. Rev. D}, \textbf{94}:125016, 2016.

\bibitem{DFJ18}
C.~Dappiaggi, H.~R.~C. Ferreira, and B.~A. Juárez-Aubry.
\newblock {Mode solutions for a Klein-Gordon field in anti–de Sitter
  spacetime with dynamical boundary conditions of Wentzell type}.
\newblock {\em Phys. Rev. D}, \textbf{97}:085022, 2018.

\bibitem{DGR23a}
J.~Dereziński, C.~Gaß, and B.~Ruba.
\newblock {Generalized integrals of Macdonald and Gegenbauer functions}.
\newblock In: Applications and $q$-Extensions of Hypergeometric Functions.
        {\em Contemporary Mathrematics} \textbf{819}, 2025.

\bibitem{DGR23b}
J.~Dereziński, C.~Gaß, and B.~Ruba.
\newblock {Point potentials  on  Euclidean space, hyperbolic space
  and sphere
in any dimension}.
\newblock {\em Ann. Henri Poincaré}, 2024.

\bibitem{DeGeorgescu}
J.~Dereziński and V.~Georgescu.
\newblock{One-Dimensional Schrödinger Operators with Complex Potentials}.
\newblock{\em Annales Henri Poincaré}, \textbf{21}:1947-2008 , 2020.

\bibitem{DS18}
J.~Dereziński and D.~Siemssen.
\newblock {Feynman propagators on static spacetimes}.
\newblock {\em Rev. Math. Phys.}, \textbf{30}:03, 2018.

\bibitem{DS19}
J.~Dereziński and D.~Siemssen.
\newblock {An evolution equation approach to the Klein–Gordon operator on curved spacetime}.
\newblock {\em Pure and Applied Analysis}, \textbf{1}:2, 2019.

\bibitem{DS22}
J.~Dereziński and D.~Siemssen.
\newblock {An Evolution Equation Approach to
Linear Quantum Field Theory}. In: Correggi, M., Falconi, M. (eds) Quantum Mathematics II. INdAM 2022. Springer INdAM Series, vol 58. Springer, Singapore.

\bibitem{DeSi}
J.~Dereziński and B. Sikorski.
\newblock {Bessel potentials and Green functions on pseudo-Euclidean spaces}.
\newblock{ \em Preprint}, \href{https://arxiv.org/abs/2406.08327}{
  arXiv:2406.08327}, to appear in Rep. Math. Phys.

\bibitem{DeWroch}
J.~Dereziński and M.~Wrochna.
\newblock{Exactly Solvable Schrödinger Operators}.
\newblock{\em Annales Henri Poincaré}, \textbf{12}:397--418 , 2011.

\bibitem{DuHor}
J.~J.~ Duistermaat and L. H\"ormander. \newblock{Fourier integral
    operators. II}. \newblock{\em Acta Math.} \textbf{128}, 183–269
(1972).


\bibitem{DullemontvB85}
C.~Dullemond and E.~van Beveren.
\newblock{Scalar field propagators in anti‐de Sitter space‐time}
 \newblock{\em J. Math. Phys. }, \textbf{26}:2050--2058, 1985.

\bibitem{Durand76}
L.~Durand.
\newblock {Expansion formulas and addition theorems for Gegenbauer functions}.
\newblock {\em J. Math. Phys.}, \textbf{17}:1933--1948, 1976.

\bibitem{FH14}
M.~B. Fröb and A.~Higuchi.
\newblock {Mode-sum construction of the two-point functions for the
  Stueckelberg vector fields in the Poincaré patch of de Sitter space}.
\newblock {\em J. Math. Phys.}, \textbf{55}:062301, 2014.

\bibitem{FSS13}
M.~Fukuma, S.~Sugishita, and Y.~Sakatani.
\newblock {Propagators in de Sitter space}.
\newblock {\em Phys. Rev. D}, \textbf{88}:024041, 2013.

\bibitem{GS68}
J.~Geheniau and C.~Schomblond.
\newblock {Fonctions de Green dans l'Univers de de Sitter}.
\newblock {\em Acad. R. Belg. Bull. Cl. Sci.}, \textbf{54}:1147, 1968.

\bibitem{Gerard}
C.~Gérard.
\newblock{\em Microlocal Analysis of Quantum Fields on Curved Spacetimes}.
\newblock ESI Lectures in Mathematics and Physics. European
Mathematical Society, Berlin, 2019.


\bibitem{GGH13}
V.~Georgescu, C.~G\'erard and D.~H\"afner.
Boundary values of resolvents of selfadjoint operators in Krein spaces.
\newblock {\em Journal of Functional Analysis}, \textbf{265}:3245-3304, 2013.


\bibitem{GGH15} V.Georgescu, C.G\'erard and D.H\"afner.
Resolvent and propagation estimates for Klein-Gordon equations with non-positive energy.
\newblock{\em J. Spectr. Theory}, \textbf{5}:113-192, 2015.

\bibitem{GW17}
C.~Gérard and M.~Wrochna.
\newblock{Hadamard Property of the in and out States for Klein–Gordon Fields on Asymptotically Static Spacetimes}.
\newblock {\em Annales Henri Poincaré}, \textbf{18}:2715--2756, 2017.

\bibitem{GW19}
C.~Gérard and M.~Wrochna.
\newblock{The massive Feynman propagator on asymptotically Minkowski spacetimes}.
\newblock {\em American Journal of Mathematics}, \textbf{141}:1501--1546, 2019.


\bibitem{GR}
I.S. Gradshteyn and I.M. Ryzhik.
\newblock {\em {Table of integrals, series, and products}}.
\newblock Academic Press, translated by Scripta Technica, Inc., 7th edition,
  2007.


\bibitem{hadamard}
J.~Hadamard. \newblock{ \em Lectures on Cauchy's Problem in Linear Partial Differential Equations}.
\newblock{Yale University Press}, 1923.

\bibitem{HawkingEllis}
S.~W. Hawking and G.~F.~R. Ellis.
\newblock {\em {The large scale structure of space-time}}.
\newblock Cambridge University Press, 1973.

\bibitem{HMM11}
A.~Higuchi, D.~Marolf, and I.~A. Morrison.
\newblock {Equivalence between Euclidean and in-in formalisms in de Sitter
  QFT}.
\newblock {\em Phys. Rev. D}, \textbf{83}:084029, 2011.

\bibitem{Hollands12}
S.~Hollands.
\newblock {Massless Interacting Scalar Quantum Fields in de Sitter Spacetime}.
\newblock {\em Ann. Henri Poincaré}, \textbf{13}:1039–1081, 2012.

\bibitem{Hollands13}
S.~Hollands.
\newblock {Correlators, Feynman Diagrams, and Quantum No-hair in de Sitter
  Spacetime}.
\newblock {\em Commun. Math. Phys.}, \textbf{319}:1--68, 2013.

  \bibitem{hollands-wald1}
  S.~Hollands and R.~M.~Wald, R.M..
  \newblock{Local {Wick} polynomials and time ordered products of
    quantum fields in curved spacetime}.
  \newblock {\em Commun. Math. Phys.}, \textbf{223}(2):289--326,
    2001.

\bibitem{IW04}
A.~Ishibashi and R.~M. Wald.
\newblock {Dynamics in non-globally-hyperbolic static spacetimes: III. Anti-de Sitter spacetime}.
\newblock {\em Class. Quant. Grav.}, \textbf{21}:2981, 2004.

  \bibitem{kaminski}
 W.~Kamiński.
  \newblock{Non-Self-Adjointness of the Klein-Gordon Operator on a Globally   Hyperbolic and Geodesically Complete Manifold: An Example}.
  \newblock{\em Annales Henri Poincaré}, \textbf{23}:4409--4427, 2022.

\bibitem{kato:perturbation}
  T.~Kato.
  \newblock{\em Perturbation Theory of Linear Operators}.
  \newblock Springer Berlin-Heidelberg-New-York, 1980.

  \bibitem{leray}
  J.~Leray.
  \newblock{Hyperbolic Differential Equations}.
  \newblock Unpublished lecture notes. The Institute for Advanced Study,
    Princeton, N.J., 1953.

\bibitem{Lokas}
E.L. {\L}okas.
\newblock {Positive frequency solutions of the Klein-Gordon equation in the
  n-dimensional de Sitter space-time}.
\newblock {\em Acta Physica Polonica B}, \textbf{26}, 1995.

\bibitem{Moschella06}
U.~Moschella.
\newblock {The de Sitter and anti-de Sitter Sightseeing Tour}.
\newblock {\em Proceedings of the 7th Seminaire Poincaré: Einstein 1905-2005 :
  Paris, France}, pages 120--133, 2005.

\bibitem{M85}
E.~Mottola.
\newblock {Particle Creation In De Sitter Space}.
\newblock {\em Phys. Rev. D}, \textbf{31}:754, 1985.

\bibitem{Naj79}
B. Najman.
Solution of a differential equation in a scale of spaces.
\newblock {\em Glas. Mat.} \textbf{14}:119-127, 1979.

\bibitem{Naj80}
 B. Najman.
 Spectral properties of the operators of Klein–Gordon type.
 \newblock {\em Glas. Mat.}, \textbf{15}:97–112, 1980.




\bibitem{NT23}
S.~Nakamura and K.~Taira.
\newblock {Essential Self-Adjointness of Klein-Gordon Type Operators on Asymptotically Static, Cauchy-Compact Spacetimes}.
\newblock {\em Commun. Math. Phys.}, \textbf{398}:1153--1169, 2023.

\bibitem{NT23_2}
S.~Nakamura and K.~Taira.
\newblock{A Remark on the Essential Self-adjointness for Klein–Gordon-Type Operators}.
\newblock{Annales Henri Poincaré}, \textbf{24}:2587--2605, 2023.


\bibitem{NIST}
F.~W.~J. Olver, A.~B. {Olde Daalhuis}, D.~W. Lozier, B.~I. Schneider, R.~F.
  Boisvert, C.~W. Clark, B.~R. Miller, B.~V. Saunders, H.~S. Cohl, and M.~A.
  McClain.
\newblock {NIST Digital Library of Mathematical Functions}.
\newblock available at \href{http://dlmf.nist.gov/}{http://dlmf.nist.gov/}, 06
  2022.

\bibitem{PT} L. Parker and D. Toms. Quantum Field Theory in Curved
  Spacetime, Cambridge University Press 2009



\bibitem{PT33}
 G. Pöschl and E. Teller.
 \newblock{Bemerkungen zur Quantenmechanik des anharmonischen Osziltators}.
 \newblock {\em Zeitschrift für Physik},
 \textbf{83}, 143--151, 1933.

\bibitem{RSII} M.~Reed B.~Simon.
 \emph{Methods of Modern Mathematical Physics, II. Fourier Analysis, Self-Adjointness},
 Academic Press, London, 1975.

 \bibitem{Rossmann78}
W.~Rossmann.
\newblock{Analysis on real hyperbolic spaces}.
 \newblock{\em  J. Funct. Anal.}, \textbf{30}:448--477, 1978.

 \bibitem{RumpfdS}
 H.~Rumpf.
 \newblock{Self-adjointness-based quantum field theory in de Sitter and anti-de Sitter space-time}.
 \newblock{\em Phys. Rev. D}, \textbf{24}:275, 1981.

  \bibitem{rumpf1}
  H.~Rumpf and H.~K.~Urbantke.
  \newblock{  Covariant ``in--out'' formalism for creation by external fields}.
  \newblock {\em Annals of Physics}, \textbf{114}, 332--355, 1978.

\bibitem{SS76}
C.~Schomblond and P.~Spindel.
\newblock {Conditions d’unicité pour le propagateur $\Delta_1(x,y)$ du champ
  scalaire dans l’univers de de Sitter}.
\newblock {\em Annales de l’I. H. P., section A}, \textbf{25} 1:67--78, 1976.

\bibitem{SSV01}
M.~Spradlin, A.~Strominger, and A.~Volovich.
\newblock {De Sitter Space}.
\newblock {\em Unity from Duality: Gravity, Gauge Theory and Strings: Les
  Houches Session LXXVI, 30 July--31 August 2001, \textup{Springer Berlin
  Heidelberg}}, pages 423--453, 2002.

\bibitem{Szmytkowski07}
R.~Szmytkowski.
\newblock {Closed forms of the Green's function and the generalized Green's
  function for the Helmholtz operator on the N-dimensional unit sphere}.
\newblock {\em J. Phys. A: Math. Theor.}, \textbf{40}:995--1009, 2007.

\bibitem{V20}
A.~Vasy.
\newblock {Essential self-adjointness of the wave operator and the limiting absorption principle on Lorentzian scattering spaces}.
\newblock {\em J. Spectr. Theory}, \textbf{10}:2,439--461, 2020.

\bibitem{Ves70}
  K.~Veseli\'c.
  \newblock{A spectral theory for the Klein–Gordon equation with an external electrostatic potential}.
  \newblock {\em Nucl. Phys. B}, \textbf{147}:215-224, 1970.

\bibitem{Ves90}
  K.~Veseli\'c.
  \newblock{A spectral theory of the Klein–Gordon equation involving a homogeneous electric field}. \newblock{\em J. Operator Theory}, \textbf{25}:319-330, 1991.

\bibitem{WW}
E.~T. Whittaker and G.~N. Watson.
\newblock {\em {A Course of Modern Analysis: An Introduction to the General
  Theory of Infinite Processes and of Analytic Functions; with an Account of
  the Principal Transcendental Functions}}.
\newblock Cambridge Mathematical Library Series, 1996.

\bibitem{Y10}
D.~R.~Yafaev.
\newblock {\em {Mathematical Scattering Theory: Analytic Theory}}.
\newblock Mathematical Surveys and Monographs, American Mathematical
Society, 2010.
\end{thebibliography}
\end{document}